\documentclass[letterpaper, 12pt, oneside]{book}

    \usepackage[left=108pt, right=74pt, top=108pt, bottom=81pt, dvips, pdftex]{geometry}
    
    \usepackage{fancyhdr} 
    \renewcommand{\bibname}{References}

    \usepackage{amsmath,amsthm, amsfonts,amssymb} 
    \usepackage{setspace} 
    \usepackage[linktocpage,bookmarksopen,bookmarksnumbered,
                pdftitle={UC Davis Ph.D. Doctoral Thesis},
                pdfauthor={Department of Mathematics},%
                pdfsubject={UC Davis Ph.D. Doctoral Thesis},%
                pdfkeywords={UC Davis Ph.D. Doctoral Thesis}]{hyperref}
    
\numberwithin{equation}{section}
 \newtheorem{theorem}{Theorem}

\newtheorem{corollary}{Corollary}

\newtheorem{definition}{Definition}
\newtheorem{example}{Example}

\newtheorem{lemma}{Lemma}

\newtheorem{remark}{Remark}

\newtheorem{assumption}{Assumption}

\newcommand{\R}{{\mathbb R}}
\newcommand{\bR}{{\mathbb R}}
\newcommand{\C}{{\mathbb C}}
\newcommand{\bC}{{\mathbb C}}
\newcommand{\Cx}{{\mathbb C}}

\newcommand{\bN}{{\mathbb N}}
\newcommand{\bZ}{{\mathbb Z}}
\newcommand{\Ir}{{\mathbb Z}}
\newcommand{\A}{{\mathcal A}}
\newcommand{\cA}{{\mathcal A}}
\newcommand{\cB}{{\mathcal B}}
\newcommand{\cH}{{\mathcal H}}
\newcommand{\cF}{{\mathcal F}}

\newcommand{\cU}{{\mathcal U}}
\newcommand{\cL}{{\mathcal L}}
\newcommand{\cK}{{\mathcal K}}

\renewcommand{\Im}{\,\mathrm{Im}\,}   
\def\idty{{\mathchoice {\mathrm{1\mskip-4mu l}} {\mathrm{1\mskip-4mu l}} %
{\mathrm{1\mskip-4.5mu l}} {\mathrm{1\mskip-5mu l}}}}

\newcommand{\Tr}{\mbox{Tr}}

\newcommand{\id}{\mathop{\rm id}}

\newcommand{\boxendproof}{\hspace*{\fill}{{$\Box$}} \vspace{10pt}}

\newcommand{\be}{\begin{equation}}
\newcommand{\ee}{\end{equation}}
\newcommand{\bea}{\begin{eqnarray}}
\newcommand{\eea}{\end{eqnarray}}
\newcommand{\beann}{\begin{eqnarray*}}
\newcommand{\eeann}{\end{eqnarray*}}
\newcommand{\eq}[1]{(\ref{#1})}

\newcommand{\tb}{{\tilde{b}}}

\newcommand{\e}{{\epsilon}}
\newcommand{\lb}{{\lambda}}
\newcommand{\wlim}{\underset{t\rightarrow\infty}{\operatorname{w^*-\, lim}}}

\usepackage{color}

\newcommand{\red}{\textcolor{red}}

    \numberwithin{equation}{section}
    \numberwithin{figure}{section}

    \newcommand{\n}{\vspace{12pt}} 
      
    \newcommand{\newchapter}[3] 
	{                           
        \chapter[#2]{#3}
        \chaptermark{#1}
        \thispagestyle{myheadings}
	}

\begin{document}
    

    \pagenumbering{roman}
    \pagestyle{plain}

    %
    %

    \singlespacing

    ~\vspace{-1.2in} 
    \begin{center}

        \begin{huge}
            Existence of the thermodynamic limit and asymptotic behavior of some irreversible quantum dynamical systems
        \end{huge}\\\n
        By\\\n
        {\sc Anna Vershynina}\\
        B.S. (Kharkov National University, Kharkov, Ukraine) 2007\\
        M.A. (University of California, Davis) 2012\\\n
        DISSERTATION\\\n
        Submitted in partial satisfaction of the requirements for the degree of\\\n
        DOCTOR OF PHILOSOPHY\\\n
        in\\\n
        MATHEMATICS\\\n
        in the\\\n
        OFFICE OF GRADUATE STUDIES\\\n
        of the\\\n
        UNIVERSITY OF CALIFORNIA\\\n
        DAVIS\\\n\n
        
        Approved:\\\n\n
        
        \rule{4in}{1pt}\\
        ~Bruno Nachtergaele (Chair)\\\n\n
        
        \rule{4in}{1pt}\\
        ~Motohico Mulase\\\n\n
        
        \rule{4in}{1pt}\\
        ~Alessandro Pizzo\
        
        \vfill
        
        Committee in Charge\\
        ~2012

    \end{center}

    \newpage

    %
    %
        

    ~\\[7.75in] 
    \centerline{
                \copyright\ Anna Vershynina,
                            2012. All rights reserved.
               }
    \thispagestyle{empty}
    \addtocounter{page}{-1}

    \newpage

    %
    %
    
    \centerline{To my mother, my constant supporter.}
    
    \newpage

    \doublespacing

    %
    %
    
    \tableofcontents
    
    \newpage

    %
    %
    
    %

    ~\vspace{-1in} 
    \begin{flushright}
        \singlespacing
        Anna Vershynina\\
        ~June 2012\\
        Mathematics
    \end{flushright}

    \begin{center}{\large Existence of the thermodynamic limit and asymptotic behavior of some irreversible quantum dynamical systems}
    \end{center}
    
    \centerline{\textbf{\underline{Abstract}}}

This dissertation discusses the properties of two open quantum systems with a general class of irreversible quantum dynamics. First we study Lieb-Robinson bounds in a quantum lattice systems. This bound gives an estimate for the speed of growth of the support of an evolved local observable up to an exponentially small error. In a second model we study the properties of a leaking cavity pumped by a random atomic beam. 

We begin by describing quantum systems on an infinite lattice with associated finite or infinite dimensional Hilbert space. The generator of the dynamics of this system is of the Lindblad-Kossakowski type and consists of two parts: the Hamiltonian interactions and the dissipative terms. We allow both of them to be time-dependent. This generator satisfies some suitable decay condition in space. We show that the dynamics with a such generator on a finite system is a well-defined quantum dynamics in a sense of a norm-continuous cocycle of unit preserving completely positive maps.

Lieb-Robinson bounds for irreversible dynamics were first considered in \cite{hastings:2004b} in the classical context and in \cite{poulin:2010} for a class of quantum lattice systems with finite-range interactions. We extend those results by proving a Lieb-Robinson bound for lattice models with a more general class of quantum dynamics.

Then we use Lieb-Robinson bounds for a finite lattice systems to prove the existence of the thermodynamic limit of the dynamics. We show that in a strong limit there exits a strongly continuous cocycle of unit preserving completely positive maps. Which means that the dynamics exists in an infinite system, where Lieb-Robinson bounds also holds.

In the second part of the dissertation we consider a system that consists of a beam of two-level atoms that pass one by one through the microwave cavity. The atoms are randomly excited and there is exactly one atom present in the cavity at any given moment. We consider both the ideal and leaky cavity and study the time asymptotic behavior of the state of the cavity.

We show that the number of photons increases indefinitely in the case of the ideal cavity. In the case of the leaking cavity the limiting state is independent of the initial state, it is not quasi-free and it is a non-equilibrium steady state. We also compute the associated energy flow.

    \newpage

    %
    %

    \chapter*{\vspace{-1.5in}Acknowledgments and Thanks}
    
        It is my great pleasure to thank my adviser Prof. Bruno Nachtergaele, who encouraged and challenged me throughout my years at the UC Davis. His patience and exceptional guidance helped me in my research and in completion of this dissertation.
        
        I am very grateful to Prof. Valentin Zagrebnov, whose invaluable collaboration helped me complete projects presented in this dissertation.
        
        I would like to thank all faculty and staff of the Math department, who created a friendly working environment on the department.
        
       Also I would like to thank my mother and my sister who always believed in me and gave me the needed support during the toughest times.

    \newpage

    %
    %

    \pagestyle{fancy}
    \pagenumbering{arabic}
       
    %
    %

    \newchapter{Introduction}{Introduction}{Introduction}
    \label{sec:Intro}
    \section{Introduction}
    
    The study of open quantum systems is significant to various areas of application of quantum theory such as quantum optics, quantum information theory, atomic physics and condensed-matter physics. In this dissertation we assume an open system dynamics of a quantum system. First we study Lieb-Robinson bounds for a lattice system with irreversible time-dependent dynamics. Then we study the properties of a state of the microwave cavity pumped by a random atomic beam. The dynamics of the system is also time-dependent and irreversible.
    
 In the past years Lieb-Robinson bounds have been shown to be a powerful tool to turn the locality properties of physical systems into useful mathematical estimates. Lieb-Robinson bounds provide an estimate for the speed of propagation of signals in a spatially extended system. 
 
 Lieb-Robinson bounds were so far considered on a lattice system (for example $\bZ^d$).  Briefly, if $A$ and $B$ are two observables supported in disjoint regions $X$ and $Y$, respectively, (see section \ref{sec:LR_rev} for the definition of the support) on a lattice system and if $\tau_t$ denotes the unitary time evolution of the system, a Lieb-Robinson bound is an estimate of the form
\begin{equation}\label{LR_rev_intro}
\|[\tau_t(A), B]\|\leq Ce^{-\mu(d(X,Y)-v|t|)},
\end{equation}
where $C,\mu$ and $v$ are positive constants and $d(X,Y)$ denotes the distance between $X$ and $Y$. Here $v$ is called Lieb-Robinson velocity. See Chapter \ref{LR_rev} for more details.

 From this inequality we see that for the small times $|t|<d(X,Y)/v$ the norm on the right-hand side is exponentially small. The reason we consider the commutator on the left-hand side of the Lieb-Robinson bounds is the following. The commutator between the observables $A$ and $B$ is zero if their supports are disjoint. The converse is also true: if the observable $A$ is such that its commutator with any observable $B$ supported outside some set $X$ is zero, then $A$ has a support inside the set $X$. This statement is also approximately true in the sense: suppose that there exist some $\epsilon>0$ such that $\|[A, B]\|\leq \epsilon \|B\| $ for any observable $B$ that is supported outside a set $X$ and some observable $A$. Then there exists an observable $A^{(\epsilon)}$ with the support inside the set $X$ that approximate the observable $A$: $\|A-A^{(\epsilon)}\|\leq \epsilon.$
 
 So Lieb-Robinson bounds say that the time evolution of the observable $A$ with the support in a set $X$ is mainly supported in the $\delta$-neighborhood of $X$, where $\delta>v|t|$ with $v$ is the Lieb-Robinson velocity.
 
 In quantum information if Alice has access to the observable $A$ with support $X$ and sends a signal to Bob who has access to the observable $B$ with support $Y$ then the strong enough signal that Bob can detect propagates with the Lieb-Robinson velocity $v$.

In 1972 Lieb and Robinson first showed the existence of the light-cone such that the amount of information signaled beyond it decays exponentially \cite{lieb:1972}. This result, known as Lieb-Robinson bounds, has been improved in \cite{hastings:2006}, \cite{hastings:2004b},  \cite{nachtergaele:2006a} -\cite{nachtergaele:2010} , for example, making the bound on the speed independed of the dimension of the spin spaces and extending the class of Hamiltonians describing the system. In these works the Lieb-Robinson bounds was considered for the unitary evolution, when the dynamics is described by the Heisenberg equation.

We consider a more general situation, when the dynamics is described by a Markovian dynamical semigroup (see Chapter \ref{Irrev-Dyn} for more details) with a time dependent generators. Lieb-Robinson bounds for irreversible dynamics were, to our knowledge, first considered in \cite{hastings:2004b} in the classical context and in \cite{poulin:2010} for a class of quantum lattice systems with finite-range interactions. We extended these results by proving the Lieb-Robinson bound for a long-range exponentially decaying time dependent interactions.

One of the main applications of Lieb-Robinson bounds is the existence of the dynamics in the thermodynamic limit. The existence of the thermodynamics limit is important as a fundamental property of any model meant to describe properties of bulk matter. In particular, such properties should be essentially independent of the size of the system which, of course, in any experimental setup will be finite. 
Our results are described in Section \ref{sec:Thermodynamics_irrev}.

In the second part of the thesis we use the results about the irreversible dynamics to study the properties of a leaking cavity pumped by a random atomic beam.

The micromaser system that we consider consists of a beam of two-level atoms that pass one by one through the microwave cavity. The atoms are excited with some probability $p$. The cavity is modeled by a single photon mode. The beam is tuned in such a way that there is exactly one atom in the cavity at any given time. While in the cavity the corresponding single atom is able to interact with a single mode of the cavity field. 

The interaction between the atom and the cavity in our model leaves the state of the atom invariant which may be interpreted as the limit of heavy atoms and soft photons. Hence, we focus our study on the time asymptotic behavior of the cavity state.

We do this both in the case of a perfect (non-leaky) cavity as well as in the case of a leaking cavity modeled by adding a Lindblad-Kossakowski dissipative term to the generator of the Hamiltonian dynamics (see section \ref{Irrev-Dyn} for details). We show that in the case of the perfect cavity, the expected number of photons in the cavity increases indefinitely in time if and only if $0 < p < 1$. In other words, fluctuations of the atom state are necessary for the pumping to occur. When the leakage is taken into the account, there is a well-defined limiting state of the cavity, which is independent of the initial state. We study this limiting state by computing its characteristic function and show that it is not quasi-free. The limiting state is a non-equilibrium steady state and we compute the associated energy flow (see section \ref{sec:States} for the definition of quasi-free and non-equilibrium steady state).

    \section{Summary of the main results}
    \subsection[Lieb-Robinson bounds and the existence of the thermodynamic limit]{Lieb-Robinson Bounds and the Existence of the Thermodynamic Limit for the Class of Irreversible Quantum Dynamics}

In Chapter \ref{sec:Thermodynamics_irrev} we prove Lieb-Robinson bounds for a class of the irreversible quantum dynamics and show the existence of the thermodynamic limit of the dynamics.

We consider a general situation, when the dynamics is described by a Markovian dynamical semigroup. The generator of the dynamics consists of both Hamiltonian and dissipative terms that may depend on time. In the finite volume $\Lambda$ the generator $\cL$ is of the following form
\begin{align}
\mathcal{L}_\Lambda(t)(A)&=\sum_{Z\subset\Lambda}  \Psi_Z(t)(A),\label{Def_L_intro}\\
\Psi_Z(t)(A)&= i[\Phi(t,Z), A]\nonumber\\
&\quad+\sum_{a=1}^{N(Z)} \left(L^*_a(t,Z)AL_{a}(t,Z)-\frac{1}{2}\{L_{a}(t,Z)^*L_{a}(t,Z), A\}\right),\nonumber
\end{align}
for all local observables $A$. Here for $Z\subset\Lambda$, $\Phi(t, Z)$ is the local time dependent interaction describing Hamiltonian terms and $L_a(t, Z)$ are local time dependent bounded operators. We put an exponential decay condition in space on the $\Psi_Z(t)$. The detailed set up is discussed in section \ref{sec:Thermo_setup}.

Fix $T>0$ and, for all observables $A$ with finite support in $\Lambda$, let $A(t), t\in [0,T]$ be a solution of the initial
value problem
\begin{equation}\label{ode_intro}
 \frac{d}{dt} A(t)=\cL_\Lambda(t)A(t),\quad A(0)=A.
\end{equation}
For $0\leq s\leq t\leq T$, define the family of maps
$\{\gamma_{t,s}^\Lambda\}_{0\leq s\leq t} \subset \cB(\cA_\Lambda, \cA_\Lambda)$ by $\gamma_{t,s}^\Lambda(A)=A(t)$,
where $A(t)$ is the unique solution of \eq{ode_intro} for $t\in [s,T]$ with initial condition $A(s)=A$.
In section \ref{sec:Thermo_Exist_dynam} we show that the dynamics $\gamma_t^\Lambda$ generated by (\ref{Def_L_intro}) exists and forms a norm-continuous cocycle of unit preserving completely positive maps.  

In section \ref{sec:Thermo_LR} we prove the Lieb-Robinson bound of the form
\begin{equation}\label{LR_intro}
\|\cK \gamma_{t}^\Lambda(B)\|\leq C(\cK, B) e^{-\mu(d(X, Y)- v|t|)},
\end{equation}
where $\cK$ is a completely bounded linear operator on the space of observables with support $X$  that vanishes on $\idty$ (see section \ref{sec:Thermo_setup} for the definition of the completely bounded map), $B$ is an observable with support $Y$, $d(X,Y)$ is the distance between the supports $X$ and $Y$, $C(\cK, B)$ is constant that depends on $\cK$ and $B$ and $v$ is a Lieb-Robinson velocity. It is important that the operator $\cK$ can be taken in a form $\cK(B)=[A,B],$ which makes the original Lieb-Robinson bound for reversible dynamics (\ref{LR_rev_intro}) a particular case of the Lieb-Robinson bound in the general form (\ref{LR_intro}).

In section \ref{sec:Thermo_Thermo} we look at the dynamics $\gamma_t^\Lambda$ when the set $\Lambda$ becomes infinitely large. We prove that the thermodynamic limit of the dynamics exists in the sense of strongly continuous one-parameter cocylce of completely positive unit preserving maps.

\subsection{Non-equilibrium state of a leaking photon cavity pumped by a random atomic beam}

The system we consider consists of the one-mode microwave cavity and a beam of randomly excited atoms. The cavity is described by quantum harmonic oscillator with the Hamiltonian $$H_C=\epsilon b^*b, $$ where $\epsilon>0$ and $b^*, b$ are boson creation and annihilation operators (see section \ref{sec:Fock_space} for more details on these operators). 

The beam of atoms is described as a quantum spin chain $$ H_A=\sum_{n\geq 1}Ea_n^*a_n,$$ where $E>0$ and $a_n^*, a_n$ are one-point fermion creation and annihilation operators for any $n\geq 1$ (see section \ref{sec:Fock_space}). 

Atoms are randomly excited with some probability $p$. They enter the cavity successively and there is only one atom present in the cavity at any given time. Atoms in the excited state interact with the cavity and produce a quantum field. The interaction between the $n-$th atom and the cavity is the following
\begin{equation*}
W(t)=\chi_{[(n-1)\tau, n\tau]}(t)(\lambda a_n^*a_n\otimes(b^*+b)),
\end{equation*} 
here $\chi_{I}(x)$ is a characteristic function of a set $I$, which is here to make sure that the $n-$th atom interacts with the cavity only when it is present there, which happen during the time interval $[(n-1)\tau, n\tau]$. The detailed set up of the system is given in the section \ref{sec:Nonequil_setup}.

At first we assume that the cavity insulates an atom-photon system from de-cohering interactions with its environment. So the Hamiltonian for the system is the sum of the Hamiltonian of the cavity and atoms and the interaction between them
\begin{equation*}
H(t)=\epsilon b^*b\otimes\idty+\sum_{n\geq1}\idty\otimes Ea_n^*a_n+\sum_{n\geq 1}\chi_{[(n-1)\tau, n\tau]}(t)(\lambda a_n^*a_n\otimes(b^*+b)).
\end{equation*} 
In section \ref{sec:Nonequil_nonleaking} we show that in this case the expected number of photons in the cavity oscillates around a linearly increasing in time function and only a beam of randomly exited atoms $(0<p<1 )$ is able to produce a pumping of the cavity by photons.

In a more general situation we should not assume that the cavity is perfectly insulated. Photons may either escape the cavity or be absorbed into the cavity walls at some constant non-zero rate $\sigma>0$. 

In this case the evolution of the system is described by a Markovian dynamics with the following generator
\begin{equation*}
L(t)(\rho_S)=-i[H(t),\rho_S]+\sigma b(\rho_S)b^*-\frac{\sigma}{2}\{b^*b, \rho_S\},
\end{equation*}
where $\sigma>0$ and $\rho_S$ is a state of the system.
In section \ref{sec:Nonequil_leaking} we show that the expected number of photons in the cavity stabilizes in time, which means that the energy leaking out of the cavity equals to the energy coming into the cavity.

In section \ref{sec:Nonequil_state} we compute the characteristic function of the limiting state of the cavity and show that the state is not quasi-free. (See Section \ref{sec:States} for the definition of a quasi-free state.)The limiting state of the system is a non-equilibrium state in a sense that it is not in a form  
\begin{equation}\label{Gibbs}
\omega_\beta(A)=\frac{\Tr(e^{-\beta b^*b}A)}{\Tr(e^{-\beta b^*b})},
\end{equation} 
for any $\beta\in\bC.$ 

Section \ref{sec:Nonequil_energy} is devoted to the calculation of the energy flow for the perfect cavity (\ref{sec:Nonequil_energy_perfect}) and the leaking cavity (\ref{sec:Nonequil_energy_leaking}).

    \newchapter{Preliminaries}{Preliminaries}{Preliminaries}
    \label{sec:Prelim}
        \section{Quantum systems}
\subsection{Fock space}\label{sec:Fock_space}
          
The quantum-mechanical states of a particle form a complex Hilbert space $\mathcal{H}$, for example it could be $L^2(\bR^d)$. The Hilbert space of $n$ identical but distinguisable particles is the tensor product $$\mathcal{H}^n=\bigotimes_{i=1}^n\mathcal{H}.$$ If the number of particles is not fixed, the states are described by vectors in the Fock space $\mathcal{F}(\mathcal{H})$ which is given by
\begin{equation*}
\mathcal{F}(\mathcal{H})=\bigotimes_{n\geq 0}\mathcal{H}^n,
\end{equation*}
where $\mathcal{H}^0=\mathbb{C}$. Therefore a vector $\psi\in\mathcal{F}(\mathcal{H})$ is a sequence of vectors $\psi=\{\psi^{(n)}\}_{n\geq 0}$, where $\psi^{(n)}=f_1\otimes ... \otimes f_n$, $f_k\in\cH$ for $1\leq k\leq n$. In quantum physics identical particles are indistinguishable which is reflected by the symmetry under the interchange of the particle coordinates. 

The first case is when the components $\psi^{(n)}$ of each $\psi$ are symmetric under the interchange of coordinates. Such particles are called \emph{bosons}, they form \emph{the Bose-Fock space} $\mathcal{F}_{+}(\mathcal{H})$. The second case is when the components $\psi^{(n)}$ of each $\psi$ are anti-symmetric under interchange of each pair of coordinates. Such particles are called \emph{fermions} and they form \emph{the Fermi-Fock} space $\cF_-(\cH)$.

Bose- and Fermi-Fock spaces are defined with the help of the operators $P_+$ and $P_-$ which are  defined on $\cF(\cH)$ as follows, for all $f_1,...f_n\in\cH$
\begin{align*}
&P_+(f_1\otimes ... \otimes f_n)=(n!)^{-1}\sum_{\pi}f_{\pi_1}\otimes...\otimes f_{\pi_n},\\
&P_-(f_1\otimes ... \otimes f_n)=(n!)^{-1}\sum_{\pi}\epsilon_\pi f_{\pi_1}\otimes...\otimes f_{\pi_n},\
\end{align*}
where $\epsilon_\pi$ is the sign on the permutation $\pi$. The sum is taken over all permutations $\pi:(1,2,...,n)\rightarrow (\pi_1,\pi_2,...\pi_n)$. These operator are extended by linearity to two densely defined operators with $\|P_{\pm}\|=1$. Then extending $P_{\pm}$ by continuity one gets two bounded operators of norm one. The Bose- and Fermi-Fock spaces are given by
\begin{equation*}
\cF_\pm(\cH)=P_\pm\cF(\cH).
\end{equation*}
A number operator $N$ is defined on $\cF(\cH)$ by 
\begin{equation*}
N\psi=\{n\psi^{(n)}\}_{n\geq 0}
\end{equation*}
with the domain $D(N)=\{\psi=\{\psi^{(n})\}_{n\geq 0}: \sum_{n\geq 0}n^2\|\psi^{(n)}\|^2<\infty\}. $

The particle creation and annihilation operators on $\cF(\cH)$ are defined as follows. For each $f\in\cH$, $a(f)\psi^{(0)}=0$ and $a^*(f)\psi^{(0)}=f$ and
\begin{align*}
&a_{\cF}(f)(f_1\otimes...\otimes f_n)=\sqrt{n}(f,f_1)f_2\otimes...\otimes f_n,\\
&a_{\cF}^*(f)(f_1\otimes...\otimes f_n)=\sqrt{n+1}f\otimes f_1\otimes...\otimes f_n,
\end{align*}
 for any $f_k\in\cH$, $k\in\mathbb{Z}$.
Extending by linearity one gets two densely defined operators. If $\psi^{(n)}\in\cH^n$ we get
$$\|a_\cF(f)\psi^{n}\|\leq\sqrt{n}\|f\|\|\psi^{n}\|\text{, } \|a_\cF^*(f)\|\leq\sqrt{n+1}\|f\|\|\psi^{(n)}\|. $$ Therefore $a_\cF(f)$ and $a_\cF^*(f)$ have well-defined extensions to the domain of $D(N^{1/2})$ of $N^{1/2}$.

From the definition one can see that $a^*$ operator is indeed the adjoint of $a$. For all $\phi,\psi\in D(N^{1/2})$ and $f\in\cH$ we have that
\begin{equation*}
(a_{\cF}^*(f)\phi, \psi)=(\phi, a_\cF(f)\psi).
\end{equation*}

The creation and annihilation operators on the Fock spaces $\cF_{\pm}(\cH)$ are defined by 
\begin{align*}
a_{\pm}(f)=P_{\pm}a(f)P_{\pm}=a(f)P_{\pm},\\
a^*_\pm(f)=P_{\pm}a^*(f)P_{\pm}=P_{\pm}a^*(f).
\end{align*}

In future when it is clear what Fock space we are working in we may drop the indices $\pm$.
In Section \ref{sec:NELS} we will use boson creation and annihilation operators $b^*, b$ and fermion creation and annihilation operators on the basis vectors $\{e_k\}_{k\geq 1}$ of the Fermion-Fock space $\cF_{-}(\cH)$
\begin{equation}\label{creation-annihilation}
a(e_k)=a_k,\text{ and } a^*(e_k)=a_k.
\end{equation}

The important properties of these operators are canonical commutation relations (CCR) for boson creation and annihilation operators  and canonical anti-commutation relations (CAR) for fermion creation and annihilation operators.

Boson operators satisfy the following CCR relations
\begin{align}\label{CCR}
&[a_+(f), a_+(g)]=0=[a^*(f), a^*(g)],\\
&[a_+(f), a^*(g)]=(f,g)\idty.\nonumber
\end{align}

Fermion operators satisfy the following CAR relations
\begin{align}\label{CAR}
&\{a_-(f), a_-(g)\}=0=\{a_-^*(f), a^*_-(g)\}\\
&\{a_-(f), a^*_-(g)\}=(f,g)\idty.\nonumber
\end{align}

\subsection{Weyl operators}

In this section we will consider the Boson-Fock space $\cF_+(\cH)$. So we will drop the plus index on the creation, annihilation operators. 

It is convenient to introduce the family of operators $\{b(f); f\in\cH\}$ as follows 
\begin{equation*}
b(f)=a(f)+a^*(f),
\end{equation*}
where $a^*, a$ are boson creation and annihilation operators.
 Each $b(f)$ is a self-adjoint linear operator on the Fock space. The canonical commutation relations (\ref{CCR}) in terms of the operators $b$ now take form 
\begin{equation*}
[b(f), b(g)]=2i\Im(f,g).
\end{equation*}

The Weyl operator is defined as follows
\begin{equation*}
W(f)=\exp\{ib(f)\}.
\end{equation*}

The operators $b(f)$ as well as the creation and annihilation operators are unbounded, while the Weyl operator is unitary. This is one of the reasons it is easier to work with Weyl operators than with creation and annihilation operators. 

The Baker-Campbell-Hausdorff formula shows that for two operators $X$ and $Y$
\begin{equation*}
e^Xe^Y=e^{X+Y}e^{\frac{1}{2}[X,Y]},
\end{equation*}
if $[X,[X,Y]]=0=[Y,[X,Y]]$. Using this formula one can get the Weyl form of the canonical commutation relations
\begin{equation*}
W(f)W(g)=e^{-\frac{i}{2}\Im(f,g)}W(f+g)=e^{-i\Im(f,g)}W(g)W(f).
\end{equation*}

Denote the algebra of observables generated by all the Weyl operators as $\mathcal{U}$.

\subsection{States}
\label{sec:States}

The \emph{state} $\omega$ on the algebra of observables $\mathcal{U}$ maps each observable $A$ into its expectation value which is in general a complex number $\omega(A)$ with the properties\\
1) normalization: $\omega(\idty)=1$\\
2) linearity: for each pair $A, B$ of observables and each pair $\lambda, \mu$ of complex numbers one has $\omega(\lambda A+\mu B)=\lambda\omega(A)+\mu\omega(B)$\\
3) positivity: for each observable $A$, $\omega(A^*A)\geq 0.$

Let $\omega$ be any state on the Weyl algebra $\mathcal{U}$. This state is well-defined if for all $f\in\cH$ all the expectation values $\omega(W(f))$ are known. We consider states such that their expectation values could be written in terms of correlation functions
\begin{align}\label{State_Weyl}
\omega(W(f))&=\omega(e^{ib(f)})=\sum_{n=0}^\infty \frac{i^n}{n!}\omega(b(f)^n)\nonumber\\ 
&=\exp\{\sum_{n=1}^\infty\frac{i^n}{n!}\omega(b(f)^n)_t\},
\end{align}
where the truncated correlation functions $\omega(...)_t$ are defined recursively
\begin{equation*}
\omega(b(f_1)...b(f_n))=\sum\omega(b(f_1)...)_t...\omega(...b(f_n))_t,
\end{equation*}
here the sum is taken over all possible ordered partitions $(1,...),...(...,n)$ of the set $\{1,...,n\}$. The expression $\omega(b(f_1)...b(f_n))_t$ is called the truncated correlation function of order $n$. See \cite{Verbeure} for more details.

A state $\omega$ of the boson algebra of observables is called \emph{a quasi-free state}, if all its truncated correlation functions of orders $n>2$ vanish.

From (\ref{State_Weyl}) it follows that the general quasi-free state is determined by its one- and two-point correlation functions and therefore it is of the following form
\begin{align}
\omega(W(f))&=\exp\{i\omega(b(f))-\frac{1}{2}\omega(b(f)b(f))_t\}\nonumber\\
&=\exp\{i\omega(b(f))-\frac{1}{2}(\omega(b(f)^2)-\omega(b(f))^2)\}.\label{quasi-free}
\end{align}

The Gibbs states or \emph{the equilibrium states} are given by
\begin{equation}\label{Gibbs}
\omega_\beta(A)=\frac{\Tr(e^{-\beta a^*a}A)}{\Tr(e^{-\beta a^*a})},
\end{equation} 
for any $\beta\in\bC.$ The state that is not in this form is called \emph{a non-equilibrium state}.
It can be shown that every Gibbs state of the form (\ref{Gibbs}) is quasi-free (see \cite{Verbeure} for details).

\section{Dynamics of quantum systems}
\subsection{Reversible dynamics}

In quantum mechanics for a physical system described in terms of a Hilbert space $\cH$ a state $\phi\in\cH$ of the systems evolves in time under the action of a one-parameter strongly continuous group of unitary operators.

A family of unitary maps $U_t\in\cB(\cH)$, $t\in\bR$ is called \emph{a strongly continuous one-parameter group} of unitary operators if\\
1) $U_tU_s=U_{t+s},$ for $s,t\in\bR$,\\
2) $\lim_{t\rightarrow 0}\|U_t\phi-\phi\|=0$, for $\phi\in\cH$,\\
3) $U_{t}^{-1}=U_{-t}$.

If $\{U_t\}_{t\in\bR}$ is a strongly continuous one-parameter group of unitary operators, then by Stone's theorem \cite{} there exists a self-adjoint operator $H$, the Hamiltonian of the system, such that
\begin{equation*}
U_t=e^{-itH},
\end{equation*}
for $t\in\bR$.

The dynamics of the system can be described using two equivalent perspectives. Either a state changes in time while all obervables remain the same, of the dynamical group acts on the set of observables leaving the states unchanged.

The second case is called "the Heisenberg picture". The time evolution of the observable $A\in\cB(\cH)$ is described by
\begin{equation*}
\tau_t(A)=U_t^*AU_t=e^{itH}Ae^{-itH}.
\end{equation*}

\subsection{Lieb-Robinson bounds for reversible dynamics}
\label{sec:LR_rev}

For Lieb-Robinson bounds the quantum systems are considered on the countable set of vertices $\Gamma$ (called lattice) which is equipped with a metric $d$. A Hilbert space $\cH_x$ is assigned to each vertex $x\in\Gamma$. For any finite subset $\Lambda\subset\Gamma$ the Hilbert space associated with it is the tensor product 
\begin{equation}\label{space_H} 
\cH_\Lambda=\bigotimes_{x\in\Lambda}\cH_x.
\end{equation}

The local algebra of observables over $\Lambda$ is $$\cA_\Lambda=\bigotimes_{x\in\Lambda}\cB(\cH_x), $$ where $\cB(\cH_x)$ denotes the algebra of bounded linear operators on $\cH_x$. 

If $\Lambda_1\subset \Lambda_2$, then we may identify $\mathcal{A}_{\Lambda_1}$ in
a natural way with the subalgebra
$\mathcal{A}_{\Lambda_1}\otimes \idty_{\Lambda_2\setminus\Lambda_1}$ of
$\mathcal{A}_{\Lambda_2}$, and simply write
$\mathcal{A}_{\Lambda_1}\subset\mathcal{A}_{\Lambda_2}$.
Then the algebra of local observables is defined as an inductive limit
\begin{equation}\label{loc_observables}
\mathcal{A}_{\Gamma}^{\rm loc}= \bigcup_{\Lambda\subset\Gamma}\mathcal{A}_\Lambda.
\end{equation}
See \cite{} for more detailed on the inductive limit. 
The $C^*$-algebra of quasi-local observables $\mathcal{A}_{\Gamma}$ is the norm
completion of $\mathcal{A}_{\Gamma}^{\rm loc}$.

The \emph{support} of the observable $A\in\mathcal{A}_\Lambda$ is the minimal set
$X\subset\Lambda$ for which  $A=A^\prime\otimes \idty_{\Lambda\setminus X}$
for some $A^\prime\in\mathcal{A}_X$.

We assume that there exists a non-increasing function $F: [0, \infty)\rightarrow (0, \infty)$
such that:\\
i) $F$ is uniformly integrable over $\Gamma$, i.e.,
\begin{equation}\label{F_i}
\|F\|:= \sup_{x\in\Gamma}\sum_{y\in\Gamma} F(d(x,y)) < \infty,
\end{equation}
and\\
ii) $F$ satisfies
\begin{equation}\label{F_ii}
C:= \sup_{x,y\in \Gamma}\sum_{z\in\Gamma}\frac{F(d(x,z))F(d(y,z))}{F(d(x,y))}  < \infty.
\end{equation}

\begin{example} As an example one may take $\Gamma=\bZ^\nu$ for some integer $\nu\geq 1$ with the metric $d(x,y)=|x-y|=\sum_{j=1}^\nu|x_j-y_j|.$ In this case, the function $F$ can be chosen as $F(|x|)=(1+|x|)^{-\nu-\epsilon}$ for any $\epsilon>0.$ To show that the constant $C$ is finite we use the triangle inequality and the symmetry of the taken supremum over $x$ and $y$
\begin{eqnarray*}
&&C=\sup_{x,y\in\bZ^\nu}\sum_{z\in\bZ^\nu}\frac{(1+|x-y|)^{\nu+\epsilon}}{(1+|x-z|)^{\nu+\epsilon}(1+|y-z|)^{\nu+\epsilon}}\\
&&\leq\sup_{x,y\in\bZ^\nu}\sum_{z\in\bZ^\nu}\frac{(1+|x-z|+|y-z|+1)^{\nu+\epsilon}}{(1+|x-z|)^{\nu+\epsilon}(1+|y-z|)^{\nu+\epsilon}}\\
&&\leq\sup_{x,y\in\bZ^\nu}\sum_{z\in\bZ^\nu}\Bigl(\frac{1}{1+|x-z|}+\frac{1}{1+|y-z|}\Bigr)^{\nu+\epsilon}.
\end{eqnarray*}
Now we use the inequality between geometric and algebraic mean: for any $\alpha\geq 1$ and any $a,b\geq 0$, $$\Bigl(\frac{a+b}{2}\Bigr)^\alpha\leq\frac{a^\alpha+b^\alpha}{2}.$$ Continuing the calculations 

\begin{eqnarray*}
&&C\leq\sup_{x,y\in\bZ^\nu}\sum_{z\in\bZ^\nu}2^{\nu+\epsilon-1}\Bigl(\Bigl(\frac{1}{1+|x-z|}\Bigr)^{\nu+\epsilon}+\Bigl(\frac{1}{1+|y-z|}\Bigr)^{\nu+\epsilon}\Bigr)\\
&&=2^{\nu+\epsilon-1}\Bigl(\sup_{x\in\bZ^\nu}\sum_{z\in\bZ^\nu}\Bigl(\frac{1}{1+|x-z|}\Bigr)^{\nu+\epsilon} +\sup_{y\in\bZ^\nu}\sum_{z\in\bZ^\nu}\Bigl(\frac{1}{1+|y-z|}\Bigr)^{\nu+\epsilon}\Bigr).\\
\end{eqnarray*}

Since two supremums are the same and each of them is achieved at any value we have 
\begin{eqnarray*}
&&C\leq2^{\nu+\epsilon}\sup_{x\in\bZ^\nu}\sum_{z\in\bZ^\nu}\Bigl(\frac{1}{1+|x-z|}\Bigr)^{\nu+\epsilon}\\
&&   \leq  2^{\nu+\epsilon}\sum_{w\in\bZ^\nu}\frac{1}{(1+|w|)^{\nu+\epsilon}}<\infty.
\end{eqnarray*}
\end{example}

Having a set $\Gamma$ with a function $F$ that satisfies (\ref{F_i}) and (\ref{F_ii}), we can define
for any $\mu>0$ the function
\begin{equation}\label{F_mu}
F_\mu(d)=e^{-\mu d}F(d),
\end{equation}
which then also satisfies i) and ii) with $\|F_\mu\|\leq \|F\|$ and $C_\mu\leq C$.

The Hamiltonian of the system is described by local Hamiltonians $H^{loc}=\{H_x\}$, where $H_x$ is a not necessarily bounded self-adjoint operator over $\cH_x$, and bounded perturbations defined in terms of interaction $\Phi$. The interaction $\Phi$ is a map from the set of subsets of $\Gamma$ to $\cA_\Gamma$ with the properties that for each finite set $X\subset\Gamma$, $\Phi(X)\in\cA_X$ and $\Phi(X)^*=\Phi(X)$. For any $\mu\geq 0$, denote by $\cB_\mu(\Gamma)$ the set of interactions for which
\begin{equation}\label{interaction_rev}
\|\Phi\|_\mu=\sup_{x,y\in\Gamma}\sum_{X\ni x,y}\frac{\|\Phi(X)\|}{F_\mu(x,y)}<\infty.
\end{equation}

For every finite subset $\Lambda\subset\Gamma$ the Hamiltonian is of the form
\begin{equation*}
H_\Lambda=H_\Lambda^{loc}+H_\Lambda^{\Phi}=\sum_{x\in\Lambda}H_x+\sum_{X\subset\Lambda}\Phi(X).
\end{equation*}

Since these operators are self-adjoint they generate a Heisenberg dynamics $$\tau_t^\Lambda(A)=e^{itH_\Lambda}Ae^{-itH_\Lambda}$$ for any $A\in\cA_\Lambda$.

The following theorem provides the Lieb-Robinson bounds.
\begin{theorem}
Let $X$ and $Y$ be two disjoint subsets of $\Lambda$. Then for any pair of local observables $A\in\cA_X$ and $B\in\cA_Y$ one has that
\begin{equation}\label{LR_rev}
\|[\tau_t^\Lambda(A),B]\|\leq \frac{2\|A\|\|B\|}{C_\mu} (e^{2\|\Phi\|_\mu C_\mu |t|}-1)\sum_{x\in X}\sum_{y\in Y}F_{\mu}(d(x,y)).
\end{equation}
  \end{theorem}
  
The bound in the theorem could be rewritten using the properties of the function $F$ as follows
\begin{equation}\label{LR_rev_reg}
  \|[\tau_t^\Lambda(A),B]\|\leq \frac{2\|A\|\|B\|}{C_\mu}\|F\|\min(|X|, |Y|) e^{-\mu(d(X, Y)-v_\mu|t|)},
\end{equation}
  where $v_\mu=\frac{2\|\Phi\|_\mu C_\mu}{\mu}$ is the Lieb-Robinson velocity.
  
  This theorem is a special case of the general Lieb-Robinson bound (\ref{LR_intro}) that we are going to consider here, so the proof of this theorem is included in the proof in section \ref{sec:Thermo_LR} as a particular case and it also can be found separately in \cite{nachtergaele:2009a}.

\subsection{Existence of the dynamics in the thermodynamic limit}

We assume that the set $\Gamma$, which is still equipped with a metric $d$, as a countable set with infinite cardinality. One can take, for example, $\Gamma=\bZ^\nu$, for some $\nu\geq1$.

The thermodynamic limit is taken over an increasing exhausting sequence of finite subsets $\Lambda_n\subset\Gamma$.
\begin{theorem}
Let $\mu>0$ and $\Phi\in\cB_\mu(\Gamma)$. The dynamics $\tau_t$ corresponding to $\Phi$ exists as a strongly continuous one-parameter group of automorphisms on $\cA_\Gamma$ such that for all $t\in\bR$
\begin{equation*}
\lim_{n\rightarrow\infty}\|\tau_t^\Lambda(A)-\tau_t(A)\|=0
\end{equation*}
for all $A\in\cA_\Gamma$.
\end{theorem}

The proof of this theorem can be found in \cite{nachtergaele:2009a} or, as a particular case, in section \ref{sec:Thermo_Thermo} where we prove the existence of the thermodynamic limit for the irreversible dynamics.

\subsection{Irreversible dynamics}\label{Irrev-Dyn}

In general, if we consider an open system taking into the account the interaction between the system and an environment we have to consider a non-Hamiltonian system. The dynamical maps $\gamma_t$ of such system form \emph{a one-parameter completely positive dynamical semigroup} on an algebra of observables $\cU$, which is described by the properties\\
1) $\gamma_t$ is completely positive,\\
2) $\gamma_t(I)=I$,\\
3) $\gamma_s\gamma_t=\gamma_{t+s},$\\
4) $\gamma_t \rightarrow \idty$ when $t\rightarrow 0$ strongly, i.e. for any $A\in\cU$, $\lim_{t\rightarrow\infty}\|\gamma_t(A) - A\|=0$.\\
A map $\gamma$ is called \emph{completely positive} if for any $n>1$, the linear maps $\gamma\otimes\id_{M_n}$ where $M_n$ are $n\times n$ complex matrices, are positive.

A strongly continuous group of automorphisms is a special case of a dynamical semigroup for $\gamma_t(A)=\tau_t(A).$

The dynamics $\gamma_t$ of the open system is the solution to the Markovian master equation
\begin{equation}\label{Markovian_ME}
\frac{d}{dt}\gamma_t(A)=\cL (\gamma_t(A)),
\end{equation}
for all $t\in\bR$, with $\gamma_0(A)=A$, where the generator $\cL$ acts on the space of observables.

It was shown in \cite{lindblad} that the generator $\cL$ is in the Lindblad-Kossakowski form  
\begin{equation}\label{Linblad_gen}
\cL(A)=i[H,A]+\sum_a(L_a^*AL_a-\frac{1}{2}\{L_a^*L_a,A\}),
\end{equation}
where $L_a$ are such that $\sum_a L_a^*L_a$ is a bounded operator and here $\{A,B\}=AB+BA$ is the anti-commutator. The Hamiltonian $H$ have bounded interaction terms, but may have unbounded on-site terms as in section \ref{sec:LR_rev}, making the generator $\cL$ unbounded.

The proof that the time evolution of an open system is Markovian in the weak coupling limit can be found in \cite{Davies}. It was shown there that the limit of the unitary dynamics that depends on a coupling constant may not be unitary, but it satisfies the Markovian master equation (\ref{Markovian_ME}) when the coupling constant converges to zero for a rescaled time variable.

In this dissertation we are going to discuss the dynamics generated by the observables of the type (\ref{Linblad_gen}), where $H$ and each observable $L_a$ are time-dependent. We will show that the evolution with this type of generator is well-defined as a norm-continuous cocycle of unit preserving completely positive maps.

    \newchapter{Thermodynamic limit}{Thermodynamic limit for a class of irreversible quantum dynamics}{Lieb-Robinson bounds and the existence of the thermodynamic limit for a class of irreversible quantum dynamics}
     \label{sec:Thermodynamics_irrev}

    \section{Set up}
    \label{sec:Thermo_setup}
The set up here is the similar to the set up for the Lieb-Robinson bounds for unitary dynamics as in section \ref{sec:LR_rev}. The system is considered on the countable set $\Gamma$ equipped with a metric $d$. We put the same restriction on the lattice: there is a non-increasing function $F:[0,\infty]\rightarrow(0,\infty)$ that satisfies (\ref{F_i}) and (\ref{F_ii}). 

The Hilbert space of states $\cH_\Lambda$ of any finite subsystem $\Lambda\subset\Gamma$ is defined by (\ref{space_H}) as a tensor product of Hilbert spaces $\cH_x$ of every point $x$ in $\Lambda$. The $C^*$-algebra of quasi-local observables $\cA_\Gamma$ is the norm completion of the algebra of local observables $\cA^{loc}_\Gamma$ defined in (\ref{loc_observables}) as a space of bounded linear operators on $\cH_\Lambda$.

The generator of the dynamics is now consists of two parts: the Hamiltonian interactions and the dissipative terms. We allow both of them to be time-dependent. 

As in the case of unitary dynamics, the Hamiltonian part is described by an interaction $\Phi(t,\cdot)$, which is now time-dependent. For every time $t\in\R$ the interaction $\Phi(t,\cdot)$ is a map from the set of finite subsets of $\Gamma$ to $\cA_\Gamma$, such that for every finite $Z\subset\Gamma$
\begin{enumerate}
\item $\Phi(t,Z)\in\cA_Z,\\$
\item$ \Phi(t,Z)^*=\Phi(t,Z).$
\end{enumerate}

The dissipative part is described by terms of Lindblad form. For every finite $Z\subset\Gamma$ these terms are defined by a set of operators $L_a(t,Z)\in\cA_Z$, where $a=1,..., N(Z)$. It is possible for $N(Z)=\infty$ with an additional assumption of the convergence of the sum. 

Then, for any finite set $\Lambda\subset\Gamma$ and time $t\in\R$ the generator $\cL_\Lambda\in\cB(\cA_\Lambda,\cA_\Lambda)$ is defined as follows: for all $A\in\cA_\Lambda$,
\begin{align}
\cL_\Lambda(t)(A)=&\sum_{Z\subset\Lambda}\Psi_Z(t)(A),\label{Def_L} \text{ where}\\
\Psi_Z(t)(A)=& i[\Phi(t,Z), A]\label{psiz}\\
&+\sum_{a=1}^{N(Z)} \left(L^*_a(t,Z)AL_{a}(t,Z)-\frac{1}{2}\{L_{a}(t,Z)^*L_{a}(t,Z), A\}\right)\nonumber,
\end{align} 
here $\{A,B\}=AB+BA$ is the anticommutator of $A$ and $B$. The operator $\Psi_Z(t)$ can be viewed as a bounded linear transformation on $\cA_\Lambda$ for any $\Lambda\supset Z$, which can be written in the form $\Psi_Z(t)\otimes\id_{\cA_{\Lambda\setminus Z}}$. For every $Z\subset\Lambda$ the norm of these maps can be bounded independently of the choice of $\Lambda\supset Z$ as follows:
\begin{equation*}
\|\Psi_Z(t)\|\leq2\|\Phi(t,Z)\|+2\sum_{a=1}^{N(Z)}\|L_a(t,Z)\|^2.
\end{equation*}
If $N(Z)=\infty$ we can insure that the sums $\sum_{a=1}^{N(Z)}\|L_a(t,Z)\|^2$ converge guaranteeing the uniform boundedness of the maps $\Psi_Z(t)$. However, it is more general to assume that the maps $\Psi_Z(t)$ defined on $\cA_Z$ are completely bounded. A map $\Psi\in\cB(\cA_Z)$ is called \textit{completely bounded} if for all $n>1$, the linear maps $\Psi\otimes\id_{M_n}$, where $M_n=\cB(\C^n)$ are $n\times n$ complex matrices, defined on $\cA_Z\otimes M_n$ are bounded with uniformly bounded norm. The \textit{cb-norm} is then defined by 
\begin{equation*}
\Vert \Psi\Vert_{\rm cb}=\sup_{n\geq 1} \Vert \Psi\otimes {\rm id}_{M_n}\Vert<\infty.
\end{equation*}
By this definition the cb-norm of $\Psi_Z(t)\in\cB(\cA_Z,\cA_Z)$ is independent of any $\Lambda$ such that $Z\subset\Lambda\subset\Gamma$.

The main assumption that we make is the following.

\begin{assumption}\label{assumption1}
Given $(\Gamma, d)$ and $F$ as described at the beginning of this section, the following hypotheses hold:

\begin{enumerate}
\item
For all finite $\Lambda\subset\Gamma$, $\mathcal{L}_\Lambda(t)$ is norm-continuous in $t$ (with the uniform operator norm on $\cB(\cA_\Lambda,\cA_\Lambda)$),
and hence uniformly continuous on compact intervals.
\item
There exists $\mu>0$ such that for every $t\in\mathbb{R}$
\begin{equation}\label{Decay_L}
\|\Psi\|_{t,\mu}:=\sup_{s\in[0,t]}\sup_{x,y\in\Lambda\subset\Gamma}
\sum_{Z\ni x,y}\frac{\|\Psi_Z(s)\|_{\rm cb}}{F_\mu(d(x,y))}<\infty.
\end{equation}
where $\Vert\cdot\Vert_{\rm cb}$ denotes the cb-norm of completely bounded
maps.
\end{enumerate}
\end{assumption}

The second part of the assumption means that the interaction weakens exponentially as the diameter of the support grows. This condition is similar to (\ref{interaction_rev}) in unitary dynamics case.

From definitions note that we have
$$
\|\mathcal{L}_\Lambda(t)\|\leq \sum_{Z\subset\Lambda}\|\Psi_Z(t)\|\leq\sum_{x,y\in\Lambda}
\sum_{Z\ni x,y}\|\Psi_Z(t)\|_{\rm cb}\leq \|\Psi\|_{t,\mu}|\Lambda|\|F\|.
$$
Define
\begin{equation}\label{Mt}
M_t=\|\Psi\|_{t,\mu}|\Lambda|\|F\|\, .
\end{equation}
From the definition (\ref{Decay_L}) it is clear that $M_s\leq M_t$ for $s<t$.

To define the dynamics of the system fix $T>0$ and, for all $A\in\A_\Lambda$, let $A(t), t\in [0,T]$ be a solution of the initial
value problem
\begin{equation}\label{ode}
 \frac{d}{dt} A(t)=\cL_\Lambda(t)A(t),\quad A(0)=A.
\end{equation}
Since $\|\mathcal{L}_\Lambda(t)\|\leq M_T<\infty$, this solution exists and is unique by the
standard existence and uniqueness results for ordinary differential equations.

For $0\leq s\leq t\leq T$, define the family of maps
$\{\gamma_{t,s}^\Lambda\}_{0\leq s\leq t} \subset \cB(\cA_\Lambda, \cA_\Lambda)$ by $$\gamma_{t,s}^\Lambda(A)=A(t),$$
where $A(t)$ is the unique solution of \eq{ode} for $t\in [s,T]$ with initial condition $A(s)=A$. 

Then, the {\em cocycle property}, $\gamma_{t,s}(A(s))=A(t)$, follows from the uniqueness of the
solution of \eq{ode}. 

Recall that a linear map $\gamma:\cA\to\cB$, where $\cA$ and $\cB$
are $C^*$-algebras is called {\em completely positive} if the maps $\gamma\otimes\id:
\cA\otimes M_n \to \cB\otimes M_n$ are positive for all $n\geq 1$. Here $M_n$ stands
for the $n\times n$ matrices with complex entries, and positive means that positive
elements (i.e., elements of the form $A^*A$)  are mapped into positive elements.


\section{Existence of a dynamics of the finite system as a semigroup of completely positive unit preserving maps}
\label{sec:Thermo_Exist_dynam}

 In this section we show that the dynamics defined in the previous section exists as a
norm-continuous cocycle of unit preserving completely positive maps. This extends the well-known result for time-independent generators of Lindblad form
\cite{lindblad:1976} to the time-dependent case. The theorem is formulated for more general generators of the Markovian dynamics, which includes the generators of the type (\ref{Def_L}).

\begin{theorem}\label{thm:finite}
Let $\cA$ be a $C^*$-algebra, $T>0$, and for $t\in [0,T]$, let $\cL(t)$ be a norm-continuous
family of bounded linear operators on $\cA$.
If\newline
(i) $\mathcal{L}(t)(\idty)=0;$\newline
(ii) for all $A\in\cA$, $\mathcal{L}(t)(A^{*})=\mathcal{L}(t)(A)^{*} $;\newline
(iii) for all $A\in\cA$,   $\mathcal{L}(t)(A^* A)-\mathcal{L}(t)(A^*)A-A^*\mathcal{L}(t)(A)\geq 0$;\newline
then the maps $\gamma_{t,s}$, $0\leq s\leq t\leq T$, defined by equation \eq{ode}, are a
norm-continuous cocycle of unit preserving completely positive maps.
\end{theorem}

It is straightforward to check that the $\cL_\Lambda(t)$ defined in \eq{Def_L}
satisfy properties (i) and (ii). Property (iii), which is called {\em complete dissipativity},
follows immediately from the observation
$$
\mathcal{L}_{\Lambda}(t)(A^* A)-\mathcal{L}_{\Lambda}(t)(A^*)A-A^*\mathcal{L}_{\Lambda}(t)(A)
=\sum_{Z\subset\Lambda}\sum_{a=1}^{N(Z)} [A, L_a(t,Z)]^*[A, L_a(t,Z)]\, \geq 0 \, .
$$
Therefore,  using this result, we conclude that, under Assumption \ref{assumption1}, for all
finite $\Lambda\subset\Gamma$, the maps $\gamma_{t,s}^\Lambda$, $0\leq s\leq t$, form
a norm-continuous cocycle of completely positive and unit preserving maps.

Before we begin the proof we make the simplification of some notations.

Let $\cL(t)$, $t\geq 0$, denote a family of operators on a $C^*$-algebra $\cA$ satisfying
the assumptions of Theorem \ref{thm:finite} and for $0\leq s\leq t$ consider the maps
$\cA\ni A \mapsto \gamma_{t,s}(A)$ defined by the solutions of (\ref{ode}) with initial condition
$A$ at $t=s$. Without loss of generality we can assume $s=0$ in the proof of the theorem because,
if we denote $\tilde{\cL}(t)=\cL(t+s)$, then $\gamma_{t,s}=\tilde{\gamma}_{t-s,0}$, where $
\tilde{\gamma}_{t,0}$ is the maps determined by the generators $\tilde{\cL}(t)$.

The maps $\gamma_{t,s}$ satisfy the equation
\begin{equation}\label{sol_gamma}
\gamma_{t,s}=\id+\int_s^t \cL(\tau)\gamma_{\tau, s}d\tau.
\end{equation}
In our proof of the complete positivity of $\gamma_{t,0}$ we will use an expression
for $\gamma_{t,0}$ as the limit of an Euler product, i.e. approximations $T_n(t)$ defined by
\begin{equation}\label{Tn}
T_n(t)= \prod_{k=n}^1\left(\id+\frac{t}{n}\mathcal{L}(\frac{kt}{n})\right)\, .
\end{equation}
The product is taken in the order so that the factor with $k=1$ is on the right.

 \begin{lemma}\label{lem:Euler_app}
Let $\cL(t)$, $t\geq 0$, denote a family of operators on a $C^*$-algebra $\cA$ satisfying
the assumptions of Theorem \ref{thm:finite}. Then, uniformly for all $t\in [0,T]$,
 \begin{equation*}
\lim_{n\to\infty} \Vert T_n(t) - \gamma_{t,0}\Vert = 0 \,,
 \end{equation*}
 where $T_n(t)$ is defined by \eq{Tn}.
 \end{lemma}

 \begin{proof}
 From the cocycle property established in Section \ref{sec:Thermo_setup}, we have
 \begin{equation*}
\gamma_{t,0}=\prod_{k=n}^1 \gamma_{t\frac{k}{n}, t\frac{k-1}{n}}.
 \end{equation*}
 Now, consider the difference
 \begin{align*}
 T_n(t)-\gamma_{t,0}&=\prod_{k=n}^1\left(\id+\textstyle\frac{t}{n}\mathcal{L}(\textstyle\frac{kt}{n})\right)
 -\prod_{k=n}^1 \gamma_{t\frac{k}{n}, t\frac{k-1}{n}}\\
 &=\sum_{j=1}^n\left[\prod_{k=n}^{j+1}\left(\id+\textstyle\frac{t}{n}\mathcal{L}(t\textstyle\frac{k-1}{n})\right)\right]
 \left[\left(\id+\frac{t}{n}\mathcal{L}(t\textstyle\frac{j-1}{n})\right)-
 \gamma_{t\frac{j}{n}, t\frac{j-1}{n}}\right] \gamma_{t\frac{j-1}{n},0}.
 \end{align*}
  To estimate the norm of this difference we look at each factor separately.

Using the boundedness of $\cL(t)$ and the fact that $M_t$, defined in \eq{Mt},
is increasing in $t$, the norm of the first factor is bounded from above by
$$
\|\prod_{k=n}^{j+1}(\id+\frac{t}{n}\mathcal{L}(t\frac{k-1}{n}))\|\leq \prod_{k=n}^1(1+
\frac{t}{n}\|\mathcal{L}(t\frac{k-1}{n})\|)\leq (1+\frac{t}{n}M_t)^n.
$$

To bound the second factor notice that from (\ref{sol_gamma}) we obtain
\begin{equation*}
\|\gamma_{t,s}\|\leq 1+\int_s^t\|\cL(\tau)\|\|\gamma_{\tau,s}\|d\tau.
\end{equation*}

Then by Gronwall inequality \cite[Theorem 2.25]{hunter} we have the following bound for
the norm of the $\gamma_{t,s}$:
\begin{equation*}
\|\gamma_{t,s}\|\leq e^{\int_s^t\|\cL(\tau)\|d\tau}\leq e^{M_t(t-s)}.
\end{equation*}

Using again (\ref{sol_gamma}) we can rewrite the second factor as follows:
 \begin{align*}
 &\left(\id+\frac{t}{n}\mathcal{L}(t\frac{j-1}{n})\right)-\gamma_{t\frac{j}{n}, t\frac{j-1}{n}}=
 \frac{t}{n}\mathcal{L}(t\frac{j-1}{n})-(\gamma_{t\frac{j}{n}, t\frac{j-1}{n}}-\id)\\
& = \int_{t\frac{j-1}{n}}^{t\frac{j}{n}}
 \left (\cL(t\frac{j-1}{n})-\cL(s)\gamma_{s, t\frac{j-1}{n}}\right)ds\\
 &=\int_{t\frac{j-1}{n}}^{t\frac{j}{n}}\left[\left(\cL(t\frac{j-1}{n})-\cL(s)\right)-\cL(s)
 (\gamma_{s, t\frac{j-1}{n}}-\id)\right]ds\\
 &=\int_{t\frac{j-1}{n}}^{t\frac{j}{n}}\left(\cL(t\frac{j-1}{n})-\cL(s)\right) ds-\int_{t\frac{j-1}{n}}^{t\frac{j}{n}}\cL(s)
 \int_{t\frac{j-1}{n}}^s\cL(\tau)\gamma_{\tau, t\frac{j-1}{n}}d\tau ds.
 \end{align*}

Therefore, the second factor is bounded from above by
\begin{align*}
\|\left(\id+\frac{t}{n}\mathcal{L}(t\frac{j-1}{n})\right)-\gamma_{t\frac{j}{n}, t\frac{j-1}{n}}\|
&\leq\frac{t}{n}\epsilon_n+M_t^2\int_{t\frac{j-1}{n}}^{t\frac{j}{n}}\int_{t\frac{j-1}{n}}^s e^{(\tau-t\frac{j-1}{n})M_t}
d\tau ds\\
&\leq\frac{t}{n}\epsilon_n+M_t^2e^{\frac{t}{n}M_t}\int_{t\frac{j-1}{n}}^{t\frac{j}{n}} (s-t\frac{j-1}{n})ds\\&=\frac{t}{n}
\left(\epsilon_n+M_t^2 e^{\frac{t}{n}M_t}\frac{t}{2n}\right),
\end{align*}
where $\epsilon_n\to 0$ as $t/n\to 0$ due
to the uniform continuity of $\cL(t)$ on the interval $[0,t]$.

The third factor can be estimated in a similar way:
 \begin{align*}
 \|\gamma_{t\frac{j-1}{n},0}\|&=\prod_{k=j-1}^1\|\gamma_{t\frac{k}{n},t\frac{k-1}{n}}\|=\prod_{k=j-1}^1\|1+\frac{t}{n}
 \mathcal{L}(s_k(\frac{t}{n})) \|\\
 &\leq\prod_{k=j-1}^1\left(1+\frac{t}{n}\|\mathcal{L}(s_k(\frac{t}{n}))\|\right)\\
 &\leq (1+\frac{t}{n}M_t)^n.
 \end{align*}

Therefore, combining all these estimates we obtain
 \begin{align*}
\|T_n(t)-\gamma_{t,0}\|&\leq n (1+\frac{t}{n}M_t)^n \frac{t}{n}\left(\epsilon_n+M_t^2 e^{\frac{t}{n}M_t}\frac{t}{2n}\right)
 (1+\frac{t}{n}M_t)^n\\
&\leq t e^{2t M_t} \left(\epsilon_n+M_t^2 e^{\frac{t}{n}M_t}\frac{t}{2n}\right).
 \end{align*}
 This bound vanishes as $n\rightarrow\infty$.
 \end{proof}

To prove Theorem \ref{thm:finite} we use the Euler-type approximation established in
Lemma \ref{lem:Euler_app}. We show that the action of $T_n(t)$ on a positive operator is bounded from below and this bound vanishes as $n$ goes to $\infty$.
\medskip

{\sc Proof of Theorem \ref{thm:finite}:}
First, we look at the each term in the Euler approximation $T_n(t)$ separately.
For any $t$ and $s$ the complete dissipativity property (iii)  of $\mathcal{L}(s)$, assumed
in the statement of the theorem, implies
\begin{align*}
0&\leq (\id+t\mathcal{L}(s))(A^*)(\id+t\mathcal{L}(s))(A)=(A^*+t\mathcal{L}(s)(A^*))(A+t\mathcal{L}(s)(A))\\
& =A^*A+tA^*\mathcal{L}(s)(A)+t\mathcal{L}(s)(A^*)A+t^2\mathcal{L}(s)(A^*)\mathcal{L}(s)(A)\\
& \leq A^*A+t\mathcal{L}(s)(A^*A)+t^2\mathcal{L}(s)(A^*)\mathcal{L}(s)(A).
\end{align*}
Since $(\mathcal{L}(s)(A))^*(\mathcal{L}(s)(A))\leq \|\mathcal{L}(s)\|^2\|A\|,$ one gets
\begin{align}\label{norm_1}
0&\leq (\id+t\mathcal{L}(s))(A^*A)+t^2\|\mathcal{L}(s)\|^2\|A\|^2\\
& \leq(\id+t\mathcal{L}(s))(A^*A)+t^2M_s^2\|A\|^2.
\end{align}

Let us apply the above inequality to the operator$B $, where $B^*B:=\|A\|^2-A^*A$. Note that $\|B^*B\|\leq\|A\|^2$,
so $\|B\|\leq\|A\|$.
\begin{align}\label{norm_2}
0&\leq (\id+t\mathcal{L}(s))(\|A\|^2-A^*A)+t^2M_s^2\|A\|^2\\
& =\|A\|^2-(\id+t\mathcal{L}(s))(A^*A)+t^2M_s^2\|A\|^2
\end{align}

{From} the (\ref{norm_1}) and (\ref{norm_2}) we obtain
\begin{equation}\label{2}
-t^2M_s^2\|A\|^2\leq (\id+t\mathcal{L}(s))(A^*A)\leq (1+t^2M_s^2)\|A\|^2
\end{equation}
and therefore:
\begin{equation*}
-(1+t^2M_s^2)\|A\|^2\leq (\id+t\mathcal{L}(s))(A^*A) \leq (1+t^2M_s^2)\|A\|^2.
\end{equation*}
So we get
\begin{equation}\label{3}
\|(\id+t\mathcal{L}(s))(A^*A)\|\leq (1+t^2M_s^2)\|A\|^2 \ .
\end{equation}

Now, in order to bound the approximation $T_n(t)$ we first derive the following auxiliary estimate.
For any fixed $n \geq 1$ we have:
\begin{equation}\label{bound}
\prod_{k=n}^1(\id+s\mathcal{L}(ks))(A^*A)\geq -s^2\|A\|^2M_{ns}^2(1+\frac{1}{n-1})^{n-1}\sum_{k=0}^{n-1} D(s)^k,
\end{equation}
where the value of $s$ is chosen to be such that
\begin{equation}\label{assumption_s}
D(s):=1+s^2M_{ns}^2 < {(1+\frac{1}{n-1})^{n-1}}/{(1+\frac{1}{n-2})^{n-2}},
\end{equation}
with the convention that $(1+\frac{1}{n-1})^{n-1} =1$, for $n=1$.

We prove this claim by induction. The statement holds for $n=1$ by (\ref{norm_2}).
Now, assume that \eq{bound} holds for $n-1$. Then
\begin{equation*}
\prod_{k=n-1}^{1}(\id+s\mathcal{L}(ks))(A^*A) +s^2\|A\|^2M_{(n-1)s}^2(1+\frac{1}{n-2})^{n-2}\sum_{k=0}^{n-2} D(s)^k\geq 0
\end{equation*}
Since the left-hand side is a positive operator, we can write it as $B^*B$. Then,
\begin{align*}
\prod_{k=n}^1(\id+s\mathcal{L}(ks))(A^*A)&=(\id+s\mathcal{L}(ns))(B^*B)\\&-s^2\|A\|^2M_{(n-1)s}^2(1+\frac{1}{n-2})^{n-2}
\sum_{k=0}^{n-2} D(s)^k\\
&\geq -s^2M_{ns}^2\|B^*B\|-s^2\|A\|^2M_{ns}^2(1+\frac{1}{n-2})^{n-2}\sum_{k=0}^{n-2} D(s)^k\, .
\end{align*}
Here, we used (\ref{norm_2}) and the fact that $M_t$ is monotone increasing.
This gives the following upper bound for $\Vert B^*B\Vert$:
\begin{align*}
\|B^*B\|&\leq \prod_{k=n-1}^{1} \|(\id+s\mathcal{L}(ks))(A^*A)\|+s^2\|A\|^2M_{(n-1)s}^2(1+\frac{1}{n-2})^{n-2}
\sum_{k=0}^{n-2}D(s)^k\\
&\leq\prod_{k=n-1}^{1}(1+s^2M_{ks}^2)\|A\|^2+s^2\|A\|^2M_{ns}^2(1+\frac{1}{n-2})^{n-2}\sum_{k=0}^{n-2}D(s)^k\\
& \leq\prod_{k=n-1}^{1}(1+s^2M_{ns}^2)\|A\|^2+s^2\|A\|^2M_{ns}^2(1+\frac{1}{n-2})^{n-2}\sum_{k=0}^{n-2}D(s)^k\\
&= \|A\|^2D(s)^{n-1}+s^2\|A\|^2M_{ns}^2(1+\frac{1}{n-2})^{n-2}\sum_{k=0}^{n-2}D(s)^k
\end{align*}
Therefore we obtain
\begin{align*}
&\prod_{k=n-1}^1(1+s\mathcal{L}(ks))(A^*A)\\
& \geq-s^2M_{ns}^2\|A\|^2D(s)^{n-1}-s^2M_{ns}^2(s^2M_{ns}^2+1)(1+\frac{1}{n-2})^{n-2}\|A\|^2\sum_{k=0}^{n-2}D(s)^k\\
&\geq-s^2M_{ns}^2\|A\|^2(1+\frac{1}{n-1})^{n-1}D(s)^{n-1}-s^2M_{ns}^2
(1+\frac{1}{n-1})^{n-1}\|A\|^2\sum_{k=0}^{n-2}D(s)^k\\&\geq-s^2M_{ns}^2
(1+\frac{1}{n-1})^{n-1}\|A\|^2\sum_{k=0}^{n-1}D(s)^k \ ,
\end{align*}
where to pass to the second inequality we use our assumption on $s$ \eq{assumption_s}.
This completes the proof of the bound (\ref{bound}).

To finish the proof of the theorem we use Lemma \ref{lem:Euler_app} to approximate the propagator and put
$s=\frac{t}{n}$ in the bound (\ref{bound}), which yields
\begin{equation}\label{last_eq}
\prod_{k=n}^1(1+\frac{t}{n}\mathcal{L}(\frac{kt}{n}))(A^*A)\geq -\frac{t^2}{n^2}\|A\|^2M_t^2(1+\frac{1}{n-1})^{n-1}
\sum_{k=0}^{n-1} D(\frac{t}{n})^k.
\end{equation}
Since $D(\frac{t}{n})^n=(1+\frac{t^2}{n^2}M_t^2)^n \rightarrow 1$  as $n\rightarrow \infty$,
we get the estimate $D(\frac{t}{n})^k\leq 2$ for $1\leq k \leq n$. The right hand side of (\ref{last_eq})
is bounded from below by $-\frac{t^2}{n^2}\|A\|^2 e \, M_t^2 2n$, which vanishes in the limit $n\rightarrow \infty$.

To show the complete positivity of $\gamma_{t,0}$ note that any generator
$\mathcal{L}_\Lambda(t)$ satisfying the assumptions of the theorem
can be considered as the generator for a dynamics on $\cA\otimes \cB(\Cx^n)$, for any
$n\geq 1$, which satisfies the same properties, and which generates $\gamma_{t,s}\otimes\id$
acting on  $\cA\otimes \cB(\Cx^n)$. By the arguments given above,
these maps are positive for all $n$. Hence, the $\gamma_{t,s}$ are completely positive.
\boxendproof

\section{Lieb-Robinson bounds for a class of irreversible dynamics}
\label{sec:Thermo_LR}

For $X\subset\Lambda$, let $\cB_X$ denote the
subspace of $\cB(\cA_X)$ consisting of all completely bounded
linear maps that vanish on $\idty$. 

It is important for us that  all operators of the form
$$
\mathcal{K}_X(B):= i [A, B]+\sum_{a=1}^{N} (L^*_a B L_{a}-\frac{1}{2}\{L_{a}^*L_{a}, B\})\,,
$$
where $A, L_a\in \cA_X$, belong to $\cB_X$, with
$$
\Vert \cK_X\Vert_{\rm cb}\leq 2\Vert A\Vert +2\sum_{a=1}^N \Vert L_a\Vert^2.
$$
In particular,
operators of the form $[A,\cdot]$ appearing in the standard Lieb-Robinson
bound (\ref{LR_rev}) are a special case of this general form. 

We can regard
$\cK_X$ as a linear transformation on $\cA_Z$, for all $Z$ such that $X\subset Z$,
by tensoring it with $\id_{\cA_{Z\setminus X}}$, and all these maps will be bounded
with norm less then $\Vert \cK_X\Vert_{\rm cb}$.

\begin{theorem}\label{thm:lrbounds}
Suppose Assumption \ref{assumption1} holds. Then  the maps $\gamma^\Lambda_{t,s}$
satisfy the following bound. For $X,Y\subset\Lambda$, and any operators $\cK\in\cB_X$ and
$B\in\cA_Y$ we have that
\begin{equation*}
\|\cK \gamma_{t,s}^\Lambda(B)\|\leq \frac{\| \cK \|_{\rm cb}\,\|B\|}{C_\mu} e^{\|\Psi\|_{t,\mu}
C_\mu |t-s|}\sum_{x\in X\subset\Lambda}\sum_{y\in Y\subset\Lambda}F(d(x,y))\,.
\end{equation*}
\end{theorem}

The proof of the Lieb-Robinson bounds for $\gamma^{\Lambda}_{t,s}$ is based on
a generalization of the strategy \cite{nachtergaele:2006} for reversible dynamics, and
on \cite{poulin:2010} for irreversible dynamics with time-independent generators.
This allows us to cover the case of irreversible dynamics with time-dependent generators.
\medskip

\noindent
{\sc Proof of Theorem \ref{thm:lrbounds}:}
Consider the function $f:[s,\infty)\rightarrow \cA$ defined by
\begin{equation*}
f(t)=\cK\gamma_{t,s}^\Lambda(B),
\end{equation*}
where $\cK\in\cB_X$ and $B\in\cA_Y$, as in the statement of the theorem.
For $X\subset\Lambda$, let $X^c=\Lambda\setminus X$ and
define $\cL_{X^c}$ and $\bar{\cL}_{X}$ by
\begin{eqnarray*}
\cL_{X^c}(t)&=&\sum_{Z, Z\cap X= \emptyset}\mathcal{L}_Z(t)\\
\bar{\cL}_X(t)&=&\cL_X(t)-\cL_{X^c}(t).
\end{eqnarray*}
Clearly, $[\cK,\cL_{X^c}(t)]=0$. Using this property, we easily derive
the following expression for the derivative of $f$:
\begin{align*}
f'(t)&=\cK\mathcal{L}(t)\gamma_{t,s}^\Lambda(B) \\
& = \mathcal{L}_{X^c}(t)\cK\gamma_{t,s}^\Lambda(B)+\cK\bar{\cL}_{X}(t)\gamma_{t,s}^\Lambda(B) \\
& = \mathcal{L}_{X^c}(t)f(t)+\mathcal{K}\bar{\cL}_{X}(t)\gamma_{t,s}^\Lambda(B)\ ,
\end{align*}
Let $\gamma^{X^c}_{t.s}$ be the cocycle generated by $\cL_{X^c}(t)$. Then,
using the expression for $f'(t)$ we find
\begin{equation*}
f(t)= \gamma^{X^c}_{t,s}f(s)+\int_s^t\gamma_{t,r}^{X^c}\cK\bar{\cL}_{X}(r)
\gamma_{r,s}^\Lambda(B)dr \ .
\end{equation*}
Since $\gamma^{X^c}_{t,s}$ is norm-contracting and
$\Vert \cK\Vert_{\rm cb}$ is an upper bound for the
$\Vert \cK\Vert$ regarded as an operator on $\cA_\Lambda$, for all
$\Lambda$, we obtain
\begin{equation}\label{f_norm_eneq}
\|f(t)\|\leq \|f(s)\|+\|\cK\|_{\rm cb}\int_s^t \|\bar{\cL}_{X}(r)\gamma_{r,s}^\Lambda(B)\|dr.
\end{equation}
Let us define the quantity
\begin{equation*}
C_B(X,t):= \sup_{\mathcal{T}\in\cB_X} \frac{\|\mathcal{T}
\gamma_{t,s}^\Lambda(B)\|}{\|\mathcal{T}\|_{\rm cb}}.
\end{equation*}
Note that we use the norm $\|\mathcal{T}\|_{\rm cb}$,
because, as mentioned before and in contrast to the usual operator norm,
 it is independent of $\Lambda$.
Then, we have the following obvious estimate:
\begin{equation*}
C_B(X,s)\leq \|B\|\delta_Y(X),
\end{equation*}
where $\delta_Y(X)=0$ if $X\cap Y=\emptyset$ and $\delta_Y(X)=1$ otherwise.
{From} the definition of the space
$\cB_X$ we get that $\mathcal{T}(B)=0$, when $\mathcal{T}\in\cB_X$, since
$B$ has a support in
$Y$ and $Y\cap X=\emptyset.$ \\
Therefore  (\ref{f_norm_eneq}) implies that
\begin{equation*}
C_B(X,t)\leq C_B(X,s)+\sum_{Z\cap X \neq \emptyset} \int_s^t\|\mathcal{L}_Z(r)\| C_B(Z,r)dr.
\end{equation*}
Iterating this inequality we find the estimate:
\begin{equation*}
C_B(X,t)\leq \|B\| \sum_{n=0}^{\infty} \frac{(t-s)^n}{n!} \, a_n \ ,
\end{equation*}
where:
\begin{equation*}
a_n\leq \|\Psi\|_{t,\mu}^nC_\mu^{n-1}\sum_{x\in X}\sum_{y\in Y}F (d(x,y)),
\end{equation*}
for $n\geq 1$ and $a_0=1$, (recall that $C_\mu$ is a constant, that appears in a definition of $F_\mu$).
The following bound immediately follows from this estimate:
\begin{equation*}
\|\cK\gamma_{t,s}^\Lambda(B)\|\leq \frac{\|\cK\|_{\rm cb}\|B\|}{C_\mu} e^{\|\Psi\|_{t,\mu}
C_\mu (t-s)}\sum_{x\in X\subset\Lambda}\sum_{y\in Y\subset\Lambda}F(d(x,y)).
\end{equation*}
If we take an exponentially decaying function $F$, defined in (\ref{F_mu}) as $F_\mu(d)=e^{-\mu d}F(d)$ we can rewrite the above bound as follows
\begin{equation}\label{LR_min}
\|\cK\gamma_{t,s}^\Lambda(B)\|\leq \frac{\|\cK\|_{\rm cb}\|B\|}{C_\mu} \|F\|\min(|X|, |Y|)
e^{-\mu(d(X,Y)-\frac{\|\Psi\|_{t,\mu} C_\mu}{\mu} (t-s))}.
\end{equation}
So the Lieb-Robinson velocity of the propagation for every $t\in\mathbf{R}$ is
\begin{equation*}
v_{t,\mu}:=\frac{\|\Psi\|_{t,\mu} C_\mu}{\mu}.
\end{equation*}
\boxendproof

Note that the Lieb-Robinson bound (\ref{LR_min}) depends only on the smallest of the supports of the
two observables. Therefore one can get a non-trivial bound even when one of the observables
has finite support but the support of the other is of infinite size (e.g., say half the system).

The bound in the LIbe-Robinson bound is \textit{uniform} in $\Lambda$. This is important
for the proof of existence of the thermodynamic limit of the dynamics, which
is one of the main applications of Lieb-Robinson bounds, which will be presented in the next section.

\section{Existence of the thermodynamic limit}
\label{sec:Thermo_Thermo}

The setup for the analysis of the thermodynamic limit can be formulated as follows.
Let $\Gamma$ be
an infinite set such as, e.g., the hypercubic lattice $\Ir^\nu$. We prove the existence
of the thermodynamic limit for an increasing and exhausting sequence of finite
subsets $\Lambda_n\subset\Gamma$, $n\geq 1$, by showing that  for each $A\in \cA_X$,
$(\gamma^{\Lambda_n}_{t,s}(A))_{n\geq 1}$ is a Cauchy sequence in the norm of
$\cA_\Gamma$. To this end we have to suppose that Assumption \ref{assumption1} (2) holds
\textit{uniformly} for all $\Lambda_n$, i.e., we can replace $\Lambda$ in \eq{Decay_L} by
$\Gamma$.

\begin{theorem}\label{thm:thermodynamiclimit}
Suppose that Assumption \ref{assumption1} holds and, in addition, that
\eq{Decay_L} holds for $\Lambda =\Gamma$. Then, there exists a strongly continuous
cocycle of unit-preserving completely positive maps $\gamma_{t,s}^\Gamma$ on
$\cA_\Gamma$ such that for all $0\leq s\leq t$, and any increasing exhausting sequence
of finite subsets  $\Lambda_n\subset\Gamma$, we have
\begin{equation}
\lim_{n\rightarrow\infty}\|\gamma_{t,s}^{\Lambda_n}(A)-\gamma_{t,s}^\Gamma(A)\|=0,
\end{equation}
for all $A\in\cA_\Gamma$.
\end{theorem}

The proof of existence of the thermodynamic limit mimics the method given in
the paper \cite{nachtergaele:2006}.

\noindent
{\sc Proof of Theorem \ref{thm:thermodynamiclimit}:}
Denote $\mathcal{L}_n=\mathcal{L}_{\Lambda_n}$ and $\gamma_{t,s}^{\Lambda_n}=\gamma_{t,s}^{(n)}$.
Let $n>m$, then $\Lambda_m\subset\Lambda_n$ since we have the exhausting sequence of subsets in $\Gamma$.
We will prove that for every observable $A\in\cA_X$ the sequence
$(\gamma_{t,s}^n(A))_{n\geq 1}$ is a Cauchy sequence.
In order to do that for any local observable $A\in\cA_X$ we consider the function
\begin{equation*}
f(t):=\gamma_{t,s}^{(n)}(A)-\gamma_{t,s}^{(m)}(A) \ .
\end{equation*}
Calculating the derivative, we obtain
\begin{align*}
f^\prime(t)&=\mathcal{L}_n\gamma_{t,s}^{(n)}(A)-\mathcal{L}_m\gamma_{t,s}^{(m)}(A)\\
&= \mathcal{L}_n(t)(\gamma_{t,s}^{(n)}(A)-\gamma_{t,s}^{(m)}(A))+(\mathcal{L}_n(t)-\mathcal{L}_m(t))\gamma_{t,s}^{(m)}(A)\\
&=\mathcal{L}_n(t)f(t)+(\mathcal{L}_n(t)-\mathcal{L}_m(t))\gamma_{t,s}^{(m)}(A).
\end{align*}
The solution to this differential equation is
\begin{align*}
f(t)=\int_s^t \gamma_{t,r}^{(n)}\red{([}\mathcal{L}_n(r)-\mathcal{L}_m(r)\red{]}\gamma_{r,s}^{(m)}\red{(}A\red{))} dr.
\end{align*}
Since $\gamma_{t,r}^(n)$ is norm-contracting, from this formula we get the estimate:
\begin{align}\label{estim-V}
\|f(t)\| &\leq \int_s^t \|(\mathcal{L}_n(r)-\mathcal{L}_m(r))\gamma_{r,s}^{(m)}(A) \|dr  \nonumber \\
&\leq \sum_{z\in \Lambda_n\setminus\Lambda_m}\sum_{Z\ni z}\int_s^t \|\Psi_Z(r)(\gamma_{r,s}^{(m)}(A))\|dr.\nonumber
\end{align}
Using the Lieb-Robinson bound and the exponential decay condition (\ref{Decay_L}),
which we assumed holds uniformly in $\Lambda$,  we find that
\begin{align*}
\|f(t)\|&\leq\frac{\|A\|}{C_\mu}\int_s^t e^{\mu v_{r,\mu} (r-s)}
\sum_{z\in \Lambda_n\setminus\Lambda_m}\sum_{Z\ni z}\|\Psi_Z(r)\|_{\rm cb}
\sum_{x\in X}\sum_{y\in Z}F(d(x,y))dr\\
&\leq \frac{\|A\|}{C_\mu}\int_s^t e^{\mu v_{r,\mu} (r-s)} \sum_{z\in \Lambda_n\setminus\Lambda_m}\sum_{x\in X}\sum_{y\in\Gamma}
\sum_{Z\ni z,y}\|\Psi_Z(r)\|_{\rm cb}F(d(x,y))dr\\
&\leq \frac{\|A\|}{C_\mu}\|\Psi\|_{t,\mu}\int_s^t e^{\mu v_{r,\mu} (r-s)}dr \sum_{z\in \Lambda_n\setminus\Lambda_m}\sum_{x\in X}
\sum_{y\in\Gamma}  F(d(x,y))F(d(y,z))\\
&\leq {\|A\|}\|\Psi\|_{t,\mu}\int_s^t e^{\mu v_{r,\mu} (r-s)}dr \sum_{z\in \Lambda_n\setminus\Lambda_m}
\sum_{x\in X}F(d(x,z))\\
&\leq  {\|A\|}\|\Psi\|_{t,\mu}\int_s^t e^{\mu v_{r,\mu} (r-s)}dr |X| \sup_{x\in X}\sum_{z\in
\Lambda_n\setminus\Lambda_m}F(d(x,z)).
\end{align*}
Since $F$ is uniformly integrable (\ref{F_i}) the tail sum above goes to zero as $n,m\rightarrow\infty$. Thus
\begin{equation*}
\|(\gamma_{t,s}^{(n)}-\gamma_{t,s}^{(m)})(A)\|\rightarrow 0, \textit{ as } n,m\rightarrow\infty.
\end{equation*}
Therefore the sequence $\{\gamma_{t,s}^{(n)}(A)\}_{n=0}^\infty$ is Cauchy and hence convergent. Denote the
limit, and its extension to $\cA_\Gamma$, as $\gamma_{t,s}^\Gamma$.

To show that $\gamma_{t,s}^\Gamma$ is strongly continuous we notice that for
$0\leq s \leq t, r\leq T$, and any $A\in\cA_\Gamma^{\rm loc}$, we have
\begin{equation*}
\|\gamma_{t,s}^\Gamma (A)-\gamma_{r,s}^\Gamma(A)\|\leq \|\gamma_{t,s}^\Gamma(A)-\gamma_{t,s}^{(n)}(A)\|+
\|\gamma_{t,s}^{(n)}(A)-\gamma_{r,s}^{(n)}(A)\|+\|\gamma_{r,s}^{(n)}(A)-\gamma_{r,s}^\Gamma(A)\|,
\end{equation*}
for any $n\in\bN$ such that $A\in\cA_{\Lambda_n}$.

The strong continuity then follows from the strong convergence of $\gamma_{t,s}^{(n)}$ to
$\gamma_{t.s}^\Gamma$, uniformly in $s\leq t\in [0,T]$, and the strong continuity of
$\gamma_{t,s}^{(n)}$ in $t$. The continuity of the extension of $\gamma^\Gamma_{t,s}$ to all
of $A\in\cA_\Gamma$ follows by the standard density argument. The argument for continuity in
the second variable, $s$, is similar.
\boxendproof

\newchapter{Non-equilibrium state}{Non-equilibrium state of a leaking photon cavity }{Non-equilibrium state of a leaking photon cavity pumped by a random atomic beam}
\label{sec:NELS}

\section{Set up}
\label{sec:Nonequil_setup}
\subsection{Description of the model}

\noindent  Our model consists of a beam of two-level atoms that passe one-by-one
the microwave cavity. Atoms in a beam are randomly excited. During the passage time $\tau$ the corresponding single atom is able to interact with the cavity. For simplicity we consider a so called \emph{tuned} case, when the cavity size is equal to the interatomic distance, so there is always one atom in the cavity.

The cavity is a one-mode resonator described by quantum harmonic oscillator with the Hamiltonian $H_C = \epsilon \,
b^*b\otimes\idty$ in the Hilbert space $\cH_C$, where $b^*$ and $b$ stay for boson (photon) creation and annihilation
operators with canonical commutation relations (CCR): $[b, b^*]=1$, $[b, b]=[b^*, b^*]=0$.

The beam of two-level, $\{E,0\}$, atoms can be described as a chain $H_{A}= \sum_{n\geq1} H_{A_n}$ of individual atoms with
Hamiltonian $H_{A_n}=\idty \otimes E \, a_n^*a_n$ in the Hilbert space $\cH_A=\otimes_{n\geq 1}\cH_{A_n}$. Here
$a_n^*$ and $a_n$  are one-point fermion creation and annihilation operators that satisfy the canonical
anti-commutation relations (CAR): $\{a_n,a_n^*\}=\idty $ and $\{a_n,a_n\}= \{a_n^*,a_n^*\}=0$, but commute for
different indices $[a_n^\sharp,a_k^\sharp] =0$ for $n\neq k$.

In our model there in only one atom present in the cavity at any given moment.  The repeated cavity-atom interaction is time-depended and we take it in the form:
\begin{equation}\label{W-int}
W_n(t)=\chi_{[(n-1)\tau, n\tau)}(t) \, \lambda \, (b^*+b)\otimes a_n^*a_n \ .
\end{equation}
Here $\chi_{I}(x)$ is characteristic function of the set $I$. 

The Hamiltonian of the whole system in the space $\cH_S:= \cH_C \otimes \cH_A$, is then the sum of the Hamiltonian of the cavity,
of the atoms, and the interaction between them
\begin{align}\label{Ham-Model}
H(t)=& \ H_C+\sum_{n\geq 1}(H_{A_n}+W_n(t))\\
=& \ \epsilon \, b^*b\otimes\idty+\sum_{n\geq1}\idty\otimes E \, a_n^*a_n + \sum_{n\geq 1}\chi_{[(n-1)\tau, n\tau)}(t)
(\lambda \, (b^*+b)\otimes a_n^*a_n) \ . \nonumber
\end{align}
Notice that for the time $t\in[(n-1)\tau, n\tau)$, only the $n$-th atom interacts with the cavity and the Hamiltonian in time-independent.

\subsection{Hamiltonian dynamics of perfect cavity}

 Let $t\in[(n-1)\tau, n\tau)$. Then the Hamiltonian
(\ref{Ham-Model}) for the $n$-th atom in the cavity takes the form
\begin{equation}\label{Ham-n}
H_n:=\epsilon \, b^*b\otimes \idty+\idty\otimes E \, a_n^*a_n + \lambda \, (b^*+b)\otimes a_n^*a_n \ .
\end{equation}

Although the atomic beam is infinite, on the (quasi-local) operator algebra of
observables $\mathfrak{A}(\cH_C \otimes\cH_A)$ we can describe our system by normal states
$\omega_{S}(\cdot):= {\rm{Tr}}_{\cH_C \otimes\cH_A} (\, \cdot \ \rho_S)$, which are defined by trace-class density matrices $\rho_S \in \mathcal{C}_{1}(\cH_C \otimes \cH_A)$, where $\mathcal{C}_1$ denotes the space of trace-class operators.

We suppose that initially our system is in a product state:
\begin{equation*}
\omega_{S}(\cdot):= (\omega_{C}\otimes \omega_{A})(\cdot)= {\rm{Tr}}_{\cH_C \otimes\cH_A} (\, \cdot \ \rho_C \otimes \rho_A) \ , \
\rho_S = \rho_C \otimes \rho_A \ ,
\end{equation*}
where $\rho_C\in\mathcal{C}_{1}(\cH_C)$ and $\rho_A\in\mathcal{C}_{1}(\cH_A)$. We denote by $\omega_{S}^{t}(\cdot):= {\rm{Tr}}_{\cH_C \otimes\cH_A} (\, \cdot \ \rho_S (t))$ its time evolution. We also often refer to density matrices as states, if this wording will not produce any confusion.

For any state $\rho_S$ on $\mathfrak{A}(\cH_C\otimes\cH_A)$ the Hamiltonian dynamics of the system
is defined by
(\ref{Ham-Model}), or by the quantum time-dependent Liouvillian generator:
\begin{equation}\label{L-Gen-t}
L(t)(\rho_S):= -i \, [H(t),\rho_S] \ .
\end{equation}
Then the state $\rho_S(t)$ of the total system at the time $t$ is a solution of the Cauchy problem corresponding to Liouville differential equation
\begin{equation}\label{Liouville's-DE}
\frac{d}{dt}\rho_S(t) = L(t)(\rho_S(t)) \ , \ \rho_S(t=0) := \rho_C \otimes \rho_A \ .
\end{equation}

Notice that in general the Hamiltonian evolution (\ref{Liouville's-DE}) with time-dependent generator
is a family of automorphisms $\{\Gamma_{t,s}\}_{0\leq s \leq t}$  with the
\textit{cocycle property}:
\begin{equation}\label{cocycle}
\rho_S(t) = \Gamma_{t,s} \, \rho_S(s) \ .
\end{equation}

For our model with tuned repeated interactions the solution of (\ref{Liouville's-DE}) and the form of the evolution
operator (\ref{cocycle}) are considerably simplified. Indeed, the generator of the dynamics of our system (\ref{Ham-Model}) is time-independent for each
interval $[(n-1)\tau, n\tau)$. Therefore, by virtue of (\ref{Ham-n}) and (\ref{L-Gen-t}), the Liouvillian generators
\begin{equation}\label{L-Gen-n}
L_n := L(t)\ ,  \ t\in[(n-1)\tau, n\tau) \ ,  \  n \geq 1 \ ,
\end{equation}
are time-independent and commuting. Note that any moment $t \geq 0$ has the representation
\begin{equation}\label{t}
t:=n(t)\tau + \nu(t) \ , \  n(t) = [t/\tau] \ \ {\rm{and}} \ \ \nu(t)\in[0, \tau) \ ,
\end{equation}
where $[x]$ denotes the integer part of the number $x \in \mathbb{R}$. Then solution of (\ref{Liouville's-DE}) for
$t\in [(n-1)\tau, n\tau)$ gets the form:
\begin{equation}\label{Sol-Liouv-Eq}
\rho_S(t)=\Gamma_{t,s=0}(\rho_C \otimes \rho_A)= e^{\nu L_n}e^{\tau L_{n-1}} \, ... \ e^{\tau L_2}e^{\tau L_1}(\rho_C \otimes \rho_A)
\ .
\end{equation}

In the next section \ref{sec:Nonequil_nonleaking} we exploit a specific structure of the Hamiltonian dynamics (\ref{Sol-Liouv-Eq}) and a
special form of interaction (\ref{W-int}) to work out an effective evolution of the perfect cavity $H_C$. Our results
concern first of all the evolution of the photon number expectation in the time-dependent cavity state
\begin{equation}\label{C-state-t}
\rho_C(t):= \Tr_{\cH_{A}}\rho_S(t) \ .
\end{equation}
\begin{assumption} \label{Init-State-remark}

\begin{itemize}
\item We suppose that initial state $\rho_S(t=0):= \rho_S$ of the system is in the product state of the form:
\begin{equation}\label{In-State}
\rho_S: =\rho_C\otimes\bigotimes_{k\geq 1}\rho_{k} \ .
\end{equation}
Here $\rho_C$ is the initial state of the cavity and $\rho_A:= \bigotimes_{k\geq 1}\rho_{k}$ is that of the atomic beam.
\item We also assume that the states $\{\rho_{k}\}_{k\geq1}$ on the algebras $\{\mathfrak{A}(\cH_{A_k})\}_{k\geq1}$ are gauge-invariant,
i.e. $[\rho_{k}, a_k^*a_k] =0$. Each individual state is a function $\rho_{k} := f_k (a_k^*a_k)$ and we also suppose that $f_k = f$, i.e. the initial beam state $\rho_A$ is homogeneous.
\end{itemize}
\end{assumption}
Then by (\ref{Ham-Model}),(\ref{Ham-n}) one gets that in our model these atomic states do
not evolve:
\begin{equation}\label{commut-atoms}
[\idty \otimes \rho_{k}, H(t)] = 0 \ , \  k \geq 1 \ .
\end{equation}


This assumption together with (\ref{Sol-Liouv-Eq}),(\ref{C-state-t}), and  (\ref{In-State}) implies a discrete evolution for the cavity state given by the following recursive formula:
\begin{align}\label{C-state-n}
 \rho_C(t=n\tau)=:\rho_C^{(n)}=& \Tr_{\cH_{A}}\rho_S(n\tau)=\Tr_{\cH_A}[e^{\tau L_n}...e^{\tau L_2}e^{\tau
L_1}(\rho_C\otimes\bigotimes_{k=1}^n\rho_{k})] \\
=& \Tr_{\cH_{A_n}}[e^{\tau L_n}\{\Tr_{\cH_{A_{n-1}}}...\Tr_{\cH_{A_1}}e^{\tau L_{n-1}}...e^{\tau L_2}e^{\tau
L_1}(\rho_C\otimes\bigotimes_{k=1}^{n-1}\rho_{k})\}\otimes\rho_{n}] \nonumber \\
=& \Tr_{\cH_{A_n}}[e^{\tau L_n}(\rho_C^{(n-1)}\otimes\rho_{n})] \ . \nonumber
\end{align}
For any density matrix $\rho \in \mathcal{C}_{1}(\cH_C)$ corresponding to normal state on
the operator algebra $\mathfrak{A}(\cH_C)$ we define the mapping $\cL: \rho \mapsto \cL(\rho)$, by
\begin{equation}\label{cL}
\cL(\rho):=\Tr_{\cH_{A_n}}(e^{\tau L_n}(\rho\otimes\rho_{n})) = \Tr_{\cH_{A_n}}[e^{-i\tau H_n}(\rho\otimes\rho_{n})
e^{i\tau H_n}] \ .
\end{equation}
Here the last equality is due to (\ref{L-Gen-t}) and (\ref{L-Gen-n}).

Note that the mapping (\ref{cL}) does not depend on $n \geq 1$, since the atomic states $\{\rho_{n}\}_{n\geq 1}$ are
supposed to be homogeneous. Then the cavity state at $t=n\tau$ is defined by the $n$-th power of $\cL$:
\begin{equation}\label{cavity state at t=n}
\rho_C^{(n)}=\cL(\rho_C^{(n-1)})=\cL^n(\rho_C) \ .
\end{equation}
Therefore, by (\ref{Sol-Liouv-Eq}), (\ref{C-state-n}) and (\ref{cavity state at t=n}) one obtains that for any time $t=n\tau+\nu$, where
$\nu \in[0, \tau)$, the cavity state is
\begin{equation}\label{cavity state at t}
\rho_C(t)=\Tr_{\cH_{A_{n+1}}}[e^{\nu L_{n+1}}(\cL^{n}(\rho_C)\otimes\rho_{n+1})] \ .
\end{equation}

We are interested in the expectation value $N(t)$ of the photon-number operator $\widehat{N}:=b^*b$ in the cavity
at the time $t$:
\begin{equation}\label{N(t)}
N(t): = \omega_{C}^{t}(\widehat{N})= \Tr_{\cH_C}(b^*b \ \rho_C(t)) \ .
\end{equation}
For  $t=n\tau$ the state of the cavity can be expressed using (\ref{cavity state at t=n}), which gives
\begin{equation}\label{mean-photon-number-n}
N(n\tau)=\Tr_{\cH_C}(b^*b \ \cL^n(\rho_C)) \ .
\end{equation}

\subsection{Quantum dynamics of leaking cavity}

For a more general situation we take into account a possible leakage of the cavity, where photons may leave the cavity at some non-zero rate $\sigma>0$. We consider this case in the framework of Kossakowski-Lindblad extension of the Hamiltonian Dynamics to irreversible Quantum Dynamics with time-dependent generator
\begin{equation}\label{Generator}
L_{\sigma}(t)(\rho_S)=-i[H(t),\rho_S]+\sigma b(\rho_S)b^*-\frac{\sigma}{2}\{b^*b,
\rho_S\} \ ,
\end{equation}
for any $\rho_S\in\mathcal{C}_1(\cH_C\otimes\cH_A)$.

Similar to (\ref{Liouville's-DE}) the evolution of the state is defined by a solution of the non-autonomous
Cauchy problem corresponding to time-dependent generator (\ref{Generator})
\begin{equation}\label{Liouville's-DE-sigma}
\frac{d}{dt}\rho_S(t) = L_\sigma(t)(\rho_S(t)) \ , \ \rho_S(t=0) := \rho_C \otimes \rho_A \ .
\end{equation}
 In general for a time-dependent generator $L_\sigma(t)$ the proof of existence of this solution is a
non-trivial problem, as in chapter \ref{sec:Thermodynamics_irrev}. 
In the present case of the tuned repeated interactions, when the Hamiltonian is time-independent for each interval $[(n-1)\tau, n\tau)$, the generator (\ref{Generator})
gets the form:
\begin{equation}\label{Generator-KL}
L_{\sigma,n}(\rho_C\otimes\rho_A):= -i[H_n,\rho_C\otimes\rho_A]+\sigma b(\rho_C\otimes\rho_A)b^*-\frac{\sigma}{2}\{b^*b,
\rho_C\otimes\rho_A\} \ .
\end{equation}

With restriction to the tuned case and with the help of (\ref{Generator-KL}) the solution of (\ref{Liouville's-DE-sigma}) for $t\in[(n-1)\tau, n\tau)$ becomes of the form:
\begin{equation}\label{Sol-Liouv-Eq-sigma}
\rho_S(t)=T^\sigma_{t,s=0}(\rho_C \otimes \rho_A)= e^{\nu L_{\sigma,n}}e^{\tau L_{\sigma, n-1}} \, ... \ e^{\tau L_{\sigma, 2}}e^{\tau L_{\sigma, 1}}(\rho_C \otimes \rho_A)
\ .
\end{equation}
Here $\{T^\sigma_{t,s}\}_{0\leq s\leq t}$ is a family of unit preserving completely positive maps with a cocycle property.
 
By duality with respect to the initial state
$\omega_{\mathcal{S}}^{0}(\cdot)=\Tr ( \cdot \ \rho_{\mathcal{C}} \otimes \rho_{\mathcal{A}})$ one can now
define the adjoint evolution mapping $\{(T_{t,0}^\sigma)^*\}_{t\geq 0} $ by relation
\begin{equation}\label{dual-sigma}
\omega_{\mathcal{S},\sigma}^{t}(A)= \Tr ( A \ T_{t,0}^{\sigma}(\rho_{\mathcal{C}} \otimes \rho_{\mathcal{A}}))=
\omega_{\mathcal{S}}^{0}((T_{t,0}^\sigma)^*(A)) \ ,
\end{equation}
for any $A \in \mathfrak{A}(\cH_{\mathcal{C}} \otimes \cH_{\mathcal{A}})$. Then we obtain for any $t = (n-1)\tau + \nu(t)$ that
\begin{equation}\label{S-state-evol-adj-sigma}
\omega_{\mathcal{S},\sigma}^{t}(\cdot)=
{\rm{Tr}}\, ((T_{t,0}^\sigma)^{\ast}(\cdot) \ \rho_{\mathcal{C}} \otimes \rho_{\mathcal{A}})
\ \ ,\ \
(T_{t,0}^\sigma)^{\ast} = \prod_{k=1}^{n-1}e^{\tau L_{\sigma,k}^\ast} \  e^{\nu(t) L_{\sigma,n}^\ast} \ .
\end{equation}
Here $\{L_{\sigma,k}^\ast\}_{k \geq 1}$ are generators, which are adjoint to (\ref{Generator-KL}).

 The discrete evolution operator (\ref{cL}) has to be modified for $\sigma > 0$ as follows:
\begin{definition}\label{cL-def-sigma} For any state $\rho$ on $\mathfrak{A}(\cH_C)$ we define the mapping\begin{equation}\label{cL-sigma}
\cL_{\sigma}(\rho):=\Tr_{\cH_{A_n}}(e^{\tau L_{\sigma,n}}(\rho\otimes\rho_{n})) \ .
\end{equation}
\end{definition}
The initial state of the cavity is $\rho_{C} = \rho_{C,\sigma}(t=0)$. Then similar to (\ref{cavity state at t}) we obtain for
$\rho_{C,\sigma}(t)$ at the moment $t=n\tau + \nu$, where $\nu \in[0, \tau)$,
\begin{equation}\label{cavity state at t-sigma}
\rho_{C,\sigma}(t)=\Tr_{\cH_{A_{n+1}}}[e^{\nu L_{\sigma,n+1}}(\cL_{\sigma}^{n}(\rho_C)\otimes\rho_{n+1})] \ .
\end{equation}

Now we define the functional 
\begin{equation}\label{C-state-t-sigma}
\omega_{C, \sigma}^{t}(\cdot): =\omega_{S, \sigma}^t(\cdot \ \otimes\idty)= \Tr_{\cH_{C}}(\, \cdot \ \rho_{C,\sigma}(t)).
\end{equation} To study the infinite-time limit for
the cavity state $\omega_{C,\sigma}(\cdot):=\lim_{t\rightarrow\infty} \omega_{C,\sigma}^{t}(\cdot)$ we consider the functional
\begin{equation}\label{Weyl-func-sigma}
\omega_{C,\sigma}(W(\alpha))=\lim_{t\rightarrow\infty} \omega_{C,\sigma}^{t}(W(\alpha)) \ ,
\end{equation}
generated by the Weyl operator on $\cH_C$:
\begin{equation*}
W(\alpha)=e^{\alpha b-\overline{\alpha}b^*} \ , \ \alpha\in\mathbb{C} \ .
\end{equation*}
Notice that convergence (\ref{Weyl-func-sigma}) for the family of the Weyl operators guarantees the weak-$\ast$ limit \cite{Bratteli}
 of the states $\omega_{C,\sigma}^{t}(\cdot)$ when $t\rightarrow\infty$.

\section{Number of photons in perfect cavity}
\label{sec:Nonequil_nonleaking}

First we consider the case of the perfect cavity, i.e. $\sigma=0$. Then for the discrete evolution operator one gets
$\cL_{\sigma=0} = \cL$, see (\ref{cL}) and (\ref{cL-sigma}).

Our first result concerns the expectation of the photon-number operator $\widehat{N}=b^*b$ in the cavity (\ref{N(t)}).
For $t=n\tau$ this expectation value takes the form (\ref{mean-photon-number-n}).

In the theorem below we suppose that the initial cavity state $\omega_{C}(\cdot)$ is also gauge invariant,
i.e. $e^{i \alpha \widehat{N}}\rho_C e^{- i \alpha \widehat{N}}= \rho_C$. For a homogeneous beam (Remark \ref{Init-State-remark})
the parameter $p:=\Tr_{\cH_{A_n}}(a_n^*a_n \ \rho_{n})$, which is independent of $n\geq 1$, denotes a probability that atom $A_n$ is in excited state.

\begin{theorem}\label{N of photons}
Let $\rho_C$ be a gauge invariant state. Then for $t=n\tau$ the expectation value (\ref{mean-photon-number-n}) of the photon number
in the cavity is
\begin{equation}\label{number of photons_Hamiltonian}
N(n\tau)=N(0)+ n \, p(1-p) \, \frac{2\lambda^2}{\epsilon^2} \, (1-\cos\epsilon\tau) +
p^2 \, \frac{2\lambda^2}{\epsilon^2}(1-\cos n\epsilon\tau)\ .
\end{equation}
\end{theorem}

If for the initial state $\rho_C$ one takes in the theorem the Gibbs state for photons at the inverse temperature $\beta$:
\begin{equation}\label{Gibbs-photons}
\rho_C^{\beta}= {e^{-\beta\epsilon b^*b}}/{\Tr_{\cH_C} e^{-\beta\epsilon b^*b}} \ ,
\end{equation}
then the number of photons (\ref{number of photons_Hamiltonian}) is
\begin{equation}\label{number of photons Gibbs}
N(t) = \frac{1}{e^{\beta\epsilon}-1}+n \, p(1-p) \ \frac{2\lambda^2}{\epsilon^2} \ (1-\cos\epsilon\tau)+
p^2 \, \frac{2\lambda^2}{\epsilon^2}\, (1-\cos n\epsilon\tau) \ .
\end{equation}

To prepare the proof of Theorem \ref{N of photons} we calculate first explicit expressions for $\cL(\rho)$ acting on the
space $\rho \in \mathcal{C}(\cH_C)$ of the cavity states and for the $n$-th power of the dual operator $\cL^{*}$
that we apply to the number operator $b^*b \in \mathfrak{A}(\cH_C)$, see (\ref{cL-dual-0}) and Remark \ref{dual-operator}.

We split these calculations into two lemmas. The proof of the Theorem \ref{N of photons} is presented at the end of this
section.

The Hamiltonian (\ref{Ham-n}) can be written in the following form
\begin{align*}
H_n=&\epsilon \ (b^*+\frac{\lambda}{\epsilon}a_n^{*}a_n)(b+\frac{\lambda}{\epsilon}a_n^{*}a_n)+Ea_n^*a_n
- \frac{\lambda^2}{\epsilon}a_n^{*}a_n\\
=&\epsilon \ {\hat{b}^*}\hat{b}+ (E-{\lambda^2}/{\epsilon})a_n^*a_n \ .
\end{align*}
Here $\hat{b}:=b+{\lambda}a_n^{*}a_n/{\epsilon}$ and $\hat{b}^*:=b^*+{\lambda}a_n^{*}a_n/{\epsilon}$ are new boson operators
since
by the CAR and CCR properties of $a_n, a_n^{*}$ and respectively $b, b^*$ one gets: $[\hat{b}, {\hat{b}}^*]=[b, b^*]=1$ and
$[\hat{b}, {\hat{b}}] = [\hat{b}^*, {\hat{b}}^*]=0$ for \textit{any} index $n$.

Let $\widehat{S}_{k}$ be $*$-isomorphism (\textit{unitary shift}) on the algebra
$\mathfrak{A}(\cH_{\mathcal{C}})\otimes \mathfrak{A}(\cH_{\mathcal{A}_{k}})$ defined by
\begin{equation}\label{Vn}
\widehat{S}_{k}(\cdot):= e^{iV_k}(\cdot) \, e^{-iV_k}  \ , \ V_k:={\lambda}(b^*-b)\otimes \eta_k /{i\epsilon} \ .
\end{equation}

From the definition of $V_n$ we obtain that
\begin{equation*}
\widehat{S}_{k}(a_n^*a_n)=a_n^*a_n.
\end{equation*}

Then the transformed Hamiltonian is easily calculated to be given by
\begin{equation}\label{Diagonalized Hamiltonian}
\widehat{S}_{n}(H_n)=\epsilon \ b^*b+(E-\frac{\lambda^2}{\epsilon})a_n^*a_n.
\end{equation}

To compute $\widehat{S}_{n}(b)$, we note that if
\begin{equation*}
F(\nu):=e^{\nu b^*-\nu b}be^{-(\nu b^*-\nu b)} \ ,
\end{equation*}
then
$$\frac{dF(\nu)}{d\nu}=e^{\nu b^*}[b^*, b]e^{-\nu b^*}=-1 .$$
Therefore
\begin{equation}\label{shift b}
F(\nu)=F(0)-\nu .
\end{equation}

Applying this formula to $\tb$ we find that
\begin{equation}\label{tb}
\widehat{S}_{n}(b)=b\otimes\idty-\idty\otimes\frac{\lambda}{\epsilon}a_n^{*}a_n ,
\text{ and }
\widehat{S}_{n}(b^*)= b^*\otimes\idty-\idty\otimes\frac{\lb}{\e}a_n^*a_n .
\end{equation}

Note that the dynamics generated by $\widehat{S}_{n}(H_n)$ (or by ${H}_n$) leaves the atomic number operator $a_n^*a_n$
invariant
$$e^{i\tau\widehat{S}_{n}(H_n)}a_n^*a_ne^{-i\tau\widehat{S}_{n}(H_n)}=a_n^*a_n. $$
\begin{lemma} \label{cL(rho)} For any state $\rho$ on $\cH_C$
\begin{eqnarray}\label{cL(rho)_p}
\cL(\rho)&=&pe^{-{\lambda}(b^*-b)/{\epsilon}}e^{-i\tau\epsilon b^*b}e^{{\lambda}(b^*-b)/{\epsilon}}\rho
e^{-{\lambda}(b^*-b)/{\epsilon}}e^{i\tau\epsilon b^*b} e^{{\lambda}(b^*-b)/{\epsilon}}\\
&+&(1-p)e^{-i\tau\epsilon b^*b}\rho e^{i\tau\epsilon b^*b} \ .\nonumber
\end{eqnarray}
\end{lemma}
\begin{proof}
Using the transformation (\ref{Vn}) $\cL$ can be written in the form
 \begin{equation}\label{cL(rho)1}
\cL(\rho)=\Tr_{\cH_{A_n}}[\widehat{S}_{n}^*(e^{-i\tau \widehat{S}_{n}(H_n)}\widehat{S}_{n}(\rho\otimes\rho_{n})e^{i\tau
\widehat{S}_{n}(H_n)})].
\end{equation}
{From} definition of $V_n$ and the fact that $a_n^*a_n$ commutes with $\rho_{A_n}$ we have
\begin{align}\label{V_n}
\widehat{S}_{n}(\rho\otimes\rho_{n})=e^{{\lambda}(b^*-b)/{\epsilon}}\rho e^{-{\lambda}(b^*-b)/{\epsilon}}\otimes
a^*_na_n\rho_{n}+\rho\otimes(1-a_n^*a_n)\rho_{n}.
\end{align}
Therefore, plugging this expression into (\ref{cL(rho)1}) we obtain
\begin{align*}
\cL(\rho)=&\Tr_{\cH_{A_n}}[\widehat{S}_{n}^*e^{-i\tau \widehat{S}_{n}(H_n)}(e^{{\lambda}(b^*-b)/{\epsilon}}\rho
e^{-{\lambda}(b^*-b)/{\epsilon}}\otimes a^*_na_n\rho_{n})e^{i\tau \widehat{S}_{n}(H_n)}]\\
&+\Tr_{\cH_{A_n}}[\widehat{S}_{n}^*e^{-i\tau \widehat{S}_{n}(H_n)}(\rho\otimes(1- a_n^*a_n)\rho_{n})e^{i\tau \widehat{S}_{n}(H_n)}].
\end{align*}
From the diagonal form (\ref{Diagonalized Hamiltonian}) of the Hamiltonian, we have
\begin{align*}
\cL(\rho)=&\Tr_{\cH_{A_n}}[e^{-{\lambda}(b^*-b)/{\epsilon}}e^{-i\tau\epsilon b^*b}e^{{\lambda}(b^*-b)/{\epsilon}}\rho
e^{-{\lambda}(b^*-b)/{\epsilon}}e^{i\tau\epsilon b^*b}e^{{\lambda}(b^*-b)/{\epsilon}}\otimes a^*_na_n\rho_{n}]\\
&+\Tr_{\cH_{A_n}}[e^{-i\tau\epsilon b^*b}e^{{\lambda}(b^*-b)/{\epsilon}}\rho
e^{-{\lambda}(b^*-b)/{\epsilon}}e^{i\tau\epsilon b^*b}\otimes (1-a_n^*a_n)a^*_na_n\rho_{n}]\\
&+\Tr_{\cH_{A_n}}[e^{-{\lambda}(b^*-b)/{\epsilon}}e^{-i\tau\epsilon b^*b}\rho e^{i\tau\epsilon
b^*b}e^{{\lambda}(b^*-b)/{\epsilon}}\otimes a_n^*a_n(1-a_n^*a_n)\rho_{n}]\\
&+\Tr_{\cH_{A_n}}[e^{-i\tau\epsilon b^*b}\rho e^{i\tau\epsilon b^*b}\otimes(1-a_n^*a_n)\rho_{n}].
\end{align*}
Since $(1-a_n^*a_n)a_n^*a_n=0$,  we get
\begin{align*}
\cL(\rho)=& \ \Tr_{\cH_{A_n}}[e^{-{\lambda}(b^*-b)/{\epsilon}}e^{-i\tau\epsilon b^*b}e^{{\lambda}(b^*-b)/{\epsilon}}\rho
e^{-{\lambda}(b^*-b)/{\epsilon}}e^{i\tau\epsilon b^*b}e^{{\lambda}(b^*-b)/{\epsilon}}\otimes a^*_na_n\rho_{n}]\\
+& \ \Tr_{\cH_{A_n}}[e^{-i\tau\epsilon b^*b}\rho e^{i\tau\epsilon b^*b}\otimes(1-a_n^*a_n)\rho_{n}]\\
=& \ p \ e^{-{\lambda}(b^*-b)/{\epsilon}}e^{-i\tau\epsilon b^*b}e^{{\lambda}(b^*-b)/{\epsilon}}\rho
e^{-{\lambda}(b^*-b)/{\epsilon}}e^{i\tau\epsilon b^*b} e^{{\lambda}(b^*-b)/{\epsilon}}\\
+& \ (1-p) \ e^{-i\tau\epsilon b^*b}\rho e^{i\tau\epsilon b^*b} \ .
\end{align*}
\end{proof}

To calculate the expectation of the photon-number operator
$\widehat{N} = b^*b$ at $t=n\tau$,
we would need to find the action of the $n$-th power $\cL^n(\rho)$ of the operator (\ref{cL}). But in fact it is  easier to calculate
this mean value using the $n$-th power of the dual operator $\cL^{*}$, which for any bounded operator $B \in \cB(\cH_C)$ and
$\rho \in \mathcal{C}_{1}(\cH_C)$, is defined by relation
\begin{equation}\label{cL-dual-0}
\Tr_{\cH_{C}}(\cL^{*}(B) \rho): = \Tr_{\cH_{C}}( B \ \cL(\rho)) \ . \
\end{equation}

Using the property (\ref{commut-atoms}) one can obtain explicit an expression for the operator $\cL^{*}$.
Indeed, by (\ref{cL-dual-0}) and by (\ref{cL})
\begin{align}\label{cL-dual-1}
\Tr_{\cH_{C}}(B \ \cL(\rho))=& \Tr_{\cH_{C}\otimes\cH_{A_n}}\{ (B\otimes \idty) \ e^{- i\tau H_n}(\rho\otimes \idty)
(\idty \otimes \rho_n)e^{i\tau H_n}\} \\
=& \Tr_{\cH_{C}\otimes\cH_{A_n}}\{(\idty \otimes \rho_n) e^{-i\tau H_n} (B\otimes \idty) e^{i\tau H_n}(\rho\otimes \idty)\} \nonumber \\
=& \Tr_{\cH_{C}\otimes\cH_{A_n}}\{e^{i\tau H_n}(B \otimes \rho_n) \ e^{-i\tau H_n}(\rho\otimes \idty)\} \nonumber \ ,
\end{align}
where we used the cyclicity of the trace $\Tr_{\cH_{C}\otimes\cH_{A_n}}$ and the commutator (\ref{commut-atoms}). Then (\ref{cL-dual-1}) yields
\begin{equation}\label{cL-dual-2}
\cL^{*}(B)=\Tr_{\cH_{A_n}}\{e^{i\tau H_n}(B\otimes \rho_n)e^{-i\tau H_n}\} \ ,
\end{equation}
which according to (\ref{cL}), is independent of $n\geq 1$.

\begin{remark} \label{dual-operator} For density matrix $\rho \in \mathcal{C}_{1}(\cH_C)$ such that $\Tr_{\cH_{C}}( \mathcal{P}(b,b^*) \, \cL(\rho)) < \infty$
for any polynomial $\mathcal{P}(b,b^*)$, one can extendthe  definition (\ref{cL-dual-0}) of $\cL^{*}$ to this class of unbounded observables
$\mathfrak{A}(\cH_C)$. The advantage of using the dual operator $\cL^{*}$ is that its consecutive application does not increase the
degree of the polynomials generated by operators $b$ and $b^*$.
\end{remark}
Now following the line of reasoning of Lemma \ref{cL(rho)} we deduce from (\ref{cL-dual-2}) that
\begin{align}\label{dual cL}
\cL^{*}(B)= & \ pe^{-{\lambda}(b^*-b)/{\epsilon}}e^{i\tau\epsilon b^*b}e^{{\lambda}(b^*-b)/{\epsilon}}B
e^{-{\lambda}(b^*-b)/{\epsilon}}e^{-i\tau\epsilon b^*b} e^{{\lambda}(b^*-b)/{\epsilon}} \\
+ & \ (1-p)e^{-i\tau\epsilon b^*b}B e^{i\tau\epsilon b^*b} \ , \nonumber
\end{align}
for $p={\rm{Tr}}_{\cH_{A_n}} (a_n^*a_n \ \rho_n)$ independent of $n$.
\begin{lemma}\label{cL-T(B)}
For $B=b^*b$ and for the dual operator $\cL^{*}$ defined by (\ref{dual cL}) we obtain:
\begin{align}\label{L^n(b^*b)}
(\cL^{*})^n(b^*b)=& \ b^*b+p \ \frac{\lambda}{\epsilon}[(1-e^{ni\epsilon\tau})\ b^*+(1-e^{-ni\epsilon\tau})\ b] \\
+& \ n \, p (1-p)\ \frac{2\lambda^2}{\epsilon^2} (1-\cos\epsilon\tau)+
p^2 \ \frac{2\lambda^2}{\epsilon^2}(1-\cos n\epsilon\tau)  \nonumber \ .
\end{align}
\end{lemma}
\begin{proof}
From (\ref{shift b}) one gets
\begin{equation*}
e^{{\lambda}(b^*-b)/{\epsilon}}b^{\#}e^{-{\lambda}(b^*-b)/{\epsilon}}=b^{\#}-{\lambda}/{\epsilon}
\end{equation*}
and so
\begin{equation*}
e^{{\lambda}(b^*-b)/{\epsilon}}b^*be^{-{\lambda}(b^*-b)/{\epsilon}}=(b^*-{\lambda}/{\epsilon})
(b-{\lambda}/{\epsilon}).
\end{equation*}
From the CCR properties of $b$ and $b^*$ we find that their evolution is
\begin{align*}
e^{i\tau\epsilon b^*b}be^{-i\tau\epsilon b^*b}&=e^{-i\tau\epsilon}b \\
e^{i\tau\epsilon b^*b}b^*e^{-i\tau\epsilon b^*b}&=e^{i\tau\epsilon}b^*.
\end{align*}
Therefore, making the shift to calculate the first term in (\ref{dual cL}) we get
\begin{eqnarray}\label{evolution-rho_shift}
&e^{-{\lambda}(b^*-b)/{\epsilon}}
(e^{i\tau\epsilon}b^*-{\lambda}/{\epsilon})(e^{-i\tau\epsilon}b-{\lambda}/{\epsilon})e^{{\lambda}(b^*-b)/{\epsilon}}\\
&=(e^{i\tau\epsilon}b^*-(1-e^{i\tau\epsilon}){\lambda}/{\epsilon})
(e^{-i\tau\epsilon}b-(1-e^{-i\tau\epsilon}){\lambda}/{\epsilon}).\nonumber
\end{eqnarray}
Hence, (\ref{dual cL}) and (\ref{evolution-rho_shift}) imply
\begin{align}\label{cL(b^*b)}
\cL^{*}(b^*b)=&pe^{-{\lambda}(b^*-b)/{\epsilon}}e^{i\tau\epsilon b^*b}e^{{\lambda}(b^*-b)/{\epsilon}}b^*b
e^{-{\lambda}(b^*-b)/{\epsilon}}e^{-i\tau\epsilon b^*b} e^{{\lambda}(b^*-b)/{\epsilon}}\\
&+ (1-p)e^{-i\tau\epsilon
b^*b}b^*b e^{i\tau\epsilon b^*b}\nonumber\\
=&p(e^{i\epsilon\tau}b^*-(1-e^{i\epsilon\tau}){\lambda}/{\epsilon})(e^{-i\epsilon\tau}b-(1-e^{-i\epsilon\tau})
{\lambda}/{\epsilon})+(1-p)b^*b\nonumber\\
=&b^*b+p\frac{\lambda}{\epsilon}(1-e^{i\epsilon\tau})b^*+p\frac{\lambda}{\epsilon}(1-e^{-i\epsilon\tau})b+
p\frac{2\lambda^2}{\epsilon^2}(1-\cos
\epsilon\tau) \ . \nonumber
\end{align}

If in (\ref{dual cL}) we put $B=b^*$, then one gets
\begin{equation}\label{cL(b^*)}
\cL^{*}(b^*)=e^{i\epsilon\tau}b^*-p(1-e^{i\epsilon\tau}){\lambda}/{\epsilon} \ ,
\end{equation}
and similarly, one obtains:
\begin{equation}\label{cL(b)}
\cL^{*}(b)=e^{-i\epsilon\tau}b-p(1-e^{-i\epsilon\tau}){\lambda}/{\epsilon} \ ,
\end{equation}
for $B=b$.

Now we are going to prove the $n$-th power formula (\ref{L^n(b^*b)}) by induction. Note that if we take $n=1$ in this formula we
get (\ref{cL(b^*b)}).

Suppose that (\ref{L^n(b^*b)}) is true for $n$ to check that it is also valid for $n+1$.
\begin{align*}
(\cL^{*})^{n+1}(b^*b)=&\cL^{*}((\cL^{*})^n(b^*b))\\
=&\cL^{*}(b^*b)+p\frac{\lambda}{\epsilon}(1-e^{ni\epsilon\tau})\cL^{*}(b^*)+
p\frac{\lambda}{\epsilon}(1-e^{-ni\epsilon\tau})\cL^{*}(b)\\
&+p\frac{2\lambda^2}{\epsilon^2}n(1-\cos\epsilon\tau)-p^2\frac{2\lambda^2}{\epsilon^2}n(1-\cos\epsilon\tau)+
p^2\frac{2\lambda^2}{\epsilon^2}(1-\cos n\epsilon\tau).
\end{align*}
Taking onto account (\ref{cL(b^*b)}), (\ref{cL(b^*)}) and (\ref{cL(b)}) we can express the action of operator $\cL^{*}$ as follows
\begin{align*}
(\cL^{*})^{n+1}(b^*b)=&b^*b+p\frac{\lambda}{\epsilon}(1-e^{i\epsilon\tau})b^*+p\frac{\lambda}
{\epsilon}(1-e^{-i\epsilon\tau})b+p\frac{2\lambda^2}{\epsilon^2}(1-\cos
\epsilon\tau)\\
&+p\frac{\lambda}{\epsilon}(1-e^{ni\epsilon\tau})e^{i\epsilon\tau}b^*-p^2\frac{\lambda^2}
{\epsilon^2}(1-e^{i\epsilon\tau})(1-e^{ni\epsilon\tau})\\
&+p\frac{\lambda}{\epsilon}(1-e^{-ni\epsilon\tau})e^{-i\epsilon\tau}b-p^2\frac{\lambda^2}
{\epsilon^2}(1-e^{-i\epsilon\tau})(1-e^{-ni\epsilon\tau})\\
&+p\frac{2\lambda^2}{\epsilon^2}n(1-\cos\epsilon\tau)-p^2\frac{2\lambda^2}{\epsilon^2}n
(1-\cos\epsilon\tau)+p^2\frac{2\lambda^2}{\epsilon^2}(1-\cos
n\epsilon\tau))\\
=&b^*b+p\frac{\lambda}{\epsilon}(1-e^{(n+1)i\epsilon\tau})b^*+p\frac{\lambda}{\epsilon}
(1-e^{-(n+1)i\epsilon\tau})b\\
&+p\frac{2\lambda^2}{\epsilon^2}(1-\cos\epsilon\tau)(n+1)-p^2\frac{\lambda^2}{\epsilon^2}(2-2\cos\epsilon\tau-2\cos n\epsilon\tau\\
&+2\cos(n+1)\epsilon\tau+2n-2n\cos\epsilon\tau-2-
2\cos n\epsilon\tau).
\end{align*}
Simplifying the last expression we get
\begin{align*}
(\cL^{*})^{n+1}(b^*b)=&b^*b+p\frac{\lambda}{\epsilon}(1-e^{(n+1)i\epsilon\tau})b^*+p\frac{\lambda}{\epsilon}
(1-e^{-(n+1)i\epsilon\tau})b\\
&+p\frac{2\lambda^2}{\epsilon^2}(1-\cos\epsilon\tau)(n+1)-p^2\frac{2\lambda^2}{\epsilon^2}(n+1)(1-\cos\epsilon\tau)\\
&+p^2\frac{2\lambda^2}{\epsilon^2}(1-\cos(n+1)\epsilon\tau),
\end{align*}
which proves (\ref{L^n(b^*b)}) formula.
\end{proof}

\begin{proof}(of Theorem \ref{N of photons}) To find the number of photons in the cavity at time $t=n\tau$ we take the
expectation in (\ref{L^n(b^*b)}). Since the initial state is gauge invariant we get
\begin{align*}
N(t)=&\Tr_{\cH_{C}}(b^*b\cL^n(\rho_C^\beta))=\Tr_{\cH_{C}}((\cL^{*})^n(b^*b)\rho_C^\beta)\\
=&\langle
b^*b\rangle+p\frac{2\lambda^2}{\epsilon^2}n(1-\cos\epsilon\tau)-p^2\frac{2\lambda^2}{\epsilon^2}n
(1-\cos\epsilon\tau)+p^2\frac{2\lambda^2}{\epsilon^2}(1-\cos
n\epsilon\tau)\\
=&N(0)+
n \, p (1-p)\ \frac{2\lambda^2}{\epsilon^2} \ (1-\cos\epsilon\tau) + p^2\frac{2\lambda^2}{\epsilon^2}(1-\cos
n\epsilon\tau) \ .
\end{align*}

\end{proof}

\begin{remark} Theorem implies that only flux of randomly exited atoms (i.e. $0<p<1$) is able to produce
a pumping of the cavity by photons. Since it is equivalent to the unlimited increasing of the cavity energy, the passage of
atoms may produce energy and entropy. We return to these observations below... .
\end{remark}

\section{Number of photons in leaking cavity}
\label{sec:Nonequil_leaking}

In the Hamiltonian case the number of photons in the cavity increases in time. To stabilize the energy in the cavity we
consider the case when the generator of the dynamics (\ref{Generator}) has a nonzero dissipative part, which describes the
leaking of photons out of the cavity. So suppose that $\sigma>0$.

The next theorem shows that the number of photons in the leaking cavity stabilizes in time for any
$\sigma > 0$. It is obviously different to the case $\sigma=0$, see  Theorem \ref{N of photons}.

\begin{theorem}\label{number of photons-sigma} For any $\sigma > 0$ and arbitrary gauge-invariant initial cavity state
$\rho_C$ such that
\begin{equation}\label{int-cond-photons-sigma}
\omega_{C,\sigma}^{t}(b^*b)\mid_{t=0} = {\rm{Tr}}_{\cH_{C}}(b^*b \ \rho_C) < \infty \ ,
\end{equation}
the number of photos in the cavity at time $t=\n\tau$ is
\begin{align}\label{mean-value}
&N_\sigma(n\tau)=\omega_{C,\sigma}^{\n\tau}(b^*b)=e^{-n(\sigma_--\sigma_+)\tau}N_\sigma(0)\\
&+p(1-p)\frac{2\lambda^2}{|\mu|^2}(1-e^{-(\sigma_{-} - \sigma_{+})\tau/2}\cos \epsilon\tau)
\frac{1-e^{-n(\sigma_--\sigma_+)\tau}}{1-e^{-(\sigma_--\sigma_+)\tau}}\nonumber\\
&+p^2\frac{2\lambda^2}{|\mu|^2}(1-e^{-n(\sigma_{-} - \sigma_{+})\tau/2}\cos n\epsilon\tau)
+\frac{\sigma_+}{\sigma_--\sigma_+}(1-e^{-n(\sigma_--\sigma_+)\tau})\nonumber \ .
\end{align}
Also one gets for the limit of the cavity photon-number mean value
\begin{eqnarray}\label{lim-photons-number-sigma}
&&N_{\sigma}^{\infty} =\omega_{C,\sigma}(b^*b) := \lim_{t\rightarrow\infty} \omega_{C}^{t}(b^*b) = \\
&& =  \frac{4\lambda^2}{4\epsilon^2+\sigma^2}   \frac{p}{1-e^{-\sigma\tau}}\{{1+e^{-\sigma\tau}(1 - 2p) -2 e^{-\sigma \tau/2}(1-p)
\cos\epsilon\tau}\}
\ .\nonumber
\end{eqnarray}
Here  $p:={\rm{Tr}}_{\cH_{A_n}}(a_n^*a_n \ \rho_{n})$.
\end{theorem}

\begin{lemma}\label{L-sigma}
For any state $\rho$ on algebra $\mathfrak{A}(\cH_C)$ one has:
\begin{equation*}
\cL_{\sigma}(\rho)=pe^{-{\lambda}(b^*-b)/{\epsilon}}e^{\tau L^C_\lambda}(e^{{\lambda}(b^*-b)/{\epsilon}}\rho
e^{-{\lambda}(b^*-b)/{\epsilon}})e^{{\lambda}(b^*-b)/{\epsilon}} +(1-p)e^{\tau L^C_0}(\rho),
\end{equation*}
where $L_\lambda^C$ acts on $\mathfrak{A}(\cH_C)$ as follows
\begin{equation}\label{L^C_lambda}
L_\lambda^C(\rho):= -i[\epsilon b^*b, \rho]+\sigma (b-{\lambda}/{\epsilon})\rho (b^*-{\lambda}/{\epsilon})-\frac{\sigma
}{2}\{(b^*-{\lambda}/{\epsilon})(b-{\lambda}/{\epsilon}), \rho\}.
\end{equation}
\end{lemma}
\begin{proof} Using unitary transformation generated by (\ref{V}) of the Hamiltonian, we define instead of (\ref{Generator-KL}):
\begin{align} \label{tilde_L}
&\widetilde{L}_{\sigma,n}(\rho\otimes\rho_{n}):= \widehat{S}_{n}L_{\sigma,n}\widehat{S}_{n}^*(\rho\otimes\rho_{n}) \\
=&-i[\widetilde{H}_n, \rho\otimes\rho_{n}]+\sigma \widehat{S}_{n}(b)(\rho\otimes\rho_{n})\widehat{S}_{n}(b^*)-\frac{\sigma
}{2}\{\widehat{S}_{n}(b^*b), \rho\otimes\rho_{n}\} \nonumber \\
=&-i[\epsilon b^*b+(E-\frac{\lambda^2}{\epsilon})a_n^*a_n, \rho\otimes\rho_{n}]+ [\sigma (b-{\lambda}/{\epsilon})\rho
(b^*-{\lambda}/{\epsilon}) \nonumber \\
&-\frac{\sigma }{2}\{(b^*-{\lambda}/{\epsilon})(b-{\lambda}/{\epsilon}), \rho\}]\otimes a_n^*a_n\rho_{n}
+(\sigma b^*\rho b-\frac{\sigma}{2}\{b^*b, \rho\})\otimes(1-a_n^*a_n)\rho_{n} \nonumber \\
=& L_\lambda^C(\rho)\otimes a_n^*a_n\rho_{A_n}+L^C_0(\rho)\otimes(1-a_n^*a_n)\rho_{n} \nonumber \ ,
\end{align}
here $L_0^C := L_{\lambda=0}^C$, see (\ref{L^C_lambda}). Then discrete evolution operator (\ref{cL-sigma}) gets the form

\begin{align*}
\cL_{\sigma}(\rho)=&\Tr_{A_n} (e^{\tau
L_n}(\rho\otimes\rho_{n}))=\Tr_{A_n}(\widehat{S}_{n}^*e^{\tau\widetilde{L}_n}(\widehat{S}_{n}(\rho\otimes\rho_{n})))\\
=&\Tr_{A_n}(\widehat{S}_{n}^*e^{\tau\widetilde{L}_n}(e^{{\lambda}(b^*-b)/{\epsilon}}\rho e^{-{\lambda}(b^*-b)/{\epsilon}}\otimes
a^*_na_n\rho_{n}+\rho\otimes(1-a_n^*a_n)\rho_{n}))\\
=&\Tr_{A_n}(\widehat{S}_{n}^*(e^{\tau L^C_\lambda}(e^{{\lambda}(b^*-b)/{\epsilon}}\rho e^{-{\lambda}(b^*-b)/{\epsilon}})\otimes
a^*_na_n\rho_{A_n}))\\
&+\Tr_{A_n}(\widehat{S}_{n}^*(e^{\tau L^C_0}(\rho)\otimes(1-a_n^*a_n)\rho_{n})).\\
\end{align*}

Simplifying the last expression 
\begin{align*}
\cL_{\sigma}(\rho)=&\Tr_{A_n}(e^{-{\lambda}(b^*-b)/{\epsilon}}e^{\tau L^C_\lambda}(e^{{\lambda}(b^*-b)/{\epsilon}}\rho
e^{-{\lambda}(b^*-b)/{\epsilon}})e^{{\lambda}(b^*-b)/{\epsilon}}\otimes a^*_na_n\rho_{n})\\
&+\Tr_{A_n}(e^{\tau
L^C_0}(\rho)\otimes(1-a_n^*a_n)\rho_{n})\\
=&pe^{-{\lambda}(b^*-b)/{\epsilon}}e^{\tau L^C_\lambda}(e^{{\lambda}(b^*-b)/{\epsilon}}\rho
e^{-{\lambda}(b^*-b)/{\epsilon}})e^{{\lambda}(b^*-b)/{\epsilon}} +(1-p)e^{\tau L^C_0}(\rho).
\end{align*}
\end{proof}

Now we look for the corresponding dual operator $\cL_{\sigma}^{*}$. It can be calculated using transformation (\ref{tilde_L})
and and the dual operator (\ref{cL-dual-2}). Suppose that $\rho_n\in \mathcal{C}_{1}(\cH_{A_n})$ is such that it commutes with $a_n^*a_n$,
$\Tr_{A_n}(\rho_n)=1$  and $p=\Tr_{A_n}(a^*_na_n \rho_n)$. Then for any bounded operator $B\in\cB(\cH_C)$ one has:
\begin{align}\label{cL(B)}
\cL_{\sigma}^{*}(B)=&\Tr_{A_n} (e^{\tau L_n^{*}}(B\otimes \rho_n))  \\
=&pe^{-{\lambda}(b^*-b)/{\epsilon}}e^{\tau(L^C_\lambda)^{*}}e^{{\lambda}(b^*-b)/{\epsilon}}B
e^{-{\lambda}(b^*-b)/{\epsilon}}e^{{\lambda}(b^*-b)/{\epsilon}} +(1-p)e^{\tau(L^C_0)^{*}}(B) \nonumber .
\end{align}

Here the dual operator acts as follows
\begin{align}
(L^C_\lambda)^{*}(B)&=i[\epsilon b^*b, B]+\sigma (b^*-{\lambda}/{\epsilon})B(b-{\lambda}/{\epsilon})-\frac{\sigma
}{2}\{(b^*-{\lambda}/{\epsilon})(b-{\lambda}/{\epsilon}), B\}\label{L^C_lambda(B)}\\
(L^C_0)^{*}(B)&=(\widetilde{L}_C^0)^{*}(B)=i[\epsilon b^*b, B]+\sigma b^*Bb-\frac{\sigma }{2}\{b^*b, B\}\label{L^C_0(B)}
\end{align}
\begin{remark}\label{dual-operator-sigma}
As we indicated in Remark \ref{dual-operator}, one can extend dual operator $\cL_{\sigma}^{*}$ to the algebra of
polynomial observables $\mathcal{P}(b,b^*) \in \mathfrak{A}(\cH_C)$.
\end{remark}

To proof Theorem \ref{number of photons-sigma} similarly to the case $\sigma=0$, the number of photons for the time $t=n\tau$ can be calculated using the dual operator
\begin{equation}\label{n-number-sigma}
N_{\sigma}^{t=n\tau} := \omega_{C,\sigma}^t(b^*b)=\Tr_{\cH_C}(b^*b  \ \cL_{\sigma}^n(\rho_C)) = \Tr_{\cH_C}((\cL_{\sigma}^{*})^n(b^*b)\ \rho_C) \ ,
\end{equation}
where $\rho_C$ is a gauge-invariant state of the cavity, see Remark \ref{dual-operator-sigma}.

We look more closely of the dual operator $\cL_{\sigma}^{*}$. Note that if one takes $B=b^*b$ in (\ref{cL(B)}) we get
\begin{equation}\label{cL^t(b^*b)+s}
\cL_{\sigma}^{*}(b^*b)=pe^{-{\lambda}(b^*-b)/{\epsilon}}e^{\tau(L^C_\lambda)^{*}}((b^*-{\lambda}/{\epsilon})(b-\frac{\lambda}
{\epsilon}))e^{{\lambda}(b^*-b)/{\epsilon}}
+(1-p)e^{\tau(L^C_0)^{*}}(b^*b).
\end{equation}

\begin{lemma}\label{number-photon-sigma}
The action of the dual operator $\cL_{\sigma}^{*}$ on the number operator of photons can be calculated
explicitly
\begin{align}\label{action on b^*b}
\cL_{\sigma}^{*}(b^*b)=&e^{-\sigma\tau}b^*b+p\frac{i\lambda}{\mu}e^{-\sigma\tau}(1-e^{\mu\tau})b^*-p\frac{i\lambda}
{\bar{\mu}} e^{-\sigma\tau}(1-e^{\bar{\mu}\tau})b\\
&+p\frac{\lambda^2}{|\mu|^2}e^{-\sigma\tau}(1-e^{\mu\tau})
(1-e^{\bar{\mu}\tau}).\nonumber
\end{align}
\end{lemma}
\begin{proof} Look on the first term in (\ref{cL^t(b^*b)+s}). Let $\gamma_{\lambda,\tau}(B)=e^{\tau(L^C_\lambda)^{*}}(B)$, for
$B\in \mathfrak{A}(\mathcal{H}_C)$. Then to find
$\gamma_{\lambda,\tau}((b^*-{\lambda}/{\epsilon})(b-{\lambda}/{\epsilon}))$ we first note that
\begin{equation*}
(L^C_\lambda)^{*}((b^*-{\lambda}/{\epsilon})(b-{\lambda}/{\epsilon}))=i\lambda b-i\lambda
b^*-\sigma(b^*-{\lambda}/{\epsilon})(b-{\lambda}/{\epsilon})
\end{equation*}
and
\begin{align*}
(L^C_\lambda)^{*}(b)&=-i\epsilon b-\frac{\sigma}{2}b+\frac{\sigma \lambda}{2\epsilon}=-\mu b+\frac{\sigma \lambda}{2\epsilon}
\end{align*}
and
\begin{align*}
(L^C_\lambda)^{*}(b^*)&=i\epsilon b^*-\frac{\sigma}{2}b^*+\frac{\sigma \lambda}{2\epsilon}=-\overline{\mu}b^*+\frac{\sigma
\lambda}{2\epsilon},
\end{align*}
where $\mu=\frac{\sigma}{2}+i\epsilon.$

Therefore we have the following system of differential equations for $\gamma_{\lambda,\tau}$
\begin{align*}
&\frac{d\gamma_{\lambda,\tau}((b^*-{\lambda}/{\epsilon})(b-{\lambda}/{\epsilon})}{d\tau}=-\sigma\gamma_{\lambda,\tau}
((b^*-{\lambda}/{\epsilon})(b-{\lambda}/{\epsilon}))+i\lambda\gamma_{\lambda,\tau}(b)-i\lambda\gamma_{\lambda,\tau}(b^*)\\
&\frac{d\gamma_{\lambda,\tau}(b)}{d\tau}=-\mu\gamma_{\lambda,\tau}(b)+\frac{\lambda\sigma}{2\epsilon}\\
&\frac{d\gamma_{\lambda,\tau}(b^*)}{d\tau}=-\overline{\mu}\gamma_{\lambda,\tau}(b^*)+\frac{\lambda\sigma}{2\epsilon}.
\end{align*}

The solution to this system is
\begin{align}
&\gamma_{\lambda,\tau}(b)=e^{-\mu\tau}(b-\frac{\lambda\sigma}{2\epsilon\mu})+\frac{\lambda\sigma}{2\epsilon\mu}=
e^{-\mu\tau}b+\frac{\lambda\sigma}{2\epsilon\mu}(1-e^{-\mu\tau})\label{gamma_b}\\
&\gamma_{\lambda,\tau}(b^*)=e^{-\overline{\mu}\tau}(b^*-\frac{\lambda\sigma}{2\epsilon\overline{\mu}})+\frac{\lambda\sigma}
{2\epsilon\overline{\mu}}=e^{-\overline{\mu}\tau}b^*+\frac{\lambda\sigma}{2\epsilon\overline{\mu}}(1-e^{-\overline{\mu}\tau})
\label{gamma_b^*}\\
&\gamma_{\lambda,\tau}((b^*-{\lambda}/{\epsilon})(b-{\lambda}/{\epsilon}))=e^{-\sigma\tau}b^*b-\frac{\lambda}
{\epsilon}b^*e^{-\sigma\tau}(\frac{i\epsilon}{\mu}e^{\mu\tau}-\frac{i\epsilon}{\mu}+1)\nonumber\\
&+{\lambda}/{\epsilon}be^{-\sigma\tau}(\frac{i\epsilon}{\overline{\mu}}e^{\overline{\mu}\tau}-\frac{i\epsilon}
{\overline{\mu}}-1)+\frac{\lambda^2}{|\mu|^2}(1-e^{-\sigma\tau})+\frac{\lambda^2}{\epsilon^2}e^{-\sigma\tau}-
\frac{\lambda^2\sigma\sin\epsilon\tau}{\epsilon|\mu|^2}e^{-\frac{\sigma}{2}\tau}\nonumber.
\end{align}

Making the shift of $b$ and $b^*$ we get
\begin{align*}
&e^{-{\lambda}(b^*-b)/{\epsilon}}\gamma_{\lambda,\tau}((b^*-{\lambda}/{\epsilon})(b-{\lambda}/{\epsilon}))
e^{{\lambda}(b^*-b)/{\epsilon}}=e^{-\sigma\tau}(b^*+{\lambda}/{\epsilon})(b+{\lambda}/{\epsilon})\\
&-{\lambda}/{\epsilon}(b^*+{\lambda}/{\epsilon})e^{-\sigma\tau}(\frac{i\epsilon}{\mu}e^{\mu\tau}-\frac{i\epsilon}{\mu}+1)\\
&+{\lambda}/{\epsilon}(b+{\lambda}/{\epsilon})e^{-\sigma\tau}(\frac{i\epsilon}{\overline{\mu}}e^{\overline{\mu}\tau}-
\frac{i\epsilon}{\overline{\mu}}-1)+\frac{\lambda^2}{|\mu|^2}(1-e^{-\sigma\tau})+\frac{\lambda^2}{\epsilon^2}e^{-\sigma\tau}-
\frac{\lambda^2\sigma\sin\epsilon\tau}{\epsilon|\mu|^2}e^{-\frac{\sigma}{2}\tau}\\
&=e^{-\sigma\tau}b^*b+\frac{i\lambda}{\mu}e^{-\sigma\tau}(1-e^{i\mu\tau})b^*-\frac{i\lambda}{\overline{\mu}}e^{-\sigma\tau}
(1-e^{-\overline{\mu}\tau})b+\frac{\lambda^2}{|\mu|^2}(1+e^{-\sigma\tau}-2e^{-\frac{\sigma}{2}\tau}\cos\epsilon\tau).
\end{align*}

So
\begin{align*}
\cL_{\sigma}^{*}(b^*b)=&e^{-\sigma\tau}b^*b+p\frac{i\lambda}{\mu}e^{-\sigma\tau}(1-e^{\mu\tau})b^*-p\frac{i\lambda}{\bar{\mu}}
e^{-\sigma\tau}(1-e^{\bar{\mu}\tau})b\\
&+p\frac{\lambda^2}{|\mu|^2}e^{-\sigma\tau}(1-e^{\mu\tau})(1-e^{\bar{\mu}\tau}).
\end{align*}
\end{proof}

From (\ref{gamma_b}) and (\ref{gamma_b^*}) we get
\begin{align}\label{action on b^*}
\cL_{\sigma}^{*}(b^*)=e^{i\epsilon\tau}e^{-\frac{\sigma}{2}\tau}b^*+p\frac{2i\lambda}{\sigma-2i\epsilon}(1-e^{i\epsilon\tau}
e^{-\frac{\sigma}{2}\tau})=e^{-\bar{\mu}\tau}b^*+p\frac{i\lambda}{\bar{\mu}}(1-e^{-\bar{\mu}\tau})
\end{align}
and
\begin{align}\label{action on b}
\cL_{\sigma}^{*}(b)=e^{-i\epsilon\tau}e^{-\frac{\sigma}{2}\tau}b-p\frac{2i\lambda}{\sigma+2i\epsilon}(1-e^{-i\epsilon\tau}
e^{-\frac{\sigma}{2}\tau})=e^{-\mu\tau}b-p\frac{i\lambda}{\mu}(1-e^{-\mu\tau}).
\end{align}

\begin{proof}(of Theorem \ref{number of photons-sigma})
The following expression can be found by representing the operator $\cL_{\sigma}^{*}$ in the matrix form
\begin{align}
&(\cL_{\sigma}^{*})^n(b^*b)=e^{-n\sigma\tau}b^*b+p\frac{i\lambda}{\mu}(e^{-n\sigma\tau}-e^{-n\bar{\mu}\tau})b^*-
p\frac{i\lambda} {\bar{\mu}}(e^{-n\sigma\tau}-e^{-n\mu\tau})b\label{n-th
power on b^*b}\\
&+p\frac{\lambda^2}{|\mu|^2}e^{-\sigma\tau}(1-e^{\mu\tau})(1-e^{\bar{\mu}\tau})\frac{1-e^{-n\sigma\tau}}
{1-e^{-\sigma\tau}}-p^2\frac{2\lambda^2}{|\mu|^2}\frac{1-e^{-n\sigma\tau}}{1-e^{-\sigma\tau}}(1-e^{-\frac{\sigma}{2}
\tau}\cos\epsilon\tau)\nonumber\\
&+p^2\frac{2\lambda^2}{|\mu|^2}(1-e^{-n\frac{\sigma}{2}\tau}\cos n\epsilon\tau)\nonumber.
\end{align}

We prove this formula by induction.

Suppose (\ref{n-th power on b^*b}) is true for $n$, we show it is true for $n+1$.
\begin{align*}
&(\cL_{\sigma}^{*})^{n+1}(b^*b)\\
&=e^{-n\sigma\tau}\cL_{\sigma}^{*}(b^*b)+p\frac{i\lambda}{\mu}(e^{-n\sigma\tau}-
e^{-n\bar{\mu}\tau})\cL_{\sigma}^{*}(b^*)- p\frac{i\lambda}{\bar{\mu}}(e^{-n\sigma\tau}-e^{-n\mu\tau})\cL_{\sigma}^{*}(b)\\
&+p\frac{\lambda^2}{|\mu|^2}e^{-\sigma\tau}(1-e^{\mu\tau})(1-e^{\bar{\mu}\tau})\frac{1-e^{-n\sigma\tau}}{1-e^{-\sigma\tau}}-
p^2\frac{2\lambda^2}{|\mu|^2}\frac{1-e^{-n\sigma\tau}}{1-e^{-\sigma\tau}}(1-e^{-\frac{\sigma}{2}\tau}\cos\epsilon\tau)\\
&+p^2\frac{2\lambda^2}{|\mu|^2}(1-e^{-n\frac{\sigma}{2}\tau}\cos n\epsilon\tau).
\end{align*}

From the formulas (\ref{action on b^*b}), (\ref{action on b^*}) and (\ref{action on b}) we get
\begin{align*}
&(\cL_{\sigma}^{*})^{n+1}(b^*b)\\
&=e^{-n\sigma\tau}(e^{-\sigma\tau}b^*b+p\frac{i\lambda}{\mu}e^{-\sigma\tau}(1-e^{\mu\tau})b^*-
p\frac{i\lambda}{\bar{\mu}}e^{-\sigma\tau}(1-e^{\bar{\mu}\tau})b\\
&+p\frac{\lambda^2}{|\mu|^2}e^{-\sigma\tau}(1-e^{\mu\tau})
(1-e^{\bar{\mu}\tau}))+p\frac{i\lambda}{\mu}(e^{-n\sigma\tau}-e^{-n\bar{\mu}\tau})e^{-\bar{\mu}\tau}b^*+p\frac{i\lambda}{\bar{\mu}}
(1-e^{-\bar{\mu}\tau}))\\
&-p\frac{i\lambda}{\bar{\mu}}(e^{-n\sigma\tau}-e^{-n\mu\tau})e^{-\mu\tau}b-p\frac{i\lambda}{\mu}
(1-e^{-\mu\tau}))\\
&+p\frac{\lambda^2}{|\mu|^2}e^{-\sigma\tau}(1-e^{\mu\tau})(1-e^{\bar{\mu}\tau})\frac{1-e^{-n\sigma\tau}}{1-e^{-\sigma\tau}}-
p^2\frac{2\lambda^2}{|\mu|^2}\frac{1-e^{-n\sigma\tau}}{1-e^{-\sigma\tau}}(1-e^{-\frac{\sigma}{2}\tau}\cos\epsilon\tau)\\
&+p^2\frac{2\lambda^2}{|\mu|^2}(1-e^{-n\frac{\sigma}{2}\tau}\cos n\epsilon\tau).
\end{align*}

Simplifying the latest expression we get
\begin{align*}
&(\cL_{\sigma}^{*})^{n+1}(b^*b)\\
&=e^{-(n+1)\sigma\tau}b^*b+p\frac{i\lambda}{\mu}(e^{-(n+1)\sigma\tau}-
e^{-(n+1)\bar{\mu}\tau})b^*-
p\frac{i\lambda}{\bar{\mu}}(e^{-(n+1)\sigma\tau}-e^{-(n+1)\mu\tau})b\\
&+p\frac{\lambda^2}{|\mu|^2}e^{-\sigma\tau}(1-e^{\mu\tau})(1-e^{\bar{\mu}\tau})\frac{1-e^{-(n+1)\sigma\tau}}
{1-e^{-\sigma\tau}}\\
&-p^2\frac{2\lambda^2}{|\mu|^2}\frac{1}{1-e^{-\sigma\tau}}(-e^{-\frac{\sigma}{2}\tau}\cos\epsilon\tau+e^{-\sigma\tau}+
e^{-(n+1)
\frac{\sigma}{2}\tau}\cos(n+1)\epsilon\tau-e^{-(n+1)\sigma\tau}\\
&+e^{-(n+1)\sigma\tau}e^{-\frac{\sigma}{2}\tau}\cos\epsilon\tau-e^{-(n+1)\frac{\sigma}{2}\tau}e^{-\sigma\tau}\cos(n+1)
\epsilon\tau)\\
&=e^{-(n+1)\sigma\tau}b^*b+p\frac{i\lambda}{\mu}(e^{-(n+1)\sigma\tau}-e^{-(n+1)\bar{\mu}\tau})b^*-p\frac{i\lambda}
{\bar{\mu}}(e^{-(n+1)\sigma\tau}-e^{-(n+1)\mu\tau})b\\
&+p\frac{\lambda^2}{|\mu|^2}e^{-\sigma\tau}(1-e^{\mu\tau})(1-e^{\bar{\mu}\tau})\frac{1-e^{-(n+1)\sigma\tau}}
{1-e^{-\sigma\tau}}\\
&-
p^2\frac{2\lambda^2}{|\mu|^2}\frac{1-e^{-(n+1)\sigma\tau}}{1-e^{-\sigma\tau}}(1-e^{-\frac{\sigma}{2}\tau}\cos\epsilon\tau)+p^2\frac{2\lambda^2}{|\mu|^2}(1-e^{-(n+1)\frac{\sigma}{2}\tau}\cos (n+1)\epsilon\tau),
\end{align*}
which proves (\ref{n-th power on b^*b}).

Note that if we take $n=1$ in (\ref{n-th power on b^*b}) we get (\ref{cL^t(b^*b)+s}), if  we take $\sigma=0$ we get
(\ref{L^n(b^*b)}) and if $\tau=0$ we get $b^*b$.

Therefore by (\ref{n-th power on b^*b}) we obtain
for a \textit{gauge-invariant} initial state the mean-value of the photon number in the open cavity
at $t=n\tau$:
\begin{align}\label{mean-value}
&N_\sigma(n\tau)= \Tr_{\cH_C}(\rho_C(\cL_{\sigma}^{*})^n(b^*b))=e^{-n(\sigma_--\sigma_+)\tau}N_\sigma(0)\\
&+p(1-p)\frac{2\lambda^2}{|\mu|^2}(1-e^{-(\sigma_{-} - \sigma_{+})\tau/2}\cos \epsilon\tau)
\frac{1-e^{-n(\sigma_--\sigma_+)\tau}}{1-e^{-(\sigma_--\sigma_+)\tau}}\nonumber\\
&+p^2\frac{2\lambda^2}{|\mu|^2}(1-e^{-n(\sigma_{-} - \sigma_{+})\tau/2}\cos n\epsilon\tau)
+\frac{\sigma_+}{\sigma_--\sigma_+}(1-e^{-n(\sigma_--\sigma_+)\tau}).\nonumber
\end{align}

In the limit $n$ goes to infinity we get
\begin{align}
&\lim_{n\rightarrow\infty}(\cL_{\sigma}^{*})^n(b^*b)\nonumber\\
&=p\frac{\lambda^2}{|\mu|^2}e^{-\sigma\tau}(1-e^{\mu\tau})
(1-e^{\bar{\mu}\tau}) \frac{1}{1-e^{-\sigma\tau}}-p^2\frac{2\lambda^2}{|\mu|^2}\frac{1}{1-e^{-\sigma\tau}}
(1-e^{-\frac{\sigma}{2}\tau}\cos\epsilon\tau)+ p^2\frac{2\lambda^2}{|\mu|^2}\nonumber\\
&=p\frac{\lambda^2}{|\mu|^2}\frac{1+e^{-\sigma\tau}-2e^{-\frac{\sigma}{2}\tau}\cos\epsilon\tau}{1-e^{-\sigma\tau}}-
p^2\frac{2\lambda^2}{|\mu|^2}\frac{e^{-\sigma\tau}-e^{-\frac{\sigma}{2}\tau}\cos\epsilon\tau}{1-e^{-\sigma\tau}}
\label{N_with sigma}.
\end{align}
\end{proof}

\begin{remark}\label{number of photons-sigmaREM}
Notice that for $\sigma > 0$ and $p=0$ the limit value (\ref{lim-photons-number-sigma}) is \textit{zero}. This is trivial
because any initial cavity state with finite mean-value of photons (\ref{int-cond-photons-sigma}) will be exhausted by the
leaking, $\sigma > 0$, and the absence of pumping: $p=0$.

Let $0<p<1$. Then for \textit{non-resonant} case $\epsilon\tau \neq 2\pi s$, where $s\in \mathbb{Z}^1$, the limit
\begin{equation}\label{sigma to zero1}
\lim_{\sigma \rightarrow 0} \omega_{C,\sigma}(b^*b) = + \infty  \ ,
\end{equation}
which corresponds to conclusion of the Theorem \ref{N of photons} about unlimited pumping of the perfect cavity,
i.e. for $\sigma=0$.

The interpretation of (\ref{lim-photons-number-sigma}) is less transparent in two special cases:\\
(a) For the resonant case $\epsilon\tau = 2\pi s$ and $\sigma=0$ the mean-value of photons
(\ref{number of photons_Hamiltonian}) is bounded and equal to $N(0)$ independent of $p$. Whereas
(\ref{lim-photons-number-sigma}) yields  the bound
\begin{equation}\label{sigma to zero2}
\lim_{\sigma \rightarrow 0} \omega_{C,\sigma}(b^*b) = p^2 \ \frac{\lambda^2}{\epsilon^2} \ ,
\end{equation}
which is $p$-dependent.\\
(b) For $p=1$ and $\sigma=0$ the mean-value of photons (\ref{number of photons_Hamiltonian}) is bounded and oscillates.
Although (\ref{lim-photons-number-sigma}) gives
\begin{equation}\label{sigma3}
\omega_{C,\sigma}(b^*b)=  \frac{\lambda^2}{|\mu|^2} \ ,
\end{equation}
for any leaking $\sigma > 0$.

So clearly, the  limits: $t\rightarrow \infty$ and $\sigma \rightarrow 0$ do not commute.
\end{remark}

\section{Long-time behavior}
\label{sec:Nonequil_state}
Recall that we consider the limiting state \eq{Weyl-func-sigma} on the Weyl operator algebra. Our first result concerning the limiting state of the leaking cavity is the following:
\begin{theorem}\label{cavity-lim-state-sigma}
The limiting cavity state
\begin{equation}\label{lim-cavity-sigma}
\omega_{C,\sigma}(\cdot):= \wlim \, \omega_{C,\sigma}^{t}(\cdot)
\end{equation}
exists and it is independent of the initial state $\rho_C$.
\end{theorem}
 
\begin{proof} Using (\ref{L^C_lambda(B)}) we find the action of $(L_\lambda^C)^{*}$ on the Weyl operator
\begin{align}
&(L_\lambda^C)^{*}(W(\alpha))\nonumber\\
&=i\epsilon[b^*b, W(\alpha)]+\sigma
(b^*-{\lambda}/{\epsilon})W(\alpha)(b-{\lambda}/{\epsilon})-\frac{\sigma}{2}\{(b^*-{\lambda}/{\epsilon})
(b-{\lambda}/{\epsilon}),
W(\alpha)\}\nonumber \\
&=(-i\epsilon(\alpha b+\overline{\alpha}b^*+|\alpha|^2)-\frac{\sigma}{2}(\alpha b
-\overline{\alpha}b^*+|\alpha|^2-{\lambda}/{\epsilon}(\alpha-\overline{\alpha}))W(\alpha)\nonumber \\
&=(-\mu\alpha
b+\overline{\mu}\overline{\alpha}b^*-\mu|\alpha|^2+\frac{\lambda\sigma}{2\epsilon}(\alpha-\overline{\alpha}))
W(\alpha)\label{LW},
\end{align}
where we denoted $\mu=\frac{\sigma}{2}+i\epsilon.$

Therefore the dynamics generated by (\ref{L^C_lambda(B)}):
\begin{equation}\label{Dyn-C}
\gamma_{\lambda,\tau}:=e^{\tau(L_\lambda^C)^{*}}
\end{equation}
is \textit{quasi-free},
which by definition \cite{Verbeure} means that for some $T_{\lambda,\tau}(\alpha)$ and $f_{\lambda,\tau}(\alpha)$
\begin{equation}\label{gammaW}
\gamma_{\lambda,\tau}(W(\alpha))=f_{\lambda,\tau}(\alpha)W(T_{\lambda, \tau}(\alpha))=e^{g_{\lambda,
\tau}(\alpha)}W(T_{\lambda, \tau}(\alpha)).
\end{equation}

We can find $T_{\lambda, \tau}(\alpha)$ and $g_{\lambda, \tau}(\alpha)$ using the differential equation the dynamics
should satisfy
\begin{equation}\label{Diff_W}
\frac{d\gamma_{\lambda,\tau}(W(\alpha))}{d\tau}=(L_\lambda^C)^{*}(\gamma_{\lambda,\tau}(W(\alpha))).
\end{equation}

The right-hand side can be calculated using (\ref{LW}) equation where instead of $\alpha$ we have $T_{\lambda,\tau}(\alpha)$
\begin{align}
(L_\lambda^C)^{*}(\gamma_{\lambda,\tau}(W(\alpha)))=&f_{\lambda, \tau}(\alpha)(-\mu T_{\lambda,
\tau}(\alpha)b+\overline{\mu}\overline{T_{\lambda, \tau}}(\alpha)b^*-\mu |T_{\lambda,
\tau}(\alpha)|^2+\frac{\lambda\sigma}{2\epsilon}(T_{\lambda, \tau}(\alpha)\nonumber\\
&-\overline{T_{\lambda, \tau}(\alpha)}))
W(T_{\lambda,\tau}(\alpha))\nonumber\\
=&(-\mu T_{\lambda, \tau}(\alpha)b+\overline{\mu}\overline{T_{\lambda, \tau}}(\alpha)b^*-\mu |T_{\lambda,
\tau}(\alpha)|^2+\frac{\lambda\sigma}{2\epsilon}(T_{\lambda, \tau}(\alpha)\label{LW_tau}\\
&-\overline{T_{\lambda,
\tau}(\alpha)}))\gamma_{\lambda,\tau}(W(\alpha))\nonumber.
\end{align}

Using the Baker-Campbell-Hausdorff formula we can write the Weyl operator  in the following form
$$W(\alpha)=e^{\alpha b-\overline{\alpha}b^*}=e^{-\overline{\alpha}b^*}e^{\alpha b}e^{-\frac{|\alpha|^2}{2}}. $$
Then the derivative of the $\tau$-dependent Weyl operator $W(T_{\lambda,\tau}(\alpha))$ will look like
\begin{align*}
\frac{dW(T_{\lambda,\tau}(\alpha))}{d\tau}=&-\frac{d\overline{T_{\lambda,
\tau}(\alpha)}}{d\tau}b^*W(T_{\lambda,\tau}(\alpha))+\frac{dT_{\lambda,\tau}(\alpha)}{d\tau}e^{-\overline{T_{\lambda,\tau}
(\alpha)}b^*}be^{T_{\lambda,\tau}b}e^{-\frac{|T_{\lambda,\tau}(\alpha)|^2}{2}}\\
=&\left(\frac{dT_{\lambda,\tau}(\alpha)}{d\tau}b-\frac{d\overline{T_{\lambda,
\tau}(\alpha)}}{d\tau}b^*+\overline{T_{\lambda,\tau}}\frac{dT_{\lambda,\tau}(\alpha)}{d\tau}\right)W(T_{\lambda,\tau}(\alpha))
\end{align*}

Therefore $\gamma_{\lambda,\tau}(W(\alpha))$ satisfies the following differential equation
\begin{align}
\frac{d\gamma_{\lambda,\tau}(W(\alpha))}{d\tau}&=f_{\lambda, \tau}(\alpha)\frac{dg_{\lambda, \tau}(\alpha)}{d\tau}
W(T_{\lambda,\tau}(\alpha))+f_\tau(\alpha)(\frac{dT_{\lambda,\tau}(\alpha)}{d\tau}b-\frac{d\overline{T_{\lambda,
\tau}(\alpha)}}{d\tau}b^*\nonumber\\
&+\overline{T_{\lambda,\tau}}\frac{dT_{\lambda,\tau}(\alpha)}{d\tau})W(T_{\lambda,\tau}(\alpha))
\nonumber\\
&=\left(\frac{dT_{\lambda,\tau}(\alpha)}{d\tau}b-\frac{d\overline{T_{\lambda,\tau}(\alpha)}}{d\tau}b^*+\frac{dg_{\lambda,
\tau}(\alpha)}{d\tau}+\overline{T_{\lambda,\tau}}\frac{dT_{\lambda,\tau}(\alpha)}{d\tau}\right)\gamma_{\lambda,\tau}
(W(\alpha))\label{dW}.
\end{align}

Because of the differential equation for the dynamics (\ref{Diff_W})  the last equation (\ref{dW}) should coincides with the
equation (\ref{LW_tau}). Then we get the following system of differential equations
\begin{equation*}
\frac{dT_{\lambda,\tau}(\alpha)}{d\tau}=-\mu T_{\lambda, \tau}(\alpha)
\end{equation*}
and
\begin{equation*}
\frac{dg_{\lambda,\tau}(\alpha)}{d\tau}=-\overline{T_{\lambda,\tau}}\frac{dT_{\lambda,\tau}(\alpha)}{d\tau}-\mu |T_{\lambda,
\tau}(\alpha)|^2+\frac{\lambda\sigma}{2\epsilon}(T_{\lambda, \tau}(\alpha)-\overline{T_{\lambda, \tau}(\alpha)}).
\end{equation*}
Using the first equation the second one could be simplified and will look like
\begin{equation*}
\frac{dg_{\lambda,\tau}(\alpha)}{d\tau}=\frac{\lambda\sigma}{2\epsilon}(T_{\lambda, \tau}(\alpha)-\overline{T_{\lambda,
\tau}(\alpha)})
\end{equation*}

The solution to the first differential equation is $$T_{\lambda, \tau}(\alpha)= e^{-\mu\tau}\alpha.$$
Therefore the second differential equation can be written as follows
\begin{align*}
\frac{dg_{\lambda,\tau}(\alpha)}{d\tau}=\frac{\lambda\sigma}{2\epsilon}(e^{-\mu\tau}\alpha-e^{-\overline{\mu}\tau}
\overline{\alpha})
\end{align*}
and the solution to this equation is
\begin{equation*}
g_{\lambda,\tau}(\alpha)=\frac{\lambda\sigma}{2\epsilon\mu}(1-e^{-\mu\tau})\alpha-\frac{\lambda\sigma}
{2\epsilon\overline{\mu}} (1-e^{-\overline{\mu}\tau})\overline{\alpha}.
\end{equation*}

Putting the solutions $T_{\lambda,\tau}(\alpha)$ and $g_{\lambda,\tau}(\tau)$ into the expression for
$\gamma_{\lambda,\tau}(W(\alpha))$ (\ref{gammaW}) we get the following
\begin{equation*}
\gamma_{\lambda,\tau}(W(\alpha))=e^{\frac{\lambda\sigma}{2\epsilon\mu}(1-e^{-\mu\tau})\alpha-\frac{\lambda\sigma}
{2\epsilon\overline{\mu}}(1-e^{-\overline{\mu}\tau})\overline{\alpha}}W(e^{-\mu\tau}
\alpha).
\end{equation*}

In order to calculate $\cL_{\sigma}^{*}(W(\alpha))$ we use (\ref{cL(B)}) equation. The first term has the following form
\begin{align*}
&e^{-{\lambda}(b^*-b)/{\epsilon}}\gamma_{\lambda,\tau}(e^{{\lambda}(b^*-b)/{\epsilon}}W(\alpha)e^{-\frac{\lambda}
{\epsilon}(b^*-b)})e^{{\lambda}(b^*-b)/{\epsilon}}\\
&=e^{-{\lambda}/{\epsilon}(\alpha-\overline{\alpha})}
e^{-{\lambda}(b^*-b)/{\epsilon}}\gamma_{\lambda,\tau}(W(\alpha))e^{{\lambda}(b^*-b)/{\epsilon}}\\
&=e^{-{\lambda}/{\epsilon}(\alpha-\overline{\alpha})}e^{\frac{\lambda\sigma}{2\epsilon\mu}(1-e^{-\mu\tau})\alpha-
\frac{\lambda\sigma}{2\epsilon\overline{\mu}}(1-e^{-\overline{\mu}\tau})\overline{\alpha}}e^{-{\lambda}/{\epsilon}
(b^*-b)}W(e^{-\mu\tau}
\alpha)e^{{\lambda}(b^*-b)/{\epsilon}}\\
&=e^{-{\lambda}/{\epsilon}((1-e^{-\mu\tau})\alpha-(1-e^{-\overline{\mu}\tau})\overline{\alpha})}
e^{\frac{\lambda\sigma}{2\epsilon\mu}(1-e^{-\mu\tau})\alpha-\frac{\lambda\sigma}{2\epsilon\overline{\mu}}
(1-e^{-\overline{\mu}\tau})\overline{\alpha}}W(e^{-\mu\tau}
\alpha)\\
&=e^{-\frac{i\lambda}{\mu}(1-e^{-\mu\tau})\alpha-\frac{i\lambda}{\overline{\mu}}(1-e^{-\overline{\mu}\tau})
\overline{\alpha}}W(e^{-\mu\tau}
\alpha).
\end{align*}

Therefore from (\ref{cL(B)}) we get
\begin{align*}
\cL_{\sigma}^{*}(W(\alpha))=&pe^{-\frac{i\lambda}{\mu}(1-e^{-\mu\tau})\alpha-\frac{i\lambda}{\overline{\mu}}
(1-e^{-\overline{\mu}\tau})\overline{\alpha}}W(e^{-\mu\tau}
\alpha)+(1-p)W(e^{-\mu\tau}\alpha)\\
=&(pe^{-\frac{i\lambda}{\mu}(1-e^{-\mu\tau})\alpha-\frac{i\lambda}{\overline{\mu}}(1-e^{-\overline{\mu}\tau})
\overline{\alpha}}+1-p)W(e^{-\mu\tau}
\alpha).
\end{align*}

Therefore
\begin{equation}\label{nth power on W}
(\cL_{\sigma}^{*})^n(W(\alpha))=\prod_{k=0}^{n-1}(pe^{-\frac{i\lambda}{\mu}(1-e^{-\mu\tau})e^{-k\mu\tau}\alpha-\frac{i\lambda}
{\overline{\mu}}(1-e^{-\overline{\mu}\tau})e^{k\mu\tau}\overline{\alpha}}+1-p)W(e^{-n\mu\tau}
\alpha).
\end{equation}

In the limit $n$ goes to infinity $W(e^{-n\mu\tau} \alpha)$ converges weakly to $\idty$, so the dependence of the limiting
state on the initial state disappears. To see the convergence of the product
let us denote
$$h_k(\alpha)=p(e^{-\frac{i\lambda}{\mu}(1-e^{-\mu\tau})e^{-k\mu\tau}\alpha-\frac{i\lambda}{\overline{\mu}}
(1-e^{-\overline{\mu}\tau})e^{k\mu\tau}\overline{\alpha}}-1).$$
The product $\prod_{k=0}^\infty (1+h_k(\alpha))$ converges if and only if the sum $\sum_{k=0}^\infty |h_k(\alpha)|$ converges.
Writing $h_k(\alpha)$ as a sum of real and imaginary parts we  get the following bound
\begin{equation*}
|h_k(\alpha)|\leq 12p\frac{\lambda|\alpha|}{|\mu|}e^{-k\frac{\sigma}{2}\tau},
\end{equation*}
from which the convergence of the series immediately follows.

The existence of the limiting state is guarantied by Levy's continuity theorem (\cite{Fristedt} Theorem 18.21).
\end{proof}

\section{Energy flux and entropy production}
\label{sec:Nonequil_energy}

\subsection{Energy variation in perfect cavity}
\label{sec:Nonequil_energy_perfect}

Since time-dependent interaction in (\ref{W-int})
is piece-wise constant, our system is autonomous on each interval $[(n-1)\tau, n\tau)$. Therefore, there is no variation of
energy on this interval although it may jump, when a new atom enters into the cavity. Note that although the total energy
corresponding to the infinite system (\ref{Ham-Model}) has no sense, its variation is well-defined.

Let the $n$-th atom is actually traveling in the cavity, i.e. $t = n(t)\tau + \nu(t)$, with $n(t)=n-1$ and $\nu(t)\in[0, \tau)$,
see (\ref{t}). Then one can compare the the expectation of total energy of the system for the moment $t_{n}=(n-1)\tau + \nu(t_{n})$,
with that, when the $n-1$-th atom was in the cavity, $t_{n-1}=(n-2)\tau + \nu(t_{n-1})$. Then by (\ref{Ham-Model}), (\ref{Ham-n})
and (\ref{Sol-Liouv-Eq}) the energy variation between two moments $t_{n-1}$ and $t_{n}$ is defined by
\begin{eqnarray}\label{DEn-1,n1}
&&\Delta\mathcal{E}(t_{n},t_{n-1}):=\Tr_{\cH_C \otimes \cH_A}(\rho_S(t_{n}) H(t_{n})) -
\Tr_{\cH_C \otimes \cH_A}(\rho_S(t_{n-1}) H(t_{n-1})).
\end{eqnarray}

\begin{lemma}
For any $\nu\in[0,\tau]$ one has
\begin{equation}
\Tr(\rho_S((n-1)\tau+\nu)H_n)=\Tr(\rho_S((n-1)\tau)H_n).
\end{equation}
\end{lemma}
\begin{proof}
\begin{align*}
&\Tr(\rho_S((n-1)\tau+\nu)H_n)=\Tr(e^{-i\nu H_n}\rho_S((n-1)\tau)e^{i\nu H_n}H_n)\\
&=\Tr(\rho_S((n-1)\tau)H_n).
\end{align*}
\end{proof}

Therefore 
\begin{eqnarray}
&&\Delta\mathcal{E}(t_{n},t_{n-1})=\Tr_{\cH_C \otimes \cH_A}(\rho_S((n-1)\tau) [H_{n} - H_{n-1}]) \label{DE_simplified} \ ,
\end{eqnarray}
Let us introduce operator
\begin{equation}\label{PI}
\pi_{n}(\tau):= e^{i \tau H_{n}}\ldots e^{i \tau H_{1}} \ .
\end{equation}
Since (\ref{Ham-n}) implies
\begin{equation}\label{H-n-H-n-1}
H_{n} - H_{n-1} = \lambda \, (b^*+b)\otimes (a_n^*a_n -a_{n-1}^*a_{n-1})+ \idty \otimes E \, (a_n^*a_n - a_{n-1}^*a_{n-1}) \ ,
\end{equation}
by (\ref{DE_simplified}), (\ref{PI}) and by $[H_{k'}, a_k^*a_k] = 0$ we obtain
\begin{eqnarray}\label{DEn-1,n2}
&&\Delta\mathcal{E}(t_{n},t_{n-1}) =   \\
&& \Tr_{\cH_C \otimes \cH_A}\{\rho_C\otimes\rho_A \ \pi_{n-1}(\tau) (\lambda \,(b^*+b)\otimes \idty) \pi_{n-1}^*(\tau)
[\idty\otimes (a_n^*a_n -a_{n-1}^*a_{n-1})]\} \nonumber \\
&& + \Tr_{\cH_C \otimes \cH_A}\{\rho_C\otimes\rho_A (\idty \otimes E (a_n^*a_n -a_{n-1}^*a_{n-1}))\} \nonumber \ .
\end{eqnarray}
\begin{lemma}\label{n-dual-on-b} For any $n\geq 1$ one gets:
\begin{eqnarray}\label{k-iter}
&&\pi_{n}(\tau) (b^* \otimes \idty)\pi_{n}^*(\tau) = e^{n i \tau \epsilon} (b^* \otimes \idty) -
\frac{\lambda}{\epsilon} (1- e^{i \tau \epsilon})\sum_{k=1}^{n} e^{(k-1) i \tau \epsilon} \idty \otimes a_k^*a_k  \ , \\
&&\pi_{n}(\tau) (b \otimes \idty)\pi_{n}^*(\tau) = e^{- n i \tau \epsilon} (b\otimes \idty) -
\frac{\lambda}{\epsilon} (1- e^{- i \tau \epsilon})\sum_{k=1}^{n} e^{-(k-1) i \tau \epsilon} \idty \otimes a_k^*a_k  \nonumber \ .
\end{eqnarray}
\end{lemma}
\begin{proof} We prove the Lemma by induction. Let
\begin{equation}\label{B-k}
B_{k}^*(\tau) : = e^{i \tau H_{k}} (b^* \otimes \idty)e^{- i \tau H_{k}} \ , \  k \geq 1 \ .
\end{equation}
Then by (\ref{Ham-n}) the operator (\ref{B-k}) is solution of equation
\begin{equation*}
\partial_{s}B_{k}^*(s)= i [H_{k}, B_{k}^*(s)] = i \epsilon B_{k}^*(s) + \lambda \idty \otimes a_k^*a_k \ ,
\ B_{k}^*(0) = b^* \otimes \idty \ ,
\end{equation*}
which has the following explicit form:
\begin{equation}\label{B*-k-sol}
B_{k}^*(\tau) = e^{i \tau \epsilon} (b^* \otimes \idty) - \frac{\lambda}{\epsilon} (1- e^{i \tau \epsilon}) \idty \otimes a_k^*a_k \ .
\end{equation}
Similarly one obtains
\begin{equation}\label{B-k-sol}
B_{k}(\tau) = e^{- i \tau \epsilon} (b \otimes \idty) - \frac{\lambda}{\epsilon} (1- e^{- i \tau \epsilon}) \idty \otimes a_k^*a_k \ .
\end{equation}

Suppose that (\ref{k-iter}) is true for $\pi_n$, we prove it for $\pi_{n+1}$.

\begin{eqnarray*}
&&\pi_{n+1}(\tau)(b^*\otimes\idty)=e^{i\tau H_{n+1}}(\pi_n(\tau)(b^*\otimes\idty))e^{-i\tau H_{n+1}}\\
&&=e^{i\tau H_{n+1}}\Bigl( e^{n i \tau \epsilon} (b^* \otimes \idty) -
\frac{\lambda}{\epsilon} (1- e^{i \tau \epsilon})\sum_{k=1}^{n} e^{(k-1) i \tau \epsilon} \idty \otimes a_k^*a_k \Bigl)e^{-i\tau H_{n+1}}\\
&&=e^{n i \tau \epsilon}e^{i\tau H_{n+1}}(b^*\otimes\idty)e^{-i\tau H_{n+1}} -
\frac{\lambda}{\epsilon} (1- e^{i \tau \epsilon})\sum_{k=1}^{n} e^{(k-1) i \tau \epsilon}e^{i\tau H_{n+1}}( \idty \otimes a_k^*a_k )e^{-i\tau H_{n+1}}.
\end{eqnarray*}

By (\ref{B*-k-sol}) and by  $[H_{k'}, a_k^*a_k] = 0$ we obtain
\begin{eqnarray*}
&&\pi_{n+1}(\tau)(b^*\otimes\idty)=e^{i (n+1)\tau \epsilon} (b^* \otimes \idty) - \frac{\lambda}{\epsilon} (1- e^{i \tau \epsilon})e^{n i \tau \epsilon} \idty \otimes a_k^*a_k\\
&& -\frac{\lambda}{\epsilon} (1- e^{i \tau \epsilon})\sum_{k=1}^{n} e^{(k-1) i \tau \epsilon} \idty \otimes a_k^*a_k\\
&&=e^{(n+1) i \tau \epsilon} (b^* \otimes \idty) -
\frac{\lambda}{\epsilon} (1- e^{i \tau \epsilon})\sum_{k=1}^{n+1} e^{(k-1) i \tau \epsilon} \idty \otimes a_k^*a_k. 
\end{eqnarray*}

Since the similar formula can be obtained for $\pi_{n+1}(b\otimes\idty)$ the last formula proves the lemma.

\end{proof}
Recall that we suppose that atomic beam is homogeneous, i.e.,
$p = \Tr_{\cH_C \otimes \cH_A}\{\rho_C\otimes\rho_A (\idty \otimes a_n^*a_n)\}$ is the probability that atom is in its excited state and $p$ is independent of $n$. Then
\begin{equation}\label{p-2p}
\Tr_{\cH_A}\{\rho_A (a_{n_1}^*a_{n_1} a_{n_2}^*a_{n_2})\} = \delta_{{n_1},{n_2}}\, p + (1-\delta_{{n_1},{n_2}})\, p^2 \ .
\end{equation}
Then the second term in the right-hand side of (\ref{DEn-1,n2}) vanishes.
Since we also supposed that the initial cavity state $\rho_C$ is gauge-invariant, by Lemma \ref{n-dual-on-b} and (\ref{p-2p}) one obtains a bound for the first term in the right-hand side of (\ref{DEn-1,n2}):
\begin{eqnarray}
&&\Tr_{\cH_C \otimes \cH_A}\{\rho_C\otimes\rho_A \ \pi_{n-1}(\tau) (\lambda \,(b^*+b)\otimes \idty) \pi_{n-1}^*(\tau)
[\idty\otimes (a_n^*a_n -a_{n-1}^*a_{n-1})]\}  \nonumber\\
&=&- \frac{\lambda^2}{\epsilon}  (1- e^{i \tau \epsilon}) \sum_{k=1}^{n-1} e^{(k-1) i \tau \epsilon}
\Tr_{\cH_A}  a_k^*a_k (a_n^*a_n -a_{n-1}^*a_{n-1})                                                     \label{1st-term} \\
&& - \frac{\lambda^2}{\epsilon}  (1- e^{-i \tau \epsilon}) \sum_{k=1}^{n-1} e^{-(k-1) i \tau \epsilon}
\Tr_{\cH_A}  a_k^*a_k (a_n^*a_n -a_{n-1}^*a_{n-1}) \nonumber \\
&=& 2 \ \frac{\lambda^2}{\epsilon} \ p(1-p) \ [\cos((n-2)\tau\epsilon) - \cos((n-1)\tau\epsilon)] \nonumber  \ .
\end{eqnarray}
Hence formulae (\ref{DEn-1,n2}) and (\ref{1st-term}) prove for the total-energy variation the following statement.
\begin{theorem}\label{ThDEn-1,n1}
The energy variation (\ref{DEn-1,n1}) between two moments $t_{n-1}$ and $t_{n}$ , where $n\geq2$ is
\begin{equation}\label{DEn-1,n1-fin}
\Delta\mathcal{E}(t_{n},t_{n-1}) = 2 \ \frac{\lambda^2}{\epsilon} \ p(1-p) \ [\cos((n-2)\tau\epsilon) - \cos((n-1)\tau\epsilon)] \ .
\end{equation}
For the total variation between $t_{1}$ and $t_{n} \geq t_{1}$ we obtain:
\begin{equation}\label{DEn-1,n1-fin}
\Delta\mathcal{E}(t_{n},t_{1}) = \sum_{k=2}^{n} \Delta\mathcal{E}(t_{k},t_{k-1}) =
2 \ \frac{\lambda^2}{\epsilon} \ p(1-p) \ [1 - \cos((n-1)\tau\epsilon)] \ .
\end{equation}
\end{theorem}

\subsection{Energy variation in the leaking cavity}\label{EEP_n0}
Although for the open cavity the time-dependent generator (\ref{Generator}) is still piecewise \textit{constant}
(\ref{Generator-KL}), the cavity energy  is continuously varying between the moments $\{t = k \tau\}_{k\geq0}$
(when the interaction may to jump (\ref{Ham-Model})) because of the leaking/injection of photons.

Therefore, as above we first concentrate on the elementary variation of the total energy, when the $n$-th atom is
going through the cavity between the moments $t^{\prime} = (n-1) \tau$ and $t^{\prime \prime} =n \tau - 0$:
\begin{align}
\Delta \mathcal{E}_{\sigma}(t^{\prime \prime}, t^{\prime}):=
\omega^{n \tau}_{\mathcal{S},\sigma}(H_n) - \omega^{(n-1) \tau}_{\mathcal{S},\sigma}(H_n) \ . \label{EnVar-sigma0}
\end{align}
Here again we used two facts: (1) for $t\in[(n-1)\tau, n\tau)$ the Hamiltonian (\ref{Ham-Model}) of the form (\ref{Ham-n})
is piecewise constant; (2) the state $\omega^{t}_{\mathcal{S},\sigma}(\cdot)$
(\ref{dual-sigma}) is time-continuous.

By virtue of (\ref{dual-sigma}) and of (\ref{S-state-evol-adj-sigma}) for operator $A= H_n$ (\ref{Ham-n}), we see
that the problem (\ref{EnVar-sigma0}) reduces to calculation of expectations:
\begin{equation}\label{EnVar-sigma1}
\epsilon \ \omega^{k\tau}_{\mathcal{S},\sigma}(b^*b \otimes \idty) \ \ {\rm{and}} \ \
\lambda \ \omega^{s\tau}_{\mathcal{S},\sigma} ((b^*+b)\otimes \eta_k) \ , \ k,s \geq 1 \ .
\end{equation}
The first expectation in (\ref{EnVar-sigma1}) is known due to (\ref{C-state-t-sigma}) and Theorem
\ref{number of photons-sigma}, see (\ref{mean-value}):
\begin{equation}\label{EnVar-sigma2}
\epsilon \ \omega^{k\tau}_{\mathcal{S},\sigma}(b^*b \otimes \idty) = \epsilon \, N_\sigma(k\tau) \ .
\end{equation}
To calculate the second expectation in (\ref{EnVar-sigma1}) we use (\ref{dual-sigma}) for operator
$A =((b^*+b)\otimes \idty)(\idty \otimes \eta_k)$
and representation (\ref{S-state-evol-adj-sigma}) for initial gauge-invariant state $\rho_{\mathcal{C}}$ and
for homogeneous atoms state $\rho_{\mathcal{A}}$.
\begin{lemma} \label{n-dual-on-b-sigma}
Let $\sigma_- > \sigma_+ \geq 0$. Then for the mappings $\{(T^\sigma_{n\tau,0})^\ast\}_{n\geq0}$, see
{\rm{(}}\ref{S-state-evol-adj-sigma}{\rm{)}}, one obtains:
\begin{align}\label{T-sigma-n-star}
(T^\sigma_{n\tau,0})^*(b\otimes\idty)&=e^{-n\mu\tau}b\otimes\idty-\frac{\lambda i}{\mu}(1-e^{-\mu\tau})
\sum_{k=1}^ne^{-(n-k)\mu\tau}\idty\otimes \eta_k \ ,
\end{align}
and $(T^\sigma_{n\tau,0})^*(b^*\otimes\idty) = ((T^\sigma_{n\tau,0})^*(b\otimes\idty))^*$.
\end{lemma}
\begin{proof}
We prove this lemma by induction. Suppose that formula (\ref{T-sigma-n-star}) is true for $(T_{n\tau,0}^\sigma)^*$.
Then we show that it is also valid for $(T^\sigma_{(n+1)\tau,0})^*$.
By virtue of (\ref{S-state-evol-adj-sigma}) and by (\ref{tilde_L}) for the action of
operator $e^{\tau L_{\sigma, n}^*}$, one gets
\begin{align*}
&(T^\sigma_{n+1})^*(b\otimes\idty)=(T^\sigma_{n})^*(e^{\tau L_{\sigma, n+1}^*}(b\otimes\idty))\\
&=(T^\sigma_{n})^*(\widehat{S}_{n+1}^{*}(e^{\tau \widehat{L}^*_{\sigma, n+1}}(\widehat{S}_{n+1}
(b\otimes\idty))))\\
&=(T^\sigma_{n})^*(\widehat{S}_{n+1}^{*}(e^{\tau \widehat{L}^*_{\sigma, n+1}} \Bigl((b-\frac{\lambda}{\epsilon})\otimes
\eta_{n+1}+b\otimes (I-\eta_{n+1})\Bigr))) \ ,
\end{align*}
where we used (\ref{tb}) in the last line. By  (\ref{Dyn-C}),(\ref{gamma_b})
combined with the shift $\widehat{S}_{n+1}^{*}$, we obtain
\begin{align*}
&\widehat{S}_{n+1}^{*}(e^{\tau\widehat{L}^*_{\sigma, n+1}}((b-\frac{\lambda}{\epsilon})\otimes \eta_{n+1}))=
\widehat{S}_{n+1}^{*}((\gamma_{\lambda,\tau}(b)-\frac{\lambda}{\epsilon})\otimes \eta_{n+1})\\
&=(e^{-\mu\tau}b-\frac{\lambda i}{\mu}(1-e^{-\mu\tau}))\otimes \eta_{n+1} \ ,\\
&\widehat{S}_{n+1}^{*}(e^{\tau\widehat{L}^*_{\sigma, n+1}}(b\otimes (I-\eta_{n+1}))) =
\widehat{S}_{n+1}^{*}(\gamma_{0,\tau}(b)\otimes (I-\eta_{n+1}))\\
&=e^{-\mu\tau}b\otimes(I - \eta_{n+1}) \ .
\end{align*}
Consequently,
\begin{align*}
&(T^\sigma_{n+1})^*(b\otimes\idty))\\
&=(T^\sigma_{n})^*((e^{-\mu\tau}b-\frac{\lambda i}{\mu}(1-e^{-\mu\tau}))\otimes
\eta_{n+1}+e^{-\mu\tau}b\otimes(\idty- \eta_{n+1}))\\
&=(T^\sigma_{n})^*(e^{-\mu\tau}b\otimes\idty-\frac{\lambda i}{\mu}(1-e^{-\mu\tau})\idty\otimes \eta_{n+1})\\
&=e^{-(n+1)\mu\tau}b\otimes\idty-\frac{\lambda i}{\mu}(1-e^{-\mu\tau})\sum_{k=1}^{n+1}
e^{-(n-k+1)\mu\tau}\idty\otimes \eta_k,
\end{align*}
which proves the lemma.
\end{proof}
Recall that $\mu= i\epsilon + \sigma/2$. Hence, in the limit $\sigma \rightarrow +0$
one recovers from this Lemma formulae (\ref{k-iter}) for the ideal cavity.
\begin{corollary}\label{int-energy-n-n-1}
Let initial cavity state $\rho_{\mathcal{C}}$ be gauge-invariant state for homogeneous
state $\rho_{\mathcal{A}}$ of the atomic beam.
Then with help of (\ref{p-2p}) and (\ref{T-sigma-n-star}) one obtains the interaction energy expectations
(\ref{EnVar-sigma1}) corresponding to the difference (\ref{EnVar-sigma0}):
\begin{align}\label{int-energy-n}
&\lambda \ \omega^{n\tau}_{\mathcal{S},\sigma} ((b^*+b)\otimes \eta_n) =  \\
&-\frac{2\lambda^2\epsilon }{|\mu|^2}\left[p(1-p)(1-e^{-\sigma\tau/2}\cos\epsilon\tau)
+ p^{2}(1-e^{-n\sigma\tau/2}\cos n\epsilon\tau)\right] \nonumber \\
&+\frac{\lambda^2\sigma}{|\mu|^2}\left[p(1-p)e^{-\sigma\tau/2}\sin\epsilon\tau
+p^{2}e^{-n\sigma\tau/2}\sin n\epsilon\tau\right] \ , \nonumber \\
&\lambda \ \omega^{(n-1)\tau}_{\mathcal{S},\sigma} ((b^*+b)\otimes \eta_n) =
- \frac{2\lambda^2\epsilon}{|\mu|^2}p^{2} (1-e^{-(n-1)\sigma\tau/2}\cos (n-1)\epsilon\tau) \nonumber \\
&+ \frac{\lambda^2\sigma}{|\mu|^2} p^{2} e^{-(n-1)\sigma\tau/2}\sin (n-1)\epsilon\tau \ .
\label{int-energy-n-1}
\end{align}
\end{corollary}
\begin{corollary}\label{energy-n-n-1}
Taking into account Theorem \ref{number of photons-sigma} and (\ref{int-energy-n}),
(\ref{int-energy-n-1}) we get for the elementary variation of the total energy (\ref{EnVar-sigma0})
\begin{align}\label{EnVar-sigma3}
&\Delta \mathcal{E}_{\sigma}(n \tau - 0, (n-1) \tau)=\epsilon \, (N_\sigma(n\tau) - N_\sigma((n-1)\tau)) \\
&+\lambda \ (\omega^{n\tau}_{\mathcal{S},\sigma} ((b^*+b)\otimes \eta_n) -
\omega^{(n-1)\tau}_{\mathcal{S},\sigma} ((b^*+b)\otimes \eta_n)) \nonumber \\
&= -\epsilon \, N_\sigma(0)(1 - e^{-\sigma\tau})
e^{-(n-1)\sigma\tau} \nonumber \\
&- p(1-p)\frac{2\lambda^2\epsilon }{|\mu|^2}(1-e^{-\sigma\tau/2}\cos\epsilon\tau)
(1-e^{-(n-1)\sigma\tau}) \nonumber \\
&+ p(1-p) \frac{\lambda^2\sigma}{|\mu|^2} e^{-\sigma\tau/2}\sin \epsilon\tau
\nonumber \\
&+ p^{2}\frac{\lambda^2\sigma}{|\mu|^2}\left[e^{-n\sigma\tau/2}\sin n\epsilon\tau -
e^{-(n-1)\sigma\tau/2}\sin (n-1)\epsilon\tau\right]. \nonumber
\end{align}
\end{corollary}
Note that in the limit of the ideal cavity: $\sigma \rightarrow 0$, one gets
for total energy variation (\ref{EnVar-sigma3}) that increment $\Delta \mathcal{E}_{\sigma}(n \tau - 0, (n-1) \tau) =0$,
which corresponds to the autonomous case. The limit of the energy increment when
$n\rightarrow\infty$ is
\begin{align}\label{lim-increm}
&\lim_{n\rightarrow\infty} \Delta \mathcal{E}_{\sigma}(n \tau - 0, (n-1) \tau) =\\
&p(1-p)\frac{\lambda^2}{|\mu|^2}\left[- 2\epsilon (1-e^{-\sigma\tau/2}\cos\epsilon\tau)
+ \sigma e^{-\sigma\tau/2}\sin \epsilon\tau \right]. \nonumber
\end{align}
\begin{remark}\label{EnVar-sigma-jump}
To consider the impact when the $n$-th atom enters the cavity we study the total energy variation
on the extended interval $((n-1)\tau -0, n\tau -0)$. Then
\begin{align}\label{EnVar-sigma-jump1}
&\Delta \mathcal{E}_{\sigma}(n\tau - 0, (n-1)\tau -0) =
(\omega^{n \tau-0}_{\mathcal{S},\sigma}(H_n) - \omega^{(n-1) \tau}_{\mathcal{S},\sigma}(H_n))\\
&+ (\omega^{(n-1) \tau}_{\mathcal{S},\sigma}(H_n) - \omega^{(n-1)\tau -0}_{\mathcal{S},\sigma}(H_{n-1})) \ ,
\nonumber
\end{align}
where the second difference $\Delta \mathcal{E}_{\sigma}((n-1)\tau,(n-1)\tau -0):=
\omega^{(n-1) \tau}_{\mathcal{S},\sigma}(H_n) - \omega^{(n-1)\tau -0}_{\mathcal{S},\sigma}(H_{n-1})$ corresponds to the
energy variation (\textit{jump}), when the $n$-th atom enters the cavity and the $(n-1)$-th atom leaves it.
\end{remark}
To calculate $\Delta \mathcal{E}_{\sigma}((n-1)\tau, (n-1)\tau -0)$ note that by the time continuity of the state
\begin{align*}
&\Delta \mathcal{E}_{\sigma}((n-1)\tau,(n-1)\tau -0)=\Tr(\rho_S((n-1)\tau)H_{n})-\Tr(\rho_S((n-1)\tau)H_{n-1})\\
&=\Tr\left(e^{\tau L_{\sigma,n-1}}\, ... \ e^{\tau L_{\sigma, 1}}(\rho_C\otimes\rho_\cA))(H_n-H_{n-1})\right)\\
&=\Tr\left(T^\sigma_{(n-1)\tau,0}(\rho_C\otimes\rho_\cA)(H_n-H_{n-1})\right)\\
&=\Tr\left(\rho_C\otimes\rho_A(T^\sigma_{(n-1)\tau,0})^*(\lambda(b^*+b)\otimes (\eta_n-\eta_{n-1})\right)\\
&=\Tr\left(\rho_C\otimes\rho_A(T^\sigma_{(n-1)\tau,0})^*(\lambda(b^*+b)\otimes (\eta_n-\eta_{n-1})\right)\\
&=\Tr\{\rho_C\otimes\rho_A (T^\sigma_{(n-1)\tau,0})^*(\lambda (b^*+b)\otimes \idty)[\idty\otimes (\eta_n -\eta_{n-1})]\},
\end{align*}
where $T^\sigma_{t=n\tau,0}=e^{\tau L_{\sigma,n}}\, ... \ e^{\tau L_{\sigma, 1}}$ is defined by (\ref{Sol-Liouv-Eq-sigma}).

If the initial cavity state is gauge-invariant, then (\ref{T-sigma-n-star}) yields for the energy jump at the moment
$t=(n-1)\tau$:
\begin{align}\label{energy-jump}
&\Delta \mathcal{E}_{\sigma}((n-1)\tau,(n-1)\tau -0)=\\
&\Tr\{\rho_C\otimes\rho_A (T^\sigma_{(n-1)\tau,0})^* (\lambda (b^*+b)\otimes \idty)
[\idty\otimes (\eta_n -\eta_{n-1})]\}\nonumber  \\
&= \frac{\lambda^2 i}{\bar{\mu}}(1-e^{-\bar{\mu}\tau})\sum_{k=1}^{n-1}e^{-(n-k-1)\bar{\mu}\tau}
\Tr_{\cH_A} (\rho_A\eta_k(\eta_n -\eta_{n-1})) \nonumber\\
&-\frac{\lambda^2 i}{\mu}(1-e^{-\mu\tau})\sum_{k=1}^{n-1}e^{-(n-k-1)\mu\tau}\Tr_{\cH_A} (\rho_A \eta_k(\eta_n -
\eta_{n-1})) .\nonumber
\end{align}
Taking into account the Bernoulli property (\ref{p-2p}) we obtain from (\ref{energy-jump})
\begin{align}
&\Delta \mathcal{E}_{\sigma}((n-1)\tau,(n-1)\tau -0)=
\frac{\lambda^2 i}{\mu}(1-e^{-\mu\tau})p(1-p)-\frac{\lambda^2 i}{\bar{\mu}}(1-e^{-\bar{\mu}\tau})p(1-p)\nonumber\\
&=p(1-p)\frac{2\lambda^2\epsilon}{|\mu|^2}(1-e^{-\sigma\tau/2}\cos\epsilon\tau)\nonumber \\
&-p(1-p)\frac{\lambda^2\sigma}{|\mu|^2}e^{-\sigma\tau/2}\sin\epsilon\tau.
\label{second_term}
\end{align}

Notice again that for $\sigma\rightarrow +0$ one obtains from (\ref{second_term}) the one-step energy variation
for the ideal cavity (\ref{1st-term}).

Summarising (\ref{EnVar-sigma3}) and (\ref{second_term}), we obtain the energy increment (\ref{EnVar-sigma-jump1})
which is due to impact of the open cavity effects (\ref{EnVar-sigma3}) and to the atomic beam pumping (\ref{second_term}):
\begin{align}\label{EnVar-sigma-jump2}
&\Delta \mathcal{E}_{\sigma}(n\tau - 0, (n-1)\tau -0) = \\
&=- \epsilon N_\sigma(0)(1 - e^{-(\sigma_{-} - \sigma_{+})\tau})
e^{-(n-1)\sigma\tau} \nonumber \\
&+ p(1-p)\frac{2\lambda^2\epsilon }{|\mu|^2}(1-e^{-\sigma\tau/2}\cos\epsilon\tau)
e^{-(n-1)\sigma\tau} \nonumber \\
&+p^{2}\frac{\lambda^2\sigma}{|\mu|^2}\left[e^{-n\sigma\tau/2}\sin n\epsilon\tau -
e^{-(n-1)\sigma)\tau/2}\sin (n-1)\epsilon\tau\right]. \nonumber
\end{align}
\begin{theorem}\label{Energy-sigma-variation}
By virtue of (\ref{EnVar-sigma-jump2}) the total energy variation between initial state at the moment $t_0 := -0$,
when the cavity is empty, and the moment $t_n := n\tau -0$, just before the $n$-th atom is ready to leave the
cavity, is
\begin{align}\label{energ-var-op-syst-0-n}
&\Delta \mathcal{E}_{\sigma}(t_n, t_0) = \sum_{k=1}^n \Delta \mathcal{E}_{\sigma}(k\tau - 0, (k-1)\tau -0)\\
&= -\epsilon N_\sigma(0)(1 - e^{-n\sigma\tau})
\nonumber \\
&+p(1-p)\frac{2\lambda^2\epsilon }{|\mu|^2}(1-e^{-\sigma\tau/2}\cos\epsilon\tau)
\frac{1 - e^{-n\sigma\tau}}{1 - e^{-\sigma\tau}}\nonumber \\
&+p^{2}\frac{\lambda^2\sigma}{|\mu|^2} e^{-n\sigma)\tau/2} \sin n\epsilon\tau
\nonumber \ .
\end{align}
Here $N_\sigma(0) = \omega^{t_0}_{\mathcal{S},\sigma}(b^*b\otimes\idty)$ is the initial number of photons in
the cavity.
\end{theorem}
\begin{remark}\label{Energy-sigma-variationBIS}
Note that the total energy variation (\ref{energ-var-op-syst-0-n}) is due to evolution of the photon number in the
open cavity (\ref{mean-value}) and the variation of the interaction energy (\ref{int-energy-n}), that give
\begin{align}\label{energ-var-op-syst-0-nBIS}
\Delta \mathcal{E}_{\sigma}(t_n, t_0) &=
\epsilon \omega^{n\tau}_{\mathcal{S},\sigma}(b^*b\otimes\idty) +
\lambda \ \omega^{n\tau}_{\mathcal{S},\sigma} ((b^*+b)\otimes \eta_n) \\
&-\epsilon \omega^{t_0}_{\mathcal{S},\sigma}(b^*b\otimes\idty) \ . \nonumber
\end{align}
For $\sigma>0$ it is uniformly bounded from above
\begin{align}\label{energ-var-op-syst-0-n+}
\Delta \mathcal{E}_{\sigma}(t_n, t_0) &\leq \frac{2\lambda^2\epsilon }{|\mu|^2} \frac{p(1-p)}{1 - e^{-(sigma\tau/2}} \\
&+p^{2}\frac{\lambda^2\sigma}{|\mu|^2} . \nonumber
\end{align}
The lower bound of (\ref{energ-var-op-syst-0-n}) is also evident. It strongly depends on the initial condition
$N_\sigma(0)$ and can be negative.
\end{remark}
The long-time asymptotic of (\ref{energ-var-op-syst-0-n}), or (\ref{energ-var-op-syst-0-nBIS}), is
\begin{align}\label{energ-var-op-syst-0-inf}
&\Delta \mathcal{E}_{\sigma}:=\lim_{n\rightarrow\infty}\Delta \mathcal{E}_{\sigma}(t_n, t_0)
=-\epsilon N_\sigma(0)
+p(1-p)\frac{2\lambda^2\epsilon }{|\mu|^2}
\frac{1-e^{-\sigma\tau/2}\cos\epsilon\tau}{1 - e^{-\sigma\tau}} \ .
\end{align}
From (\ref{energ-var-op-syst-0-inf}) one gets that in the open cavity with $\sigma_{-} - \sigma_{+}>0$
the asymptotic of the total-energy variation is bounded from above and from below
\begin{align}\label{bound+}
&\Delta \mathcal{E}_{\sigma}\leq-\epsilon N_\sigma(0)
+ \frac{2\lambda^2\epsilon }{|\mu|^2} \frac{p(1-p)}{1 - e^{-\sigma\tau/2}} \ , \\
&\Delta \mathcal{E}_{\sigma}\geq-\epsilon \, N_\sigma(0)
+ \frac{2\lambda^2\epsilon }{|\mu|^2} \frac{p(1-p)}{1 + e^{-\sigma\tau/2}} \ .
\label{bound-}
\end{align}

For the short-time regime $n\tau\ll 1$ one gets for (\ref{energ-var-op-syst-0-n})
\begin{align}\label{energ-var-op-syst-0-short}
&\Delta \mathcal{E}_{\sigma}(t_n, t_0) = - n\tau \sigma N_\sigma(0)\\
&+n\tau p(1-p)\frac{2\lambda^2\epsilon }{|\mu|^2}(1-\cos\epsilon\tau)
\frac{\sigma}{1 - e^{-\sigma\tau}} \nonumber \\
&+n\tau p^{2}\frac{\lambda^2\sigma\epsilon}{|\mu|^2} + \mathcal{O}((n\tau)^2)
\nonumber \ ,
\end{align}
i.e. a linear asymptotic behaviour.

This case coincides with result for the ideal cavity (Theorem \ref{ThDEn-1,n1}) when the rate of environmental
pumping $\sigma =0$.

\subsection{Entropy production in perfect cavity}

In this section, we make contact with thermodynamics of
systems out of equilibrium, and in particular with the second law of thermodynamics, see \cite{Bruneau:2006}-\cite{Bruneau:2010}.

Let $\rho$ and $\rho_0$ be two normal states on algebra $\mathfrak{M}(\mathcal{H})$. We define the \textit{relative entropy}
${\rm{Ent}}(\rho|\rho_0)$ of the state $\rho$ with respect to $\rho_0$ by
\begin{equation}\label{Rel-Ent}
{\rm{Ent}}(\rho|\rho_0):= \Tr_{\mathcal{H}}(\rho \ln \rho - \rho\ln \rho_0) \geq 0 \ ,
\end{equation}
where non-negativity follows from the \textit{Jensen inequality}: $\Tr_{\mathcal{H}} (\rho \ln \mathcal{A})
\leq \ln \Tr_{\mathcal{H}} (\rho \mathcal{A})$, applied to observable $\mathcal{A}= \rho_0/\rho$.

In the case of the non-leaking cavity we have:
\begin{equation}\label{1DS-sigma=0}
\Delta S(t) := {\rm{Ent}}(\rho_{S}(t)|\rho_{S}(t=0)) = \Tr_{\cH_C \otimes \cH_A}\{ \rho_{S}(t)
(\ln\rho_{S}(t)-\ln\rho_C\otimes\rho_A)\} \ ,
\end{equation}
where dynamics is defined by (\ref{Sol-Liouv-Eq}). Suppose that all atoms of the beam are in the Gibbs state with
temperature $1/\beta$, which formally can be written as
\begin{equation}\label{Atom-Gibbs}
\rho_A (\beta) : = \bigotimes_{n\geq1} \rho_{A_n} (\beta)  \  , \  \ \rho_{A_n} (\beta):= \frac{e^{-\beta H_{A_n}}}{Z(\beta)} \ ,
\end{equation}
see (\ref{Ham-Model}). Since $\rho_{S}(t=0) = \rho_C\otimes\rho_A = (\rho_C\otimes \idty) (\idty \otimes\rho_A)$ and
$\Tr_{\cH_C \otimes \cH_A}\{\rho_{S}(t) \ln \rho_{S}(t)\} = \Tr_{\cH_C \otimes \cH_A}\{\rho_{S} (0)\ln \rho_{S}(0)\}$,
the relative entropy (\ref{1DS-sigma=0}) gets the form
\begin{eqnarray}\label{2DS-sigma=0}
\Delta S(t) :&=& \Tr_{\cH_C} \{[\rho_C-\rho_C^{(n)}]\ln \rho_{C}\} \\
&-& \beta \ \sum_{k=1}^{n} \Tr_{\cH_S}\{[\rho_{S}(0)-\rho_{S}(n\tau)] (\idty \otimes H_{A_k})\}  \nonumber \ ,
\end{eqnarray}
for $t = n \tau + \nu$, see (\ref{t}). Here $\rho_C^{(n)}$ is defined by (\ref{C-state-n}).
\begin{remark}\label{Entr-commut}
For any Hamiltonian $H_n$ that acts non-trivially on $\cH_C\otimes\cH_{A_n}$ and any Hamiltonian $H_{A_k}$ acting on $\cH_{A_k}$ we have  $[H_n, H_{A_k}] = 0$ for $n \neq k$. Note that by (\ref{commut-atoms}) for our model $[H_n, H_{A_k}] = 0$ for any $n,k$.
\end{remark}
Then by virtue of (\ref{C-state-n}) and(\ref{cL}) we have
\begin{eqnarray}\label{rho-S-nk}
&&\Tr_{\cH_S}\{\rho_{S}(n\tau) \ (\idty \otimes H_{A_k})\} \\
&& = \Tr_{\cH_S}\{e^{- i \tau H_n}...e^{-i \tau H_{k+1}} \rho_{S}(k\tau)e^{i \tau H_{k+1}} \ldots {e^{i \tau H_n}} \
(\idty \otimes H_{A_k})\} \nonumber \\
&& = \Tr_{\cH_S}\{\rho_{S}((k\tau) \ (\idty \otimes H_{A_k})\}  \nonumber \ .
\end{eqnarray}
By the same argument one obtains also
\begin{eqnarray}\label{rho-S-0k-1}
&&\Tr_{\cH_S}\{\rho_{S}(0) \ (\idty \otimes H_{A_k})\} \\
&& = \Tr_{\cH_S}\{e^{- i \tau H_{k-1}}...e^{-i \tau H_1} \rho_{S}(0)e^{i \tau H_1} \ldots e^{i \tau H_{k-1}} \
(\idty \otimes H_{A_k})\} \nonumber \\
&& = \Tr_{\cH_S}\{\rho_{S}((k-1)\tau) \ (\idty \otimes H_{A_k})\}  \nonumber \ .
\end{eqnarray}

In the case of our model (see Remark \ref{Entr-commut}) $H_{A_k}=e^{i\tau H_k}(H_{A_k})e^{-i\tau H_k}$. Therefore the last formula gets the form
\begin{eqnarray}\label{rho-S-0k-1-model}
&&\Tr_{\cH_S}\{\rho_{S}(0) \ (\idty \otimes H_{A_k})\}=\Tr_{\cH_S}\{\rho_{S}(k\tau) \ (\idty \otimes H_{A_k}).
\end{eqnarray}

Equations (\ref{rho-S-nk}) and (\ref{rho-S-0k-1-model}) shows that the second term in the entropy production (\ref{2DS-sigma=0}) vanishes.

If we suppose that the initial cavity state is Gibbs state for the temperature $1/\beta$ (\ref{Gibbs-photons}), then by
(\ref{mean-photon-number-n}) and (\ref{number of photons Gibbs}) one gets
\begin{eqnarray}\label{EntrProd-OurMod}
&&\Delta S(t)= \Tr_{\cH_C} \{[\rho_C-\rho_C^{(n)}]\ln \rho_{C}\} \\
&& =\Tr_{\cH_C} \{[\rho_C-\rho_C^{(n)}] (- \beta \, \epsilon \, b^{*}b )\} \nonumber \\
&&= \beta\epsilon(N(t)-N(0))\ , \nonumber
\end{eqnarray}
where the photon number $N(t)$ is defined in (\ref{number of photons Gibbs}).

If we denote by $\Delta \mathcal{E}^{C}(t) = \epsilon (N(t) - N(0))$ the energy variation of the thermal cavity defined by the
photon number variation, then (\ref{EntrProd-OurMod}) expresses the 2nd Law of Thermodynamics
\begin{equation}\label{2nd-Law}
\Delta S(t) = \beta \Delta \mathcal{E}^{C}(t) \ ,
\end{equation}
for the pumping of cavity.

\begin{remark}
In the general situation when $[H_n, H_{A_n}]\neq 0$ combining (\ref{2DS-sigma=0}) with (\ref{rho-S-nk}) and (\ref{rho-S-0k-1}) we get for the entropy production at the moment
$t = n \tau + \nu$
\begin{eqnarray}\label{DS-sigma=0-bis}
\Delta S(t) :&=& \Tr_{\cH_C} \{[\rho_C-\rho_C^{(n)}]\ln \rho_{C}\} \\
&+& \beta \ \sum_{k=1}^{n} \Tr_{\cH_S}\{[\rho_{S}(k \tau) - \rho_{S}((k-1)\tau)] \ (\idty \otimes H_{A_k})\} \nonumber \ .
\end{eqnarray}
The last term in (\ref{DS-sigma=0-bis}) can be rewritten into standard form \cite{Bruneau:2006}-\cite{Bruneau:2010}, if one uses the identities:
\begin{eqnarray*}\label{last-term1}
&&\Tr_{\cH_S}\{\rho_{S}(k\tau) \ (\idty \otimes H_{A_k})\} =
\Tr_{\cH_C \otimes \cH_A }\{e^{\tau L_k}\ldots e^{\tau L_1}(\rho_{C}\otimes \rho_{A})\ (\idty \otimes H_{A_k})\} \\
&&=\Tr_{\cH_C \otimes \cH_A }\{e^{- i \tau H_{k}}(e^{\tau L_{k-1}}\ldots e^{\tau L_1}[\rho_{C}\otimes \bigotimes_{n=1}^{k-1}\rho_{n}])
\otimes \rho_{k} e^{i \tau H_{k}} \ (\idty \otimes H_{A_k})\} [\idty \otimes \bigotimes_{m>k} \rho_{m}]\nonumber \\
&&=\Tr_{\cH_C \otimes \cH_{A_{k}}} \{(\rho_{C}^{(k-1)} \otimes \rho_{k}) e^{i \tau H_{k}} \ (\idty \otimes H_{A_k})e^{- i \tau H_{k}}\}
\nonumber  \ ,
\end{eqnarray*}
and
\begin{equation*}\label{last-term2}
\Tr_{\cH_S}\{\rho_{S}((k-1)\tau) \ (\idty \otimes H_{A_k})\} =
\Tr_{\cH_C \otimes \cH_{A_{k}}} \{\rho_{C}^{(k-1)} \otimes \rho_{k} H_{A_k}\} \ .
\end{equation*}
For $t = n \tau + \nu$ they yield the formula for the non-leaking entropy production
\begin{eqnarray}\label{DS-sigma=0-fin}
&&\Delta S(t)= \Tr_{\cH_C} \{[\rho_C-\rho_C^{(n)}]\ln \rho_{C}\} + \\
&&\beta \ \sum_{k=1}^{n} \Tr_{\cH_C \otimes \cH_{A_{k}}} \{(\rho_{C}^{(k-1)} \otimes \rho_{k})
[e^{i \tau H_{k}} (\idty \otimes H_{A_k})e^{- i \tau H_{k}} - \idty \otimes H_{A_k}]\}  \nonumber  \ .
\end{eqnarray}
\end{remark}

%
%
%

\end{document}